%% file: main.tex
\documentclass[sigconf, nonacm]{acmart}
\usepackage{amsmath,amsfonts,amsthm}
\usepackage{color}
\usepackage{xspace}
\usepackage[linesnumbered,ruled,vlined]{algorithm2e}
\usepackage{enumitem}
\usepackage{subfigure}
\usepackage{caption}
\usepackage{natbib}
\usepackage{booktabs}
\usepackage{amsmath}
\usepackage{array}
\usepackage{multirow}
\usepackage{tikz}
\usepackage{fancyhdr}

\usetikzlibrary{positioning,calc,backgrounds,fit,shapes.multipart}

\newcommand\vldbdoi{XX.XX/XXX.XX}
\newcommand\vldbpages{XXX-XXX}
\newcommand\vldbvolume{14}
\newcommand\vldbissue{1}
\newcommand\vldbyear{2020}
\newcommand\vldbauthors{\authors}
\newcommand\vldbtitle{\shorttitle} 
\newcommand\vldbavailabilityurl{URL_TO_YOUR_ARTIFACTS}
\newcommand\vldbpagestyle{plain}

\newcommand{\kewu}[1]{{\color{red}#1}\xspace}

\newcommand{\revision}[1]{{\color{black}#1}\xspace}

\newtheorem*{problem*}{Problem Definition}

\begin{document}

\title{Revisiting the Maximum Defective Clique Problem: Faster Branching and a Tighter Upper Bound}

\author{Kewu Yang}
\affiliation{%
  \institution{Harbin Institute of Technology, Shenzhen}
}
\email{220110422@stu.hit.edu.cn}

\author{Kaiqiang Yu}
\affiliation{%
  \institution{State Key Laboratory of Novel Software Technology, Nanjing University}
}
\authornote{Corresponding authors.}
\email{kaiqiang.yu@nju.edu.cn}

\author{Shengxin Liu}
\affiliation{%
  \institution{Harbin Institute of Technology, Shenzhen}
}
\authornotemark[1]
\email{sxliu@hit.edu.cn}

\author{Zhaoquan Gu}
\affiliation{%
  \institution{Harbin Institute of Technology, Shenzhen}
  }
\affiliation{%
  \institution{PengCheng Laboratory}
  }
\authornotemark[1]
\email{guzhaoquan@hit.edu.cn}



\input{0-Abstract}

\maketitle

\pagestyle{\vldbpagestyle}
\begingroup\small\noindent\raggedright\textbf{PVLDB Reference Format:}\\
\vldbauthors. \vldbtitle. PVLDB, \vldbvolume(\vldbissue): \vldbpages, \vldbyear.\\
\href{https://doi.org/\vldbdoi}{doi:\vldbdoi}
\endgroup
\begingroup
\renewcommand\thefootnote{}\footnote{\noindent
This work is licensed under the Creative Commons BY-NC-ND 4.0 International License. Visit \url{https://creativecommons.org/licenses/by-nc-nd/4.0/} to view a copy of this license. For any use beyond those covered by this license, obtain permission by emailing \href{mailto:info@vldb.org}{info@vldb.org}. Copyright is held by the owner/author(s). Publication rights licensed to the VLDB Endowment. \\
\raggedright Proceedings of the VLDB Endowment, Vol. \vldbvolume, No. \vldbissue\ %
ISSN 2150-8097. \\
\href{https://doi.org/\vldbdoi}{doi:\vldbdoi} \\
}\addtocounter{footnote}{-1}\endgroup

\ifdefempty{\vldbavailabilityurl}{}{
\vspace{.3cm}
\begingroup\small\noindent\raggedright\textbf{PVLDB Artifact Availability:}\\
The source code, data, and/or other artifacts have been made available at \url{https://github.com/Thaumaturge2020/BBRes/}.
\endgroup
}

\input{1-Intro}
\input{2-Definition}

\input{3-SOTA}

\input{4-BBRes}
\input{5-UB}
\input{6-Experiments}

\input{7-Related}
\input{8-Conclusion}


\clearpage

\bibliographystyle{ACM-Reference-Format}
\bibliography{reference}
\balance
\input{9-Appendix}


\end{document}

%% file: 0-Abstract.tex
\begin{abstract}
The $k$-defective clique model relaxes the strict completeness constraint of the traditional clique by allowing up to $k$ missing edges, providing a robust formulation for detecting cohesive structures in noisy graphs. Consequently, the maximum $k$-defective clique problem has attracted significant attention. State-of-the-art exact algorithms predominantly adopt the branch-and-bound framework, which recursively partitions the current problem instance (or branch) into two sub-problems via a branching procedure, until each sub-problem becomes trivially solvable. However, this strategy often leads to excessive branching by overlooking intermediate sub-problems that are non-trivial yet efficiently solvable. While recent studies have attempted to refine branching procedures, they fail to address this structural redundancy. To address this, we propose \texttt{BBRes}, a framework that incorporates a novel early termination strategy into the recursive branching process. By employing a specialized polynomial-time solver to identify and resolve tractable sub-instances, \texttt{BBRes} effectively avoids redundant branching steps. Additionally, we design a tailored branching strategy that synergizes with this termination mechanism. As a result, \texttt{BBRes} achieves an improved theoretical worst-case time complexity. To enhance practical performance, we propose a tighter upper bound based on a novel double graph coloring method integrated with max-flow techniques, which is orthogonal to the branching framework. Extensive experiments show that \texttt{BBRes} achieves \revision{at least 2X speedup over state-of-the-art methods on a substantial fraction of the datasets.}
\end{abstract}

%% file: 1-Intro.tex
\section{Introduction}

Graphs are widely used to model relationships between entities in various applications, such as social media, biological science, and e-commerce. Extracting cohesive/dense subgraphs from large real-world graphs is a fundamental problem in graph analytics. Such cohesive subgraphs often carry interesting information, which can facilitate various tasks across different domains. Examples include finding communities in social networks~\cite{bedi2016community,FHQ+20survey}, detecting anomalies in financial networks or social media~\cite{ahmed2016survey,berry2004emergent,leung2005unsupervised,yu2021graph}, and discovering biologically relevant functional groups in biological networks~\cite{wang2010recent}. 

One notable cohesive subgraph model is the clique, which requires that every pair of vertices is connected by an edge~\cite{eppstein2013listing,cheng2011finding,tomita2010simple,tomita2017efficient,san2016new,pattabiraman2015fast,pardalos1994maximum,carraghan1990exact}. However, the strict completeness requirement of cliques is often too restrictive for real-world applications due to data quality issues caused by noise and missing values. To solve this issue, recent studies have relaxed the clique model and proposed various clique relaxation models, including quasi-clique~\cite{khalil2022parallel,pei2005mining,yu2023fast,zeng2006coherent,xia2025iterqc}, $k$-plex~\cite{seidman1978kplex,chang2022efficient,wang2023fast,xiao2017fast,zhou2021improving,gao2024maximum}, $k$-club~\cite{bourjolly2002exact}, and $k$-defective clique~\cite{yu2006predicting,chen2021computing,gao2022exact,jin2024kdclub,Chang24kDC-2,dai2024theoretically,Luo24defective,chang2023kdefective}. In this paper, we focus on the $k$-defective clique model, which relaxes the clique model by allowing up to $k$ missing edges, where $k$ is a positive integer. Conceptually, a $k$-defective clique is a subgraph induced by a set of vertices $S$ that contains at least $\binom{|S|}{2}-k$ edges. We study the problem of finding the maximum $k$-defective clique, i.e., the one with the largest number of vertices, which has been used for various applications such as interaction prediction in biological networks~\cite{yu2006predicting}, cluster detection~\cite{dai2023maximal}, and social network analysis~\cite{gschwind2021branch}.

\smallskip
\noindent\textbf{State-of-the-art methods}. Recently, a number of exact algorithms have been proposed in the literature for finding the maximum $k$-defective clique~\cite{chen2021computing,gao2022exact,jin2024kdclub,Chang24kDC-2,dai2024theoretically,Luo24defective,chang2023kdefective,jang2025efficient}. The majority of them follow a similar \emph{branch-and-bound (BB)} framework~\cite{Chang24kDC-2,dai2024theoretically,chang2023kdefective}. Specifically, this framework recursively partitions the problem into two sub-problems on \emph{smaller} subgraphs via a \emph{branching} method, continuing until each sub-problem can be solved trivially (i.e., the corresponding subgraph is already a $k$-defective clique). Each sub-problem corresponds to a branch. To improve efficiency, reduction rules are developed to prune unpromising branches that cannot contain the largest $k$-defective clique. Among these reductions, upper-bound-based pruning is often the most powerful in practice. The rationale is that the algorithm computes an upper bound on the size of any $k$-defective clique derivable from the current branch, and it safely prunes the branch if this upper bound is not larger than the size of the largest $k$-defective clique found so far. Therefore, tight upper bounds are preferred to enhance practical performance. In terms of theoretical complexity, we note that recent works focus on sharpening the branching method to improve the worst-case time complexity. The best known complexity is $O^*(\gamma_k^{n})$~\cite{dai2024theoretically}, where $O^*$ suppresses polynomial factors, $n$ is the number of vertices, and $\gamma_k$ is the largest real root of the equation $x^{k+3}-2x^{k+2}+x^2-x+1 = 0$. 

While some recent approaches, such as \texttt{DnBK}~\cite{Luo24defective} and \texttt{WODC}~\cite{jang2025efficient}, attempt to depart from the standard BB framework to reduce the exponential base of the time complexity, these theoretical improvements come at a prohibitive cost. Specifically, they introduce a substantial overhead factor scaled by $(\delta\Delta)^{\Theta(k)}$, where $\delta$ and $\Delta$ denote the degeneracy and maximum degree of the graph, respectively. In real-world graphs where $\Delta$ is typically large, this overhead explodes as $k$ increases, rendering such methods computationally impractical for difficult instances. Consequently, existing state-of-the-art algorithms, whether adhering to the standard BB framework or alternative paradigms, still suffer from significant efficiency bottlenecks in both theory and practice.

\smallskip
\noindent\textbf{Our method}. In this paper, instead of solely refining the branching strategies as existing methods do, we focus on \emph{both the termination strategy and the branching strategy}. Specifically, we observe that existing methods terminate the recursive branching procedure of generating new branches only when the current sub-problem can be solved \emph{trivially}, which incurs an excessive number of recursive calls to reach such a trivial branch. As a result, they generate an excessive number of branches. Motivated by this, we propose to terminate the recursive branching procedure once the current sub-problem can be solved efficiently (e.g., in polynomial time) via a non-recursive solver, thereby generating fewer branches. We call this strategy the \emph{early termination strategy}. To implement this strategy, we first formulate \emph{non-trivial} input-restricted problems (of finding the largest $k$-defective clique) called \texttt{MissingTwoDeg}, which corresponds to a collection of non-trivial branches. Then, we develop a greedy method called \texttt{IRSolver} to solve \texttt{MissingTwoDeg} instances in polynomial time. Based on \texttt{IRSolver}, our early termination strategy halts the recursive branching and invokes \texttt{IRSolver} when the current branch corresponds to an input-restricted problem. Furthermore, to fully exploit the potential of our early termination strategy, we design a new branching strategy called \textbf{BS-three}. Compared to existing strategies, \textbf{BS-three} better complements the early termination strategy as it tends to generate fewer branches. With the early termination strategy and new branching strategy, our branch-and-bound algorithm called \texttt{BBRes} achieves the worst-case time complexity of $O^*(\lambda_k^n)$, where $\lambda_k$ is the largest real root of the equation $x^{k+4}-2x^{k+3}+x^{3}-x+1=0$. For example, $\lambda_k=1.381$, 1.705, 1.867 when $k=1$, 2, 3, respectively. We remark that, compared to the state-of-the-art method \texttt{MDC}~\cite{dai2024theoretically}, our \texttt{BBRes} further improves the worst-case time complexity from $O^*(\gamma_k^n)$ to $O^*(\lambda_k^n)$, where $\lambda_k$ is \emph{strictly smaller than} $\gamma_k$ for $k\geq 1$.

In addition, to further boost practical performance, we propose a new upper bound on the size of the $k$-defective clique to be found in a branch, by leveraging double graph coloring and max-flow techniques. Specifically, we perform graph coloring \emph{twice}, assigning each vertex a pair of colors. Then, based on the double-coloring information, we construct a constrained max-flow problem to compute the upper bound. Compared to existing color-based upper bounds~\cite{Chang24kDC-2,dai2024theoretically,chen2021computing} that also utilize graph coloring, our upper bound \textbf{UB-Double} is \emph{practically tighter}, as it leverages more topological information by applying graph coloring twice and incorporating the max-flow technique. We remark that this upper bound is orthogonal to the branch-and-bound framework and can also be applied to existing methods~\cite{chen2021computing,gao2022exact,jin2024kdclub,Chang24kDC-2,dai2024theoretically,chang2023kdefective}.  

\smallskip
\noindent\textbf{Contributions}. We summarize our main contributions as follows. 
\begin{itemize}[leftmargin=*, topsep=0pt]
    \item We propose a new branch-and-bound algorithm called \texttt{BBRes} which is based on our newly developed early termination strategy and branching strategy. \texttt{BBRes} has the worst-case time complexity of $O^*(\lambda_k^n)$, improving upon the state-of-the-art method~\cite{dai2024theoretically} since $\lambda_k$ is strictly smaller than $\gamma_k$ for $k\geq 1$. (Section~\ref{sec:BBRes})
    \item We further propose a new upper bound of the size of $k$-defective cliques to be found in a branch based on the double graph coloring technique and the max-flow technique. Our upper bound is tighter than the existing ones in practice. (Section~\ref{sec:ub})
    \item Finally, we conduct extensive experiments on a collection of 139 real-world graphs. The results show that our \texttt{BBRes} outperforms other baselines, including  \texttt{kDC2}~\cite{Chang24kDC-2}, \texttt{MDC}~\cite{dai2024theoretically}, \texttt{DnBK}~\cite{Luo24defective}, and \texttt{WODC}~\cite{jang2025efficient}, by solving more problem instances. (Section~\ref{sec:exp}) 
\end{itemize}
For the rest of the paper, we define the problem in Section~\ref{sec:pre}, present state-of-the-art methods in Section~\ref{sec:baseline}, review related work in Section~\ref{sec:related}, and conclude the paper in Section~\ref{sec:conclu}.

%% file: 2-Definition.tex
\section{Preliminaries}
\label{sec:pre}
We consider an unweighted and undirected graph $G = (V, E)$, where $V$ is the vertex set and $E$ is the edge set.
Let $n=|V|$ and $m=|E|$ denote the number of vertices and edges, respectively.
We define $\overline{E}$ as the set of edges that are missing in $G$ (referred to as \textbf{non-edges}), i.e., $\overline{E}=\{(u,v)\in V\times V\mid u\neq v \text{ and } (u,v)\notin E\}$. 

Given a vertex $v$ in $V$, let $N_G(v)$ (resp. $\overline{N}_G(v)$) be the set of neighbors (resp. non-neighbors) of $v$ in $G$; formally, $N_G(v) = \{ u \in V \mid (u, v) \in E \}$ (resp. $\overline{N}_G(v) = \{ u \in V \mid (u, v) \in \overline{E}\}$). Accordingly, we define the degree of $v$ in $G$ as $d_G(v) = |N_G(v)|$, and its non-degree as $\overline{d}_G(v) = |\overline{N}_G(v)|$. We note that a vertex is neither a neighbor nor a non-neighbor of itself. Given a vertex subset $S\subseteq V$, we use $G[S]$ to denote the subgraph of $G$ induced by $S$, i.e., $G[S]$ consists of the set of vertices $S$ and the set of edges $\{(u,v)\in E\mid u,v\in S\}$. All subgraphs considered in this paper are induced subgraphs. Given a subgraph $g$ of $G$, we use $V(g)$, $E(g)$, and $\overline{E}(g)$ to denote its sets of vertices, edges, and non-edges, respectively. 

\noindent\textbf{Abbreviations}. For simplicity, we omit the subscript $G$ from the notation when the context is clear. Furthermore, for $S\subseteq V$, we abbreviate $N_{G[S]}(\cdot)$ as $N_{S}(\cdot)$, $d_{G[S]}(\cdot)$ as $d_{S}(\cdot)$, and $E(G[S])$ as $E(S)$. Similarly, we abbreviate $\overline{N}_{G[S]}(\cdot)$, $\overline{d}_{G[S]}(\cdot)$, and $\overline{E}(G[S])$ as $\overline{N}_{S}(\cdot)$, $\overline{d}_{S}(\cdot)$, and $\overline{E}(S)$, respectively.

In this paper, we focus on the concept of $k$-defective clique.

\begin{definition}[$k$-Defective Clique~\cite{yu2006predicting}]
A graph $g$ is said to be a \emph{$k$-defective clique} if it contains at most $k$ non-edges, i.e., $|\overline{E}(g)| \leq k$ or equivalently, $|E(g)| \geq \frac{|V(g)|(|V(g)|-1)}{2} - k$.
\end{definition}

Clearly, a 0-defective clique is simply a clique, in which every pair of vertices is adjacent. In addition, we note that the $k$-defective clique satisfies the \emph{hereditary} property, i.e., any subgraph of a $k$-defective clique is also a $k$-defective clique~\cite{pattillo2013clique}.
We now formally state the problem studied in this paper.

\begin{problem*}[Maximum $k$-Defective Clique~\cite{chen2021computing,gao2022exact,Luo24defective,chang2023kdefective,Chang24kDC-2,dai2024theoretically}]
Given a graph $G = (V,E)$ and a positive integer $k$, the \emph{maximum $k$-defective clique problem} aims to find the maximum $k$-defective clique in $G$, i.e., the $k$-defective clique with the largest number of vertices.
\end{problem*}

We note that the maximum $k$-defective clique in $G$ may not be unique, and the maximum $k$-defective clique problem is known to be NP-hard~\cite{LEWIS1980219,TBBB13}.
\revision{Moreover, although a larger $k$ relaxes the feasibility constraint, it also enlarges the family of feasible solutions. In particular, the number of maximal $k$-defective cliques can grow exponentially with $k$~\cite{chang2023kdefective,Chang24kDC-2,dai2024theoretically}, resulting in a larger search space and making the problem harder in practice.}

%% file: 3-SOTA.tex
\section{State-of-the-Art Algorithms}
\label{sec:baseline}
The state-of-the-art algorithms for the maximum $k$-defective clique predominantly adopt the \emph{branch-and-bound} (BB) framework~\cite{chang2023kdefective,Chang24kDC-2,dai2024theoretically}.
The core idea is to recursively partition the current problem instance (i.e., search space), which aims to find the largest $k$-defective clique, into two smaller sub-problems via \emph{branching} until each of them can be solved trivially.
Specifically, a problem instance (or branch) is represented as a triple $(g,S,C)$, where:
\begin{itemize}[leftmargin=*, topsep=0pt]
    \item Graph $g$ is a subgraph of the input graph $G$ induced by the vertex set $S\cup C$, i.e., $g=G[S\cup C]$;
    \item Partial set $S$ is a set of vertices that induces a $k$-defective clique and must be contained in the largest $k$-defective clique found within this branch;
    \item Candidate set $C$ is a set of vertices that will be considered to be included in $S$.
\end{itemize}
Solving a branch $(g,S,C)$ means finding in this branch the largest $k$-defective clique $g^*$ that (1) includes all vertices in $S$ and (2) is a subgraph of $g$, i.e., $S\subseteq V(g^*)$ and $V(g^*)\subseteq C\cup S$. Clearly, solving the branch $(G,\emptyset,V)$ finds the maximum $k$-defective clique in $G$.

To solve a branch $(g,S,C)$, the framework recursively partitions the branch into two sub-branches via branching. In particular, it selects a vertex $v^*$ called \emph{pivot} from the candidate set $C$, and generates two sub-branches. The first one $(g_1,S\cup\{v_p\},C\setminus\{v_p\})$ removes the pivot $v_p$ from $C$ to $S$ (which aims to explore the largest $k$-defective clique that includes $v_p$); the second one $(g_2,S,C\setminus\{v_p\})$ discards the pivot $v_p$ from $C$ (which aims to explore the largest $k$-defective clique that excludes $v_p$). Here, $g_1$ and $g_2$ are the subgraphs induced by the updated partial and candidate sets in each sub-branch. By recursively solving both sub-branches, the framework ensures that all possibilities are explored and thus solves the original branch.

\begin{algorithm}[t]
    \caption{The state-of-the-art branch-and-bound (BB) framework~\cite{chang2023kdefective,Chang24kDC-2,dai2024theoretically}}
    \label{alg:basic}
    \small
    \KwIn{A graph $G=(V,E)$ and a positive integer $k$}
    \KwOut{The maximum $k$-defective clique $g^*$}
    \tcc{Stage-I: With the diameter-two property}
    Let $g^*\gets \emptyset$ be the largest $k$-defective clique seen so far\;
    Let $V=\{v_1, v_2, \ldots, v_n\}$ be a degeneracy ordering of vertices in $G$\;
    \ForEach{$v_i \in \{v_1,v_2,\cdots,v_n\}$}{
        $G_{v_i} \gets$ the subgraph of $G$ induced by $N^{\leq 2}(v_i) \cap \{v_i,v_{i+1},\dots,v_n\}$ \;
        \texttt{BB\_Rec}$(G_{v_i},\{v_i\},V(G_{v_i})\setminus\{v_i\})$\;
    }
    \tcc{Stage-II: Without the diameter-two property}
    \lIf{$|V(g^*)| < k + 1$}{\texttt{BB\_Rec}$(G,\emptyset,V)$}
    \textbf{return} $g^*$\;
    
    \SetKwBlock{Enum}{Procedure \texttt{BB\_Rec}$(g,S,C)$}{}
    \smallskip
    \Enum{
        \tcc{Reduction}
        $UB \gets$ an upper bound of the branch\;
        \lIf{$UB \leq |V(g^*)|$}{\textbf{return}}
        Refine $C$ (and $g$) by applying reduction rules\;
        \tcc{Termination}
        \If{$g$ is a $k$-defective clique}{
            \lIf{$|V(g)|>|V(g^*)|$}{$g^*\leftarrow g$}
            \textbf{return}\;
        }
        \tcc{Branching}
        $v_p\gets$ a pivot selected from $C$\; Construct $g_1$ and $g_2$ based on $v_p$, $S$, and $C$\;
        \texttt{BB\_Rec}$(g_1,S\cup\{v_p\},C\setminus\{v_p\})$\;
        \texttt{BB\_Rec}$(g_2,S,C\setminus\{v_p\})$\;
    }
\end{algorithm}

\noindent\textbf{Summary}. We summarize the BB framework in Algorithm~\ref{alg:basic}. Specifically, the algorithm maintains the largest $k$-defective clique $g^*$ seen so far during the recursion (Line 1). To boost the practical performance, it utilizes \emph{\textbf{the diameter-two property}} (\cite{chen2021computing}) of large $k$-defective cliques (i.e., any $k$-defective clique with at least $k+2$ vertices has diameter at most 2) and involves the following stages. 
\begin{itemize}[leftmargin=*,topsep=0pt]
    \item \textbf{Stage-I} adopts the diameter-two property for pruning by assuming that the largest $k$-defective clique is of size at least $k+2$ (Lines 2-5). Specifically, it first computes a degeneracy ordering of the vertices in $G$ (Line 2), which can be done efficiently in $O(m)$ by the peeling algorithm~\cite{BVZ03m}. Then, it divides the problem of finding the maximum $k$-defective clique into $n$ sub-problems (Lines 3-5). The $i$-th one aims to find the maximum $k$-defective clique that (1) includes $v_i$ and (2) is a subgraph of $G[\{v_i,v_{i+1},...,v_n\}]$ by invoking the BB procedure with the branch $(G_{v_i},\{v_i\},V(G_{v_i})\setminus\{v_i\})$ (Lines 3-5). In particular, $G_{v_i}$ is the subgraph of $G$ induced by $v_i$'s two-hop neighbors in $\{v_i,v_{i+1},...,v_n\}$, i.e., $N^{\leq 2}(v_i) \cap \{v_i,v_{i+1},\dots,v_n\}$ (note that a $k$-defective clique with at least $k+2$ vertices has diameter of at most 2 and thus the largest one containing $v_i$ is a subset of $N^{\leq 2}(v_i)$). Clearly, if the found largest $k$-defective clique $g^*$ is of at least $k+1$ vertices, $g^*$ is guaranteed to be a maximum $k$-defective clique of $G$.
    \item \textbf{Stage-II} continues the search for the maximum $k$-defective clique when Stage-I does not succeed, i.e., when $g^*$ obtained in Stage-I has fewer than $k+1$ vertices (Line 6). To this end, it invokes the BB procedure with the branch $(G,\emptyset,V)$.
\end{itemize}
In addition, the BB procedure called \texttt{BB\_Rec} is summarized in Lines 8-18. 
We note that the recursive procedure terminates when $g$ becomes a $k$-defective clique since $g$ is the largest $k$-defective clique in the branch (Lines 12-14).
To reduce the search space, the framework employs two techniques as follows.
\begin{itemize}[leftmargin=*, topsep=0pt]
    \item \textbf{Upper-bound-based reductions} (Lines 9-10). It first computes an upper bound of the size of the largest $k$-defective clique in the branch. Then, the branch can be terminated if the upper bound is no larger than the largest $k$-defective clique seen so far. 
    \item \textbf{Pivot-based branching} (Lines 15-18). The number of branches generated by the pivot-based branching depends on the pivot selection strategy (Line 15). We note that recent studies have explored various strategies for pivot selection to reduce the number of branches~\cite{chang2023kdefective,Chang24kDC-2,dai2024theoretically}.
\end{itemize}
For brevity, we do not review the detailed implementation of these techniques in state-of-the-art algorithms.

\smallskip
\noindent\textbf{Time complexity}. Following Algorithm~\ref{alg:basic}, the state-of-the-art studies have focused on improving the worst-case time complexity by sharpening the pivot selection strategy~\cite{chang2023kdefective,Chang24kDC-2,dai2024theoretically}. As a result, the latest algorithm has the time complexity of $O^*(\gamma_k^n)$~\cite{dai2024theoretically}, where $O^*$ ignores the polynomial factors and $\gamma_i$ is the largest real root of the equation $x^{i+3} - 2x^{i+2} + x^2 - x + 1 = 0$. 

\smallskip
\noindent\textbf{Remark}. It is worth noting that two recent algorithms, \texttt{DnBK}~\cite{Luo24defective} and \texttt{WODC}~\cite{jang2025efficient}, employ frameworks distinct from the above BB framework to successfully reduce the exponential base in time complexity. However, this improvement comes at the cost of introducing a substantial overhead factor scaled by $(\delta\Delta)^{\Theta(k)}$, where $\delta$ is the degeneracy of the input graph $G$ and $\Delta$ is the maximum degree of $G$. This factor becomes prohibitively large when $k$ is not a constant, severely limiting scalability. As demonstrated in our experiments (Section~\ref{sec:exp}), both methods are less competitive compared to our proposed algorithm. We defer the detailed complexity analysis of both algorithms to the related work section (Section~\ref{sec:related}).

%% file: 4-BBRes.tex
\section{A new Branch-and-Bound Framework: \texttt{BBRes}}
\label{sec:BBRes}

\subsection{\texttt{BBRes}: Motivation and Overview}
\label{subsec:BBRes-overview}
\noindent\textbf{Motivation}.
The existing BB framework recursively partitions the current problem instance $(g,S,C)$ into two sub-problems via branching until each of them can be solved \emph{trivially}, i.e., $g$ becomes a $k$-defective clique. As a result, the branching procedure, which is biased towards generating trivial instances, produces an exponentially large number of branches in the worst case (e.g., $O(\gamma_k^n)$~\cite{dai2024theoretically}), thereby dominating the time complexity. To reduce the number of branches, we propose to \emph{guide the branching procedure towards non-trivial yet solvable problem instances}. The rationale is to terminate the recursive branching procedure early whenever the current problem instance can be solved efficiently (e.g., in polynomial time) by a dedicated solver, thus generating fewer branches. 

\smallskip
\noindent\textbf{Overview}. Motivated by the above, we develop a new BB framework called \texttt{BBRes}, as summarized in Algorithm~\ref{alg:BBRes}. \texttt{BBRes} differs from Algorithm~\ref{alg:basic} mainly in the \emph{early termination strategy} (Lines 12-15) and the \emph{branching strategy}
(Lines 16-18). Specifically, we first formulate an input-restricted problem instance (i.e., a branch $(g,S,C)$ satisfying certain constraints), which is \emph{non-trivial} (i.e., $g$ can be a non-$k$-defective clique). Note that such input-restricted instances will be solved recursively via branching in Algorithm~\ref{alg:basic}. However, we observe that these instances can be solved efficiently in polynomial time by a non-recursive method called \texttt{IRSolver} in Section~\ref{subsec:BBRes-Termination}. Based on this observation, our \texttt{BBRes} can terminate the recursive branching procedure once it reaches an input-restricted instance, i.e., satisfying the termination conditions, and uses \texttt{IRSolver} to solve the branch. Furthermore, we observe that the existing branching strategies are not suitable for our \texttt{BBRes} since they are designed based on the framework of recursively branching to trivial branches. To boost the performance of \texttt{BBRes}, we further propose a new branching strategy in Section~\ref{subsec:BBRes-Branching}. With this newly-designed branching strategy, \texttt{BBRes} achieves better worst-case time complexity. In addition, we remark that we propose a novel upper bound (Lines 9-11) based on graph coloring and max-flow techniques (details will be discussed in Section~\ref{sec:ub}), which is orthogonal to the framework and further enhances practical performance. We remark that our \texttt{BBRes} also employs those existing techniques, including reductions and pre-processing techniques for reducing the input graph~\cite{dai2024theoretically,Chang24kDC-2,jang2025efficient}, that are orthogonal to the framework.

\begin{algorithm}[t]
    \caption{Our framework \texttt{BBRes}}
    \label{alg:BBRes}
    \small
    \KwIn{A graph $G=(V,E)$ and a positive integer $k$}
    \KwOut{The maximum $k$-defective clique $g^*$}
    \tcc{Stage-I: With the diameter-two property}
    Let $g^*\gets \emptyset$ be the largest $k$-defective clique seen so far\;
    Let $V=\{v_1, v_2, \ldots, v_n\}$ be a degeneracy ordering of vertices in $G$\;
    \ForEach{$v_i \in \{v_1,v_2,\cdots,v_n\}$}{
        $G_{v_i} \gets$ the subgraph of $G$ induced by $N^{\leq 2}(v_i) \cap \{v_i,v_{i+1},\dots,v_n\}$ \;
        \texttt{BBRes\_Rec}$(G_{v_i},\{v_i\},V(G_{v_i})\setminus\{v_i\})$\;
    }
    \tcc{Stage-II: Without the diameter-two property}
    \lIf{$|V(g^*)| < k + 1$}{\texttt{BBRes\_Rec}$(G,\emptyset,V)$}
    \textbf{return} $g^*$\;
    
    \SetKwBlock{Enum}{Procedure \texttt{BBRes\_Rec}$(g,S,C)$}{}
    \smallskip
    \Enum{
        \tcc{Reductions (Section~\ref{sec:ub})}
        $UB \gets$ an upper bound of the branch\;
        \lIf{$UB \leq |V(g^*)|$}{\textbf{return}}
        Refining $C$ (and $g$) by applying reduction rules\;
        \tcc{Early Termination Strategy (Section~\ref{subsec:BBRes-Termination})}
        \If{$g$ satisfies the early termination conditions, i.e., \textbf{Condition 1} or \textbf{Condition 2} in Section~\ref{subsec:BBRes-Termination}}{
            $g_{opt}\gets$\texttt{IRSolver}$(g,S,C)$\;
            \lIf{$|V(g_{opt})|>|V(g^*)|$}{$g^*\leftarrow g_{opt}$}
            \textbf{return}\;
        }
        \tcc{Branching (Section~\ref{subsec:BBRes-Branching})}
        $v_p\gets$ a pivot selected from $C$ based on our new strategy\;
        \texttt{BBRes\_Rec}$(g_1,S\cup\{v_p\},C\setminus\{v_p\})$\;
        \texttt{BBRes\_Rec}$(g_2,S,C\setminus\{v_p\})$\;
    }
\end{algorithm}

\subsection{\texttt{BBRes}: Early Termination Strategy}
\label{subsec:BBRes-Termination}
Consider a problem instance $(g,S,C)$. Note that a problem instance $(g,S,C)$ is said to be \emph{trivial} if $g$ is a $k$-defective clique. 
\revision{Our early termination strategy is motivated by the observation that certain non-trivial instances admit a particularly simple structure. Specifically, if every vertex in the candidate set $C$ has at most two non-neighbors within $C$, then in the complement graph of $G[C]$, every vertex has degree at most 2. Hence, each connected component of the complement graph is a path, a cycle, or an isolated vertex. This structural property enables an efficient algorithm for the remaining subproblem.}
%
Thus, we define the input-restricted problem instances, namely MissingTwoDeg, as follows.

\begin{definition}[MissingTwoDeg problem instance]
\label{def:IR}
 A problem instance $(g,S,C)$ is said to be a MissingTwoDeg problem if it satisfies one of the following conditions:
 \begin{itemize}[leftmargin=*, topsep=0pt]
     \item \textbf{Condition 1}. $g$ is a $k$-defective clique, i.e., $|\overline{E}(g)|\leq k$;
     \item \textbf{Condition 2}. $g$ is not a $k$-defective clique, and for any vertex $v$ in $C$, $v$ has at most 2 non-neighbors in $C$, i.e., $|\overline{E}(g)|>k$ and $\forall v\in C$, $\overline{d}_C(v)\leq 2$.
 \end{itemize}
\end{definition}
\revision{Condition 1 corresponds to the trivial termination case, whereas instances satisfying Condition 2 are non-trivial but remain tractable due to the special structure of the complement graph induced by $C$.}
Thus, terminating the branching procedure once a MissingTwoDeg instance is reached can reduce the number of generated branches while still allowing the remaining subproblem to be solved efficiently. One remaining question is how to solve a MissingTwoDeg problem efficiently in polynomial time. To this end, we introduce a greedy method called \texttt{IRSolver}.

\smallskip
\noindent\textbf{Overview of \texttt{IRSolver}}. If $g$ is a $k$-defective clique, we can solve the problem by returning $g$ directly. Otherwise, \texttt{IRSolver} solves the problem in a greedy manner, which runs in multiple rounds, as summarized in Algorithm~\ref{alg:IRSolver}. Specifically, it maintains two sets $S_{opt}$ (initially as $S$) and $C_{temp}$ (initially as $C$). The key idea is to iteratively select a vertex $v^{*}$ from $C_{temp}$ and move it from $C_{temp}$ to $S_{opt}$ while maintaining that $G[S_{opt}]$ is a $k$-defective clique until this is no longer possible (Lines 3-10, details will be discussed later). Finally, it returns $G[S_{opt}]$ as the solution (Line 11). We then elaborate on the details regarding the greedy strategy of selecting $v^*$, the correctness, and the time complexity analysis.

\begin{algorithm}[t]
    \caption{Our greedy method \texttt{IRSolver}}
    \label{alg:IRSolver}
    \small
    \KwIn{A MissingTwoDeg problem instance $(g,S,C)$ and a positive integer $k$}
    \KwOut{The maximum $k$-defective clique $g_{opt}$ that is a subgraph of $g$ and contains $S$}
    \lIf{$g$ is a $k$-defective clique}{\Return $g$ as $g_{opt}$ directly}
    $S_{opt}\gets S$, $C_{temp}\gets C$\;
    \While{true}{
        $\Gamma_{min}\gets \{v\in C_{temp}\mid \forall w\in C_{temp}, \overline{d}_{S_{opt}}(v)\leq \overline{d}_{S_{opt}}(w)\}$\;
        \uIf{\textbf{Case 1}: $\min_{v\in \Gamma_{min}}\overline{d}_{C_{temp}}(v)\leq 1$}{
            $v^*\gets \arg\min_{v\in \Gamma_{min}}\overline{d}_{C_{temp}}(v)$\;
        }\ElseIf{\textbf{Case 2}: $\min_{v\in \Gamma_{min}}\overline{d}_{C_{temp}}(v)=2$}{
            $v^*\gets \arg\min_{v\in \Gamma_{min}}\overline{d}_{\Gamma_{min}}(v)$\;
        }
        \lIf{$G[S_{opt}\cup\{v^*\}]$ is not a $k$-defective clique}{\textbf{break}}
        $C_{temp}\gets C_{temp}\setminus\{v^*\}$; $S_{opt}\gets S_{opt}\cup\{v^*\}$\;
    }
    \Return $G[S_{opt}]$ as the solution $g_{opt}$;
\end{algorithm}

\smallskip
\noindent\textbf{Greedy strategy of \texttt{IRSolver}}. At each round (Lines 3-10), \texttt{IRSolver} selects a vertex $v^*$ from $C_{temp}$ based on the following greedy strategy. Specifically, let $\Gamma_{min}$ be the set of vertices each of which has the smallest number of non-neighbors in $S_{opt}$, formally,
\begin{equation}
    \Gamma_{min} = \{v\in C_{temp} \mid \forall w\in C_{temp}, \overline{d}_{S_{opt}}(v)\leq \overline{d}_{S_{opt}}(w)\}.
\end{equation}
Then, we select vertex $v^*$ from $\Gamma_{min}$. In general, there are two cases.
\begin{itemize}[leftmargin=*, topsep=0pt]
    \item \textbf{Case 1:} $\min_{v\in \Gamma_{min}}\overline{d}_{C_{temp}}(v)\leq 1$, i.e., there exists a vertex $v$ in $\Gamma_{min}$ that has the number of non-neighbors in $C_{temp}$ smaller than 2. In this case, we select from $\Gamma_{min}$ the vertex with the smallest number of non-neighbors in $C_{temp}$ as $v^*$, formally,
    \begin{equation}
    \label{Eq:greedy-case1}
        \text{Greedy strategy 1: } v^*\gets \arg\min_{v\in \Gamma_{min}}\overline{d}_{C_{temp}}(v).
    \end{equation}
    
    \item \textbf{Case 2:} $\min_{v\in \Gamma_{min}}\overline{d}_{C_{temp}}(v) = 2$, i.e., each vertex in $\Gamma_{min}$ has two non-neighbors in $C_{temp}$. Then, we select from $\Gamma_{min}$ the vertex with the smallest number of non-neighbors in $\Gamma_{min}$, formally,
    \begin{equation}
    \label{eq:greedy-case2}
        \text{Greedy strategy 2: } v^*\gets \arg\min_{v\in \Gamma_{min}}\overline{d}_{\Gamma_{min}}(v).
    \end{equation}
\end{itemize}
We note that every vertex in $C_{temp}$ has at most two non-neighbors in $C_{temp}$ based on the definition of the MissingTwoDeg problem. Thus, the above strategy covers all possible cases. In addition, we remark that if there are multiple choices when selecting $v^*$ in Equation~(\ref{Eq:greedy-case1}) or Equation~(\ref{eq:greedy-case2}), e.g., more than one vertex in $\Gamma_{min}$ has the smallest number of non-neighbors in $C_{temp}$ or $\Gamma_{min}$, $v^*$ can be chosen arbitrarily from among them.

We also provide an example of \texttt{IRSolver} in Section~\ref{app:example} of our technical report~\cite{appendix}.

\smallskip
\noindent\textbf{Correctness of \texttt{IRSolver}}. The correctness of \texttt{IRSolver} can be guaranteed by the following lemma.
\begin{lemma}
\label{lemma:IRSlover}
    Given a MissingTwoDeg problem instance $(g,S,C)$, \texttt{IRSolver} finds the largest $k$-defective clique $g_{opt}$ that is a subgraph of $g$ and contains all vertices in $S$.
\end{lemma}
\begin{proof}
    Note that if $g$ is a $k$-defective clique, the problem can be directly solved by returning $g$ (Line 1). In the following, we assume that $g$ is not a $k$-defective clique. Clearly, \texttt{IRSolver} can terminate at Line 9 since $S\cup C$ is a superset of $S_{opt}$ and is not a $k$-defective clique. Thus, assume that \texttt{IRSolver} runs in $\tau + 1$ rounds ($\tau\geq 1$). At the $i$-th $(1\leq i\leq \tau)$ round, it adds $v^*_{i}$ to $S_{opt}$. Finally, it terminates at the $(\tau+1)$-st round since $S\cup\{v^*_1,...,v^*_{\tau+1}\}$ is no longer a $k$-defective clique. In addition, for $0\leq i\leq \tau+1$, we let
    \begin{equation}
        S_{opt}^i=S\cup\{v^*_1,...,v^*_{i}\} \text{ and } C_{temp}^i=C\setminus\{v^*_1,...,v^*_{i}\}.
    \end{equation}
    \begin{equation}
        \Gamma_{min}^i = \{v\in C_{temp}^i\mid \forall w\in C_{temp}^i, \overline{d}_{S_{opt}^i}(v)\leq \overline{d}_{S_{opt}^i}(w)\}
    \end{equation}
    Note that $S_{opt}^0$ and $C_{temp}^0$ are $S$ and $C$, respectively.
    
    We first show that, when \texttt{IRSolver} terminates, $g_{opt}=G[S_{opt}^{\tau}]$ is a \emph{maximal $k$-defective clique} in $g$, i.e., any vertex in $C\setminus V(g_{opt})$ cannot be added to $g_{opt}$ to form a larger $k$-defective clique. The reasons are as follows. Consider the last round of \texttt{IRSolver} where $G[S_{opt}^{\tau}\cup \{v^*_{\tau+1}\}]$ is no longer a $k$-defective clique. We note that $v^*_{\tau+1}$ has the smallest number of non-neighbors in $S_{opt}^{\tau}$ among other vertices in $C_{temp}^{\tau}$ (since it is selected from $\Gamma_{min}$). Therefore, for any vertex $v$ in $C_{temp}^{\tau}$, $G[S_{opt}^{\tau}\cup \{v\}]$ is not a $k$-defective clique since $|\overline{E}(S_{opt}^{\tau}\cup \{v\})|=|\overline{E}(S_{opt}^{\tau})|+\overline{d}_{S_{opt}^{\tau}}(v)\geq |\overline{E}(S_{opt}^{\tau})|+\overline{d}_{S_{opt}^{\tau}}(v^*_{\tau+1})>k$. 

    We then show that $g_{opt}$ is the largest $k$-defective clique in $g$ that contains $S$ by contradiction. Note that $g_{opt}$ contains $S$ clearly since $S_{opt}$ is a superset of $S$. Let $g_{sol}$ be the largest $k$-defective clique in $g$ that contains $S$. Assume that $g_{opt}$ is not the largest $k$-defective clique that contains $S$, i.e., $|V(g_{sol})|>|V(g_{opt})|$. Then, we will show that, by the following construction process starting with $g_{sol}$, there exists a largest $k$-defective clique $g'$ containing $S$ such that $g_{opt}$ is a subgraph of $g'$, i.e., $V(g_{opt}) \subset V(g')$, which contradicts to the fact that $g_{opt}$ is a maximal $k$-defective clique in $g$.

    Finally, to complete the proof, we introduce an iterative process to construct $g'$ based on $g_{opt}$ and $g_{sol}$. Let $\langle v_1^*,v_2^*,...,v_{\tau}^* \rangle$ be the ordering of vertices in $V(g_{opt})\setminus S$. The procedure has four steps.
    \begin{itemize}[leftmargin=*, topsep=0pt]
        \item \textbf{Step 1: Initialization.} Initialize $g'_0$ to be $g_{sol}$ and $i$ to be 0;
        \item \textbf{Step 2: Termination checking.} If $g_i'$ contains all vertices in $g_{opt}$, i.e., $V(g_i')\supseteq V(g_{opt})$, set $g'=g_i'$ and stop the process.
        \item \textbf{Step 3: Construction.} Find the vertex $v^*_{o_i}$ with the smallest index $o_i$ $(1\leq o_i\leq \tau)$ that is in $V(g_{opt})\setminus S$ but not in $g_{i}'$. 
        We can easily deduce that $g_i'$ contains all vertices in $S_{opt}^{o_i-1}$ and $V(g_i')\setminus S_{opt}^{o_i-1}$ is a subset of $C_{temp}^{o_i-1}$, formally,
        \begin{equation}
            S_{opt}^{o_i-1} \subseteq V(g_i') \text{ and } V(g_i')\setminus S_{opt}^{o_i-1} \subseteq C_{temp}^{o_i-1}.
        \end{equation} 
        Based on the above, we have
        \begin{equation}
            V(g_i') = S_{opt}^{o_i-1} \cup \Psi_i \text{ where } \Psi_i= V(g_i')\setminus S_{opt}^{o_i-1}.
        \end{equation}
        We then have three cases to construct $g'_{i+1}$ (note that $\overline{d}_{\Psi_i}(v^*_{o_i})\leq \overline{d}_{C}(v^*_{o_i})\leq 2$ based on Definition~\ref{def:IR}).
        
        \begin{itemize}[leftmargin=*]
            \item \textbf{Case 1:} $\overline{d}_{\Psi_i}(v^*_{o_i})\leq 1$. Let $v$ be an arbitrary vertex in $\Psi_i$ if $\overline{d}_{\Psi_i}(v^*_{o_i})= 0$; otherwise if $\overline{d}_{\Psi_i}(v^*_{o_i})= 1$, let $v$ be $v^*_{o_i}$'s non-neighbor in $\Psi_i$. We construct $g_{i+1}'$ by swapping vertex $v$ with $v^*_{o_i}$, formally, $g_{i+1}' = g[V(g_i')\setminus \{v\}\cup \{v^*_{o_i}\} ]$. Clearly, $g_{i+1}'$ is a $k$-defective clique since $|\overline{E}(g_{i+1}')| = |\overline{E}(g_{i}')|-\overline{d}_{g_i'}(v) +\overline{d}_{V(g_i')\setminus\{v\}}(v^*_{o_i})\leq |\overline{E}(g_{i}')|-\overline{d}_{S_{opt}^{o_i-1}}(v) +\overline{d}_{S_{opt}^{o_i-1}}(v^*_{o_i}) + \overline{d}_{\Psi_i\setminus\{v\}}(v^*_{o_i})\leq k$ (note that $\overline{d}_{\Psi_i\setminus\{v\}}(v^*_{o_i})$ is clearly equal to 0 and $v^*_{o_i}$ has the smallest number of non-neighbors in $S_{opt}^{o_i-1}$ among other vertices in $C_{temp}^{o_i-1}$ based on our greedy strategy).
            \item \textbf{Case 2:} $\overline{d}_{\Psi_i}(v^*_{o_i})= 2$ and $\overline{d}_{\Gamma_{min}^{o_i-1}}(v^*_{o_i}) \leq 1$. In this case, all $v^*_{o_i}$'s non-neighbors in $C_{temp}^{o_i-1}$ are in $\Psi_i$ (note that both $\Psi_i$ and $\Gamma_{min}^{o_i-1}$ are subsets of $C_{temp}^{o_i-1}$).
            Let $u$ be a $v^*_{o_i}$'s non-neighbor in $\Psi_i$ but not in $\Gamma_{min}^{o_i-1}$ (note that such a vertex $u$ exits based on the condition of this case). We construct $g_{i+1}'$ by swapping vertex $u$ with $v^*_{o_i}$, formally, $g_{i+1}' = g[V(g_i')\setminus \{u\}\cup \{v^*_{o_i}\} ]$. 
            We can deduce that $g_{i+1}'$ is a $k$-defective clique since (1) $\overline{d}_{S_{opt}^{o_i-1}}(u)$ is strictly larger than $\overline{d}_{S_{opt}^{o_i-1}}(v^*_{o_i})$ (note that vertices in $\Gamma_{min}^{o_i-1}$ has the smallest number of non-neighbors in $S_{opt}^{o_i-1}$ and $v^*_{o_i}$ is selected from $\Gamma_{min}^{o_i-1}$ based on our greedy strategy while $u$ is not in $\Gamma_{min}^{o_i-1}$) and (2) $|\overline{E}(g_{i+1}')| = |\overline{E}(g_{i}')|-\overline{d}_{g_i'}(u) +\overline{d}_{V(g_i')\setminus\{u\}}(v^*_{o_i})\leq |\overline{E}(g_{i}')|-\overline{d}_{S_{opt}^{o_i-1}}(u) +\overline{d}_{S_{opt}^{o_i-1}}(v^*_{o_i}) + \overline{d}_{\Psi_i\setminus\{u\}}(v^*_{o_i}) = |\overline{E}(g_{i}')|-\overline{d}_{S_{opt}^{o_i-1}}(u) +\overline{d}_{S_{opt}^{o_i-1}}(v^*_{o_i}) + 1 \leq k$.
            \item \textbf{Case 3:} $\overline{d}_{\Psi_i}(v^*_{o_i}) = 2$ and $\overline{d}_{\Gamma_{min}^{o_i-1}}(v^*_{o_i}) = 2$. In this case, any vertex $v$ in $\Gamma_{min}^{o_i-1}$ has two non-neighbors in $\Gamma_{min}^{o_i-1}$ based on our greedy strategy 2 and is adjacent to all vertices in $C\setminus \Gamma_{min}^{o_i-1}$ based on Definition~\ref{def:IR}. Clearly, there exists a circle in the complementary graph of $g[\Gamma_{min}^{o_i-1}]$ that contains $v^*_{o_i}$. W.l.o.g., let $\langle w_0,w_1,...,w_c=w_0 \rangle$ ($c\geq 2$) be the circle where vertex $w_j$ is not adjacent to $w_{j+1}$ $(0\leq j\leq c-1)$ in $g$, and let $w_0 = v^*_{o_i}$. Let $\{w_{x_1},w_{x_2},...,w_{x_t}\}$ be the vertices in $\{w_0,...,w_c\}\cap \Psi_i$, as illustrated in Figure~\ref{fig:case3}. Clearly, $w_0=w_c$ is not in $\{w_{x_1},w_{x_2},...,w_{x_t}\}$. We construct $g_{i+1}'$ by swapping the set of vertices $\{w_{x_1},w_{x_2},...,w_{x_t}\}$ with $\{w_{x_1+1},w_{x_2+1},...,w_{x_t+1}\}$, i.e., $g_{i+1}'=g[V(g'_i)\setminus \{w_{x_1},w_{x_2},...,w_{x_t}\} \cup \{w_{x_1+1},w_{x_2+1},...,w_{x_t+1}\}$. We can deduce that $w_0$ is in $\{w_{x_1+1},...,w_{x_t+1}\}$ since $w_0$'s non-neighbors $w_1$ and $w_{c-1}$ are both in $\Psi_i$ in this case. Therefore, $g'_{i+1}$ contains $v^*_{o_i}$ \revision{and is a $k$-defective clique, since (1) every vertex in $\{w_0,w_1,\ldots,w_c\}$ belongs to $\Gamma_{min}^{o_i-1}$, and hence all these vertices have the same number of non-neighbors in $S_{opt}^{o_i-1}$; and (2) we know that $g_1=g[\{w_{x_1},w_{x_2},\ldots,w_{x_t}\}]$ and $g_2=g[\{w_{x_1+1},w_{x_2+1},\ldots,w_{x_t+1}\}]$ are isomorphic under the mapping $f:V(g_1)\rightarrow V(g_2)$ defined by $f(w_j)=w_{j+1}$. Indeed, for any $(w_x,w_y)\in E(g_1)$, we have $x-y\equiv 1 \pmod{t}$, and thus $(x+1)-(y+1)\equiv 1 \pmod{t}$, which implies $(f(w_x),f(w_y))=(w_{x+1},w_{y+1})\in E(g_2)$. Thus, the corresponding vertices have the same number of non-neighbors within the induced subgraph.}
            \end{itemize} 
        \item  \textbf{Step 4: Repetition.} Increase $i$ by one and go to \textbf{Step 2}.
    \end{itemize}
    Clearly, the invariants $o_i<o_{i+1}$ and $|V(g_i')|=|V(g_{i+1})'|=|V(g_{sol})|$ are kept during the above construction. Therefore, it can terminate and find the largest $k$-defective clique $g'$ in $g$ that contains $g_{opt}$ and has more vertices than $g_{opt}$, i.e., $V(g_{opt}) \subset V(g')$.
\end{proof}

\begin{figure}[t]
\centering
\label{fig:case3}
\subfigure[Illustrating $\{w_{x_1},...,w_{x_t}\}$]{
        \includegraphics[width=0.2\textwidth]{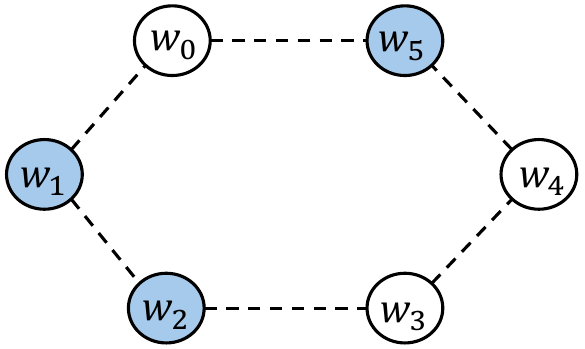}
}
\subfigure[Illustrating $\{w_{x_1+1},...,w_{x_t+1}\}$]{
        \includegraphics[width=0.2\textwidth]{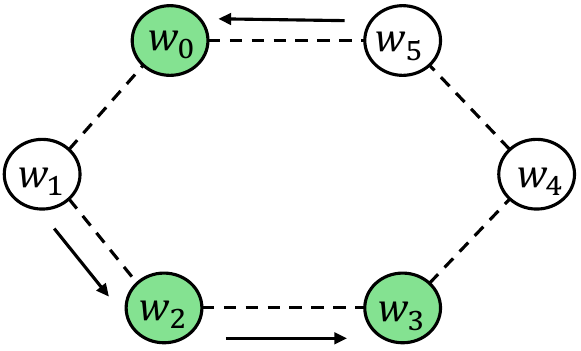}
}

\vspace{-0.2in}
\caption{An example to illustrate Case 3 in the proof of Lemma~\ref{lemma:IRSlover} with $c=6$ and $\{w_{x_1},...,w_{x_t}\}=\{w_1,w_2,w_5\}$ (indicated in blue); we have $\{w_{x_1+1},...,w_{x_t+1}\}=\{w_2,w_3,w_6=w_0\}$ (indicated in green).}
\label{fig:case3}
\end{figure}

\noindent \textbf{Time complexity of \texttt{IRSolver}}. \revision{The time complexity of \texttt{IRSolver} is bounded by $O(|S||C|+|C|^2)$. To support the greedy selection efficiently, we employ bucket queues (a form of linear heap) to dynamically maintain, for each vertex $v$ in the current candidate set, the key $k(v)=3\cdot \overline{d}_{S_{opt}}(v)+\overline{d}_{C_{temp}}(v)$. This key exactly encodes the lexicographic order first by $\overline{d}_{S_{opt}}(v)$ and then by $\overline{d}_{C_{temp}}(v)$, since $\overline{d}_{C_{temp}}(v)\le 2$. Moreover, since $\overline{d}_{S_{opt}}(v)\le |S|+|C|$, we have $k(v)\in [0,\,3(|S|+|C|)]$. Accordingly, we use an array of size $3(|S|+|C|)+1$, where the $i$-th bucket stores a linked list of vertices whose key equals $i$. The initialization of \texttt{IRSolver} requires computing $\overline{d}_{S_{opt}}(v)$ and $\overline{d}_{C_{temp}}(v)$ for all vertices in $C$, which takes $O(|S||C|+|C|^2)$ time. During the greedy process, the $\arg\min$ operation for selecting $v^*$ over $\Gamma_{min}$ is supported directly by the bucket structure, rather than by recomputing the minimum from scratch in each iteration. In particular, the minimum non-empty bucket exactly contains the vertices in $\Gamma_{min}$ that are minimal under the above lexicographic order. Since the maintained key is non-decreasing during the process, the pointer to the minimum non-empty bucket moves forward at most $O(|S|+|C|)$ times in total. In addition, deleting the selected vertex $v^*$ takes $O(1)$ time, and updating all affected vertices also takes $O(1)$ time per iteration, since each vertex in $C$ has at most two non-neighbors within $C$. Hence, after initialization, the total maintenance cost of the bucket structure is $O(|S|+|C|)$. Therefore, the overall time complexity of \texttt{IRSolver} is dominated by the initialization step and remains $O(|S||C|+|C|^2)$.}

\subsection{\texttt{BBRes}: New Branching Strategy}
\label{subsec:BBRes-Branching}

\noindent \textbf{Motivation.} 
Consider the branching procedure at a branch $(g,S,C)$. Recent studies have shown that the performance of the proposed method depends on the pivot selection strategy used during branching. We first briefly revisit the existing pivot selection strategy~\cite{Chang24kDC-2}. Specifically, it selects from $C$ the vertex that has at least one non-neighbor in $S$ as the pivot $v_p$; if no such vertex exists, it selects from $C$ an arbitrary vertex as the pivot $v_p$. We observe that the existing branching strategy~\cite{Chang24kDC-2} could select those vertices that have at most two non-neighbors in $C$ as the pivot. We note that those vertices can be skipped when selecting the pivot, and if no pivot can be chosen from $C$, the branch $(g,S,C)$ falls in the MissingTwoDeg problem instances and can be solved by \texttt{IRSolver}. Therefore, the existing branching strategy is not suitable for our framework as it will create some redundant branches.

\noindent \textbf{Our branching strategy}. Motivated by the above, we propose the following branching strategy called \textbf{BS-three}.
\begin{itemize}[leftmargin=*, topsep=0pt]
    \item \textbf{BS-three}. Given a branch $(g,S,C)$, the pivot $v_p$ is selected as the vertex in $C$ that (1) has at least \emph{three} non-neighbors in $C$ and (2) has at least one non-neighbor in $S$; if no such vertex exists, the pivot $v_p$ is chosen as an arbitrary vertex in $C$ that has at least \emph{three} non-neighbors in $C$. 
\end{itemize}
We remark that if \textbf{BS-three} cannot find a pivot in $C$, the branch becomes a MissingTwoDeg problem and can be solved by \texttt{IRSolver}. In addition, \textbf{BS-three} cannot be applied to the previous BB framework~\cite{chang2023kdefective,Chang24kDC-2,dai2024theoretically} since they cannot handle the MissingTwoDeg problems (when no pivot can be found in $C$). 

\smallskip
\noindent\textbf{Time complexity of \texttt{BBRes}}. With the proposed \textbf{BS-three}, our \texttt{BBRes} will generate $O(\lambda_k^{\delta\Delta})$ branches when $|V(g^*)|\geq k+1$, and generate $O(\lambda_k^{n})$ branches otherwise in the worst case. Here, $\delta$ is the degeneracy of $G$, $\Delta$ is the maximum degree of $G$, and $\lambda_k$ is the largest real root of the equation $x^{k+4}-2x^{k+3}+x^{3}-x+1=0$ (based on the analytical method in~\cite{fomin2010exact}). For example, $\lambda_k=1.381$, 1.705, 1.867 when $k=1$, 2, 3. Thus, \texttt{BBRes} has the worst-case time complexity of $O^*(\lambda_k^{\delta\Delta})$ when $|V(g^*)|\geq k+1$, and of $O(\lambda_k^{n})$ otherwise, where $O^*$ suppresses the polynomials, as summarized in Lemma~\ref{lemma:worst-case-TC}. Recall that the state-of-the-art method following Algorithm~\ref{alg:basic} has the time complexity of $O^*(\gamma_k^n)$ where $\gamma_k$ is the largest real root of $x^{k+3}-2x^{k+2}+x^2-x+1 = 0$, e.g., $\gamma_k=1.466, 1.755, 1.889$ for $k=1,2,3$. We remark that $\lambda_k$ is \emph{strictly smaller} than $\gamma_k$ for $k\geq 1$. \revision{A formal proof is provided in Appendix~\ref{app:proofs} of our technical report~\cite{appendix}; see the proof of Lemma~\ref{lemma:worst-case-TC}. Moreover, for $k=1,3,5,10,15,20$, the corresponding values of $\gamma_k$ are $1.465$, $1.8885$, $1.9750$, $1.9993$, $1.99998$, and $1.9999993$, while the corresponding values of $\lambda_k$ are $1.3803$, $1.8668$, $1.9706$, $1.9991$, $1.99997$, and $1.9999992$.}

\begin{lemma}
\label{lemma:worst-case-TC}
    Given a graph $G$ and a integer $k$, \texttt{BBRes} runs in $O(n(\Delta\delta)^2\lambda_k^{\Delta\delta})$ when the largest $k$-defective clique in $G$ is of size larger than $k+1$ (i.e., $|V(g^*)|\geq k+2$), and runs in $O(n^2\lambda_k^n)$, where $\delta$ is the degeneracy of $G$, $\Delta$ is the maximum degree of $G$, and $\lambda_k$ is the largest real root of the equation $x^{k+4}-2x^{k+3}+x^{3}-x+1=0$.
\end{lemma}
\begin{proof}
    We note that each recursion of \texttt{BBRes} runs in polynomial time $O(|S||C|+|C|^2)$, which is dominated by \texttt{IRSolver} (details are discussed in Section~\ref{subsec:BBRes-Termination}). When $|V(g^*)|\geq k+2$, we can utilize the diameter-two property, then the size of $G_{v_i}$ can be bounded by $\Delta\delta$~\cite{Chang24kDC-2} and thus $|C|\leq \Delta\delta$; otherwise, we have $|C|\leq n$. We then analyze the number of recursions (or branches) and put the details in Appendix~\ref{app:proofs} of our technical report~\cite{appendix}.
\end{proof}

%% file: 5-UB.tex
\section{Double-Coloring Upper Bound}\label{sec:ub}
The upper bound estimation is a crucial component in algorithms for the maximum $k$-defective problem, as it directly impacts the efficiency and effectiveness of the BB search process. Existing state-of-the-art methods typically rely on graph coloring techniques~\cite{chang2023kdefective,Chang24kDC-2,dai2024theoretically}. By assigning colors to vertices in the candidate set, these methods partition the candidate set into several independent sets, thereby facilitating the computation of upper bounds. However, this classic approach makes a simplifying assumption: it treats every pair of vertices in different color classes as if they must be connected by an edge. While this assumption simplifies the analysis, it often leads to loose upper bounds, especially in graphs where the connectivity between color classes is sparse.

Recognizing this limitation, we propose to advance upper bound computation by introducing a novel \emph{double-coloring} strategy in this section. Rather than relying on a single coloring, we assign each vertex a pair of colors, derived from two independent coloring processes. This dual perspective enables us to capture more refined relationships between vertices. For example, if two vertices fall into different color classes in the first coloring but share the same color in the second, we gain additional insight: despite what the first coloring suggests, these vertices actually belong to an independent set under the second coloring.
To illustrate, consider the example in Figure~\ref{fig:tikz_pair}, which contrasts the granularity of structural information captured by the two schemes. 
As shown in Figure~\ref{fig:tikz_pair}(a), single coloring assigns identical colors to non-adjacent vertices (e.g., the blue nodes), implicitly treating them as a clique and thus identifying only 2 missing edges. 
In contrast, Figure~\ref{fig:tikz_pair}(b) demonstrates that double coloring distinguishes these vertices via unique color pairs, successfully exposing all 4 missing edges (indicated by dashed lines). 
In the following, we first present the computation of the double-coloring upper bound (Section~\ref{sec:ub-computation}) and then prove its correctness (Section~\ref{sec:ub-correctness}).

\begin{figure}[t]
\centering
\subfigure[Single coloring]{
\centering
\begin{tikzpicture}[scale=0.65, node distance=1.5cm, thick] 
  \node[circle, minimum size=3mm, fill=red, draw] (1) at (1.5,0) {};
  \node[circle, minimum size=3mm, fill=red, draw] (2) at (2.5,0.5) {};
  \node[circle, minimum size=3mm, fill=blue, draw] (3) at (1.5,2) {};
  \node[circle, minimum size=3mm, fill=blue, draw] (4) at (2.5,1.5) {};
  \node[circle, minimum size=3mm, draw] (5) at (0,1) {};
  
  \begin{scope}[on background layer]
    \node[fit=(1)(2)(3)(4), fill=gray!10, rounded corners, inner sep=5pt, label=below:$C$] {};
    \node[fit=(5), fill=gray!20, rounded corners, inner sep=5pt, label=below:$S$] {};
  \end{scope}
  
  \foreach \a/\b in {5/1,5/2,5/3,5/4,1/4,2/3}
    \draw[thick] (\a) -- (\b);
  \foreach \a/\b in {1/2,3/4}
    \draw[thick,dashed] (\a) -- (\b);
\end{tikzpicture}
}
\hspace{0.08\columnwidth}
\subfigure[Double coloring]{
\centering
\begin{tikzpicture}[scale=0.65, node distance=1.5cm, thick]
  \node[circle split, minimum size=3mm, draw, path picture={
    \fill[red] (path picture bounding box.west) rectangle (path picture bounding box.north east) -- cycle;
    \fill[yellow] (path picture bounding box.south west) rectangle (path picture bounding box.east) -- cycle;
  }] (1) at (1.5,0) {};
  \node[circle split, minimum size=3mm, draw, path picture={
    \fill[red] (path picture bounding box.west) rectangle (path picture bounding box.north east) -- cycle;
    \fill[green] (path picture bounding box.south west) rectangle (path picture bounding box.east) -- cycle;
  }] (2) at (2.5,0.5) {};
  \node[circle split, minimum size=3mm, draw, path picture={
    \fill[blue] (path picture bounding box.west) rectangle (path picture bounding box.north east) -- cycle;
    \fill[yellow] (path picture bounding box.south west) rectangle (path picture bounding box.east) -- cycle;
  }] (3) at (1.5,2) {};
  \node[circle split, minimum size=3mm, draw, path picture={
    \fill[blue] (path picture bounding box.west) rectangle (path picture bounding box.north east) -- cycle;
    \fill[green] (path picture bounding box.south west) rectangle (path picture bounding box.east) -- cycle;
  }] (4) at (2.5,1.5) {};
  \node[circle, minimum size=3mm, draw] (5) at (0,1) {};
  
  \begin{scope}[on background layer]
    \node[fit=(1)(2)(3)(4), fill=gray!10, rounded corners, inner sep=5pt, label=below:$C$] {};
    \node[fit=(5), fill=gray!20, rounded corners, inner sep=5pt, label=below:$S$] {};
  \end{scope}
  
  \foreach \a/\b in {5/1,5/2,5/3,5/4,1/4,2/3}
    \draw[thick] (\a) -- (\b);
  \foreach \a/\b in {1/2,3/4,1/3,2/4}
    \draw[thick,dashed] (\a) -- (\b);
\end{tikzpicture}
}

\vspace{-8pt} 

\subfigure[Cost-flow graph with costs]{
\centering
\begin{tikzpicture}[scale=0.65, node distance=1.5cm, thick]
  \node[circle, minimum size=3mm, draw] (A) at (0,1) {s};
  \node[circle, minimum size=3mm, fill=red, draw] (B) at (1,0) {};
  \node[circle, minimum size=3mm, fill=blue, draw] (D) at (1,2) {};
  \node[circle, minimum size=3mm, fill=yellow, draw] (C) at (2.5,0) {};
  \node[circle, minimum size=3mm, fill=green, draw] (E) at (2.5,2) {};
  \node[circle, minimum size=3mm, draw] (F) at (3.5,1) {t};
  \draw[thick,->] (A) to[bend left=20] node[midway, above, sloped] {0} (B);
  \draw[thick,->] (A) to[bend right=20] node[midway, below, sloped] {1} (B);
  \draw[thick,->] (A) to[bend left=20] node[midway, above, sloped] {1} (D);
  \draw[thick,->] (A) to[bend right=20] node[midway, below, sloped] {0} (D);
  \draw[thick,->] (B) to node[midway, below, sloped] {0} (C);
  \draw[thick,->] (B) to node[pos=0.3, below, sloped] {0} (E);
  \draw[thick,->] (D) to node[pos=0.3, below, sloped] {0} (C);
  \draw[thick,->] (D) to node[midway, below, sloped] {0} (E);
  \draw[thick,->] (C) to[bend left=20] node[midway, above, sloped] {0} (F);
  \draw[thick,->] (C) to[bend right=20] node[midway, below, sloped] {1} (F);
  \draw[thick,->] (E) to[bend left=20] node[midway, above, sloped] {1} (F);
  \draw[thick,->] (E) to[bend right=20] node[midway, below, sloped] {0} (F);
\end{tikzpicture}
}
\hspace{0.08\columnwidth}
\subfigure[Calculated constrained max-flow]{
\centering
\begin{tikzpicture}[scale=0.65, node distance=1.5cm, thick]
  \node[circle, minimum size=3mm, draw] (A) at (0,1) {s};
  \node[circle, minimum size=3mm, fill=red, draw] (B) at (1,0) {};
  \node[circle, minimum size=3mm, fill=blue, draw] (D) at (1,2) {};
  \node[circle, minimum size=3mm, fill=yellow, draw] (C) at (2.5,0) {};
  \node[circle, minimum size=3mm, fill=green, draw] (E) at (2.5,2) {};
  \node[circle, minimum size=3mm, draw] (F) at (3.5,1) {t};
  \draw[thick,->] (A) to[bend left=20] node[midway, above, sloped] {0} (B);
  \draw[thick,->] (A) to[bend right=20] node[midway, below, sloped] {1} (B);
  \draw[thick,->] (A) to node[midway, below, sloped] {0} (D);
  \draw[thick,->] (B) to node[midway, below, sloped] {0} (C);
  \draw[thick,->] (B) to node[pos=0.3, below, sloped] {0} (E);
  \draw[thick,->] (D) to node[midway, below, sloped] {0} (E);
  \draw[thick,->] (C) to node[midway, above, sloped] {0} (F);
  \draw[thick,->] (E) to[bend left=20] node[midway, above, sloped] {1} (F);
  \draw[thick,->] (E) to[bend right=20] node[midway, below, sloped] {0} (F);
\end{tikzpicture}
}
\vspace{-0.2in}
\caption{Comparison of coloring strategies and upper bound computation ($k=2$). (a) Single coloring misses non-edges. (b) Double coloring reveals them via unique color pairs. (c) Cost-flow network with edge costs. (d) Max flow of 3 with cost 2 (flow shown on edges), yielding UB $= |S|+3$.}
\label{fig:tikz_pair}
\end{figure}
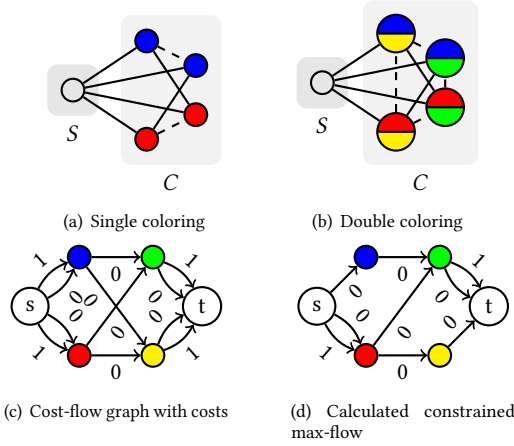

\subsection{Double-Coloring Upper Bound Computation}\label{sec:ub-computation}

\noindent \textbf{Additional notations}.
We now leverage the richer structural information revealed by double coloring. Specifically, for each vertex $v$ in the candidate set $C$, we assign a pair of colors, $(col_1(v), col_2(v))$, where $col_1(\cdot)$ and $col_2(\cdot)$ are two different vertex colorings. Both colorings respect the graph's edges: for any edge $(u, v)$ in $G$, we ensure $col_1(u) \neq col_1(v)$ and $col_2(u) \neq col_2(v)$. We refer to each ordered pair $(col_1(v), col_2(v))$ as the \emph{color pair} of $v$. The specific procedures for generating these colorings will be introduced later.

We further introduce a \emph{double-coloring indicator} function $I_G(u, v)$, which is defined as $1$ if two distinct vertices $u$ and $v$ share the same color in either of the two colorings, and $0$ otherwise. In other words, $I_G(u, v)$ indicates whether $u$ and $v$ are placed in the same color class in at least one of the two independently obtained colorings. It is easy to observe that, for any induced subgraph $g$, the quantity $\frac{1}{2} \sum_{u, v \in V(g)} I_g(u, v)$ provides a lower bound on the number of missing edges in $g$, i.e., $\frac{1}{2} \sum_{u, v \in V(g)} I_g(u, v) \leq |\overline{E}(g)|$, where $|\overline{E}(g)|$ is the number of non-edges in $g$. This is because $I_g(u, v) = 1$ implies that $u$ and $v$ are not adjacent in $g$, but the converse may not hold. 

\smallskip
\noindent \textbf{Overview of double-coloring upper-bound computation}. Given a branch $(g, S, C)$, our goal is to leverage the structural information provided by double coloring to select a subset $D$ of the candidate set $C$ such that the size of the maximum $k$-defective clique $g'$ in this branch (i.e., $S \subseteq V(g') \subseteq S \cup C$) is at most $|D| + |S|$. In other words, $|D| + |S|$ serves as our double-coloring upper bound, denoted as \textbf{UB-Double}. However, this is a non-trivial task, as incorporating information from two colorings introduces additional dependencies that must be carefully managed. 

To compute \textbf{UB-Double}, we construct a corresponding \emph{cost-flow graph} instance based on the current branch $(g,S,C)$ and employ the classic \emph{constrained maximum flow algorithm}~\cite{ahuja1993network,goldberg1987solving,dinic1970maxflowAlg}. In this setting, the objective is to maximize the flow (which corresponds to the number of vertices added to $S$), while ensuring that the total cost (which corresponds to the number of non-edges) does not exceed a threshold, say $k - |\overline{E}(S)|$. After obtaining the maximum flow in the cost-flow graph, we show that it corresponds to a subset $D$ of the candidate set $C$. Our \textbf{UB-Double} is then defined as $|S| + |D|$.

In the following, we present our double-coloring upper bound method  \textbf{UB-Double} based on the cost-flow graph and constrained maximum flow computation. 

\smallskip
\noindent\textbf{Double-coloring upper bound via cost-flow graph}. 
Given a branch $(g, S, C)$, we aim to compute a tight upper bound on the size of the maximum $k$-defective clique by leveraging the structural information provided by double coloring. To this end, we first assign two colorings to the vertices in $C$: the first coloring, $col_1(\cdot)$, is assigned greedily according to the degeneracy ordering; the second coloring, $col_2(\cdot)$, is assigned \revision{arbitrarily}\footnote{\revision{In our current implementation, the vertices in $C$ are processed according to their current memory order. As shown in Section~\ref{sec:second-coloring} of the technical report, this strategy achieves the best overall performance among the tested second-coloring orders. Designing more effective second-coloring orders remains an interesting future work.}}, with the additional constraint that no two vertices share the same color pair, i.e., $(col_1(u), col_2(u)) = (col_1(v), col_2(v))$ only if $u = v$.

Based on the assigned colorings, we construct the \emph{cost-flow (directed) graph} $g^c = (V^c, E^c)$. First, we define the vertex set $V^c$. Let $V'_1 = \{col_1(v) \mid v \in C\}$ and $V'_2 = \{col_2(v) \mid v \in C\}$ be the sets of colors derived from the two coloring processes. The vertex set is then defined as $V^c = \{s\} \cup \{t\} \cup V'_1 \cup V'_2$. 
As in Figure~\ref{fig:tikz_pair}(c) (with $k=2$), this vertex set forms a layered topology: $V'_1$ and $V'_2$ constitute two distinct layers bridging the source $s$ and the sink $t$. Next, we construct the edge set $E^c$ to encode the vertex selection and the $k$-defective constraint. The construction rules are as follows:
\begin{itemize}[leftmargin=*, topsep=0pt]
    \item \textbf{Source edges:} For each vertex $u \in V'_1$, we add $k+1$ parallel edges from $s$ to $u$, denoted as $(s, u)$, with capacity $1$ and costs $0, 1, \dots, k$, respectively. These edges distribute the allowable missing-edge budget.
    \item \textbf{Sink edges:} Similarly, for each vertex $v \in V'_2$, we add $k+1$ parallel edges from $v$ to $t$, with capacity $1$ and costs $0, 1, \dots, k$, respectively.
    \item \textbf{Internal edges:} For each vertex $x \in C$ in the original graph, let $u = col_1(x)$ and $v = col_2(x)$. We add a directed edge $(u, v)$ connecting the two color layers, with capacity $1$ and cost $\overline{d}_S(x)$. This edge represents the potential selection of vertex $x$.
    \item \textbf{Reverse edges:} For every edge constructed above, we add a corresponding reverse edge with capacity $0$ and the negation of the original cost.
\end{itemize}

Given the constructed cost-flow graph instance $g^c = (V^c, E^c)$, we solve the \emph{constrained maximum flow problem}~\cite{ahuja1995capacity,goldberg1987solving,dinic1970maxflowAlg}, which aims to maximize the flow such that the total cost does not exceed $k - |\overline{E}(S)|$. The value of the maximum flow obtained from this graph corresponds precisely to the size of the subset $D \subseteq C$ that can be added to $S$ while satisfying the $k$-defective constraint. Our double-coloring upper bound, denoted as \textbf{UB-Double}, is then defined as $|S| + |D|$.
As an illustration, Figure~\ref{fig:tikz_pair}(d) presents the computation result for the running example (where $k=2$). The algorithm identifies a valid flow of value 3 with a total cost of 2. Since the cost is within the budget, this implies a valid subset size $|D|=3$, yielding a refined upper bound of $|S|+3=4$. Furthermore, there exists a direct mapping between flows in $g^c$ and the subset $D$. Specifically, if a unit of flow in the maximum flow solution passes through an edge $(u, v)$ with $u \in V'_1$ and $v \in V'_2$, we include the corresponding vertex in the set $D$. To formalize this mapping, let $p: v \rightarrow (col_1(v), col_2(v))$ denote the correspondence from $g$ to $g^c$, where $v \in C$ and $(col_1(v), col_2(v)) \in E(g^c)$.

\subsection{Correctness and Complexity of Our Double-Coloring Upper Bound}\label{sec:ub-correctness}
We now establish the correctness proof and analyze the time complexity of our double-coloring upper bound \textbf{UB-Double}.

\smallskip
\noindent \textbf{Correctness proof}. We first establish the correctness of our proposed \textbf{UB-Double}. Specifically, given any branch $(g,S,C)$ and its corresponding cost-flow graph instance $g^c$, our goal is to show that for any $k$-defective clique $g'$ in this branch (where $S \subseteq V(g') \subseteq S \cup C$), the size of $g'$ cannot exceed $|S| + |D|$, where $D \subseteq C$ is the subset selected by our algorithm. 

The central idea of our proof is to relate the number of missing edges in the induced subgraph $G[S \cup D]$ to the cost of the constrained maximum flow in the cost-flow graph. Specifically, we observe that the total number of missing edges in $G[S \cup D]$ is equal to the sum of three terms: $|\overline{E}(S\cup D)| = |\overline{E}(S)| + |\overline{E}(D)| + \sum_{u \in D} \overline{d}_S(u)$. Our proof proceeds in two main steps:
\begin{itemize}[leftmargin=*, topsep=0pt]
    \item First, as shown in Lemma~\ref{lem:double-coloring-indicator}, for any subset $D$, the double-coloring indicator provides a lower bound on the number of missing edges within $D$, i.e., $\frac{1}{2} \sum_{u, v \in D} I_{G}(u,v) \leq |\overline{E}(D)$|.
    \item Second, as established in Lemma~\ref{lem:double-coloring-flow}, the cost $cost(g^c)$ computed in the cost-flow graph $g^c$ exactly matches the sum of the double-coloring indicator over $D$ and the number of missing edges between $S$ and $D$, i.e., $cost(g^c) = \frac{1}{2} \sum_{u, v \in D} I_G(u,v) + \sum_{u \in D} \overline{d}_S(u)$.
\end{itemize}
Putting these observations together, we see that if the total cost in the cost-flow graph does not exceed $k-\overline{E}(S)$, then the total number of missing edges in any $k$-defective clique formed by $S$ and $D$ is at most $k$. This means that any $k$-defective clique can be represented as a feasible flow in our construction. Since our constrained maximum flow algorithm finds the largest possible $D$ under the cost constraint, the solution $|S| + |D|$ serves as a valid upper bound on the size of any $k$-defective clique containing $S$ in this branch. Thus, our \textbf{UB-Double} produces a correct upper bound.

We now provide the detailed proof. We begin by presenting a lemma that follows directly from the definition of the double-coloring indicator.
\begin{lemma}\label{lem:double-coloring-indicator}
Given a graph $G = (V, E)$ and any subset $D \subseteq V$, we have $\frac{1}{2} \sum_{u, v \in D} I_{G}(u,v) \leq |\overline{E}(D)|$.
\end{lemma}

We are ready to give the following result, where the complete derivation is provided in Appendix~\ref{app:proofs} of our technical report~\cite{appendix}.

\begin{lemma}\label{lem:double-coloring-flow}
Given a constructed cost-flow graph $g^c$ and its constrained maximum flow, there exists a corresponding vertex set $D$ in $C$ such that $cost(g^c) = \frac{1}{2} \sum_{u, v \in D} I_G(u,v) + \sum_{u \in D} \overline{d}_S(u)$.
\end{lemma}

\revision{
\noindent \textbf{Effectiveness of \textbf{UB-Double} over \textbf{UB-Single}}.
We first briefly review the standard single-coloring-based upper bound, denoted by \textbf{UB-Single}~\cite{chang2023kdefective,Chang24kDC-2,dai2024theoretically}. Given a branch $(g,S,C)$, previous methods compute an upper bound by assigning a single coloring $col(\cdot)$ to the vertices in the candidate set $C$ according to the degeneracy ordering. This is exactly the first coloring step used in our double-coloring method \textbf{UB-Double}. The resulting bound is denoted by \textbf{UB-Single}.
We further show that our \textbf{UB-Double} is always no larger than the existing single-coloring-based upper bound \textbf{UB-Single}~\cite{chang2023kdefective,Chang24kDC-2,dai2024theoretically}, and is therefore always at least as tight.

\begin{lemma}
\label{lemma:tight-bound}
Given an instance $(g,S,C)$, \textbf{UB-Double} $\leq$ \textbf{UB-Single}.
\end{lemma}
The proof is provided in Appendix~\ref{app:proofs} of our technical report~\cite{appendix}.
}

\noindent \textbf{Time complexity of \textbf{UB-Double}}.
Compared with the standard single-coloring-based upper bound \textbf{UB-Single}, computing \textbf{UB-Double} incurs additional overhead. Consider a branch $(g, S, C)$. \textbf{UB-Double} first colors the vertices in $C$ twice: the first coloring requires $O(|E(g[C])|)$ time~\cite{brelaz1979new}, and the second coloring requires $O(|C| \times c) = O(|C|^2)$ time, where $c$ is the number of colors for $C$ and $c \leq |C|$. Constructing the corresponding cost-flow graph takes $O(ck) = O(|C|k)$ time. Finally, we apply a classical maximum flow algorithm~\cite{ahujia1993network,goldberg1987solving,dinic1970maxflowAlg} to solve the problem, which runs in $O(k|C|\log|C|)$ time. Thus, the overall time complexity of \textbf{UB-Double} is $O(k|C|\log|C| + |C|^2)$. \revision{By comparison, computing \textbf{UB-Single} needs $O(|C|+|E(g[C])|)$ time~\cite{chang2023kdefective,Chang24kDC-2,dai2024theoretically}, which is smaller. Nevertheless, in our setting, $|C|$ is bounded by the graph degeneracy $\delta$, which is typically small on real-world graphs. Hence, computing \textbf{UB-double} remains practically efficient.
}

%% file: 6-Experiments.tex
\section{Experiments}
\label{sec:exp}
In this section, we conduct experiments to evaluate the performance of our proposed algorithm \texttt{BBRes} against four baselines.
\begin{itemize}[leftmargin=*,topsep=0pt]
    \item \texttt{kDC2} \footnote{https://lijunchang.github.io/kDC-two/}: the existing algorithm proposed in~\cite{Chang24kDC-2}.
    \item \texttt{MDC} \footnote{https://github.com/dawhc/MaximumDefectiveClique/}: the existing method proposed in~\cite{dai2024theoretically}.
    \item \texttt{DnBK} \footnote{https://github.com/cy-Luo000/Maximum-k-Defective-Clique/}: the existing algorithm proposed in~\cite{Luo24defective}, augmented with the same Stage II techniques (Algorithm~\ref{alg:basic}) as \texttt{kDC2} and \texttt{MDC}, for solving the maximum $k$-defective clique problem without the size constraint.
    \item \texttt{WODC} \footnote{https://github.com/SNUCSE-CTA/WODC/}: the state-of-the-art method proposed in~\cite{jang2025efficient}.
\end{itemize} 
We note that prior algorithms, including \texttt{MADEC$^+$}~\cite{chen2021computing}, \texttt{KDBB}~\cite{gao2022exact}, \texttt{KD-Club}~\cite{jin2024kdclub}, and \texttt{kDC}~\cite{chang2023kdefective}, run significantly slower than the above baselines proposed recently and thus are omitted for the comparison.
All algorithms are implemented in C++, compiled with -O3, and run on a machine with an Intel CPU @ 2.60GHz and 256GB main memory. 
Following previous studies~\cite{Chang24kDC-2,dai2024theoretically, Luo24defective}, we choose $k$ from $\{1,3,5,10,15,20\}$ and set the time limit as 3 hours (i.e., 10,800s) for a fair comparison. 
We evaluate the algorithms on a widely adopted collection \footnote{http://lcs.ios.ac.cn/\textasciitilde caisw/Resource/realworld\%20graphs.tar.gz} of \textbf{139} real-world graphs~\cite{Chang24kDC-2,dai2024theoretically,Luo24defective} containing a total of $5.87 \times 10^7$ vertices.
Our source code can be found in~\cite{appendix}.

\subsection{Comparison with Baselines}

\begin{table}[t]
\caption{Number of solved instances by the algorithms with a 3-hour limit (best performers are highlighted in bold).}\label{tab:num-instances-all}
 \vspace{-0.15in}
\resizebox{0.55\columnwidth}{!}{
\begin{tabular}{|l|rrrrr|}
\hline
& \multicolumn{5}{c|}{Real-world graphs}   \\
& \texttt{BBRes} & \texttt{kDC2} & \texttt{MDC} & \texttt{DnBk} & \texttt{WODC}  \\
\hline
$k=1$ & \textbf{139} & 137 & 138 & 135 & 138 \\
$k=3$ & \textbf{139} & 136 & 136 & 135 & \textbf{139}  \\
$k=5$ & \textbf{139} & 136 & 131 & 134 & \textbf{139}  \\
$k=10$ & \textbf{135} & 128 & 127 & 124 & 134  \\
$k=15$ & \textbf{132} & 126 & 121 & 119 & 130 \\
$k=20$ & \textbf{129} & 114 & 113 & 111 & 127  \\
\hline
\end{tabular}
}
\end{table}

\noindent \textbf{Number of solved instances.} Table~\ref{tab:num-instances-all} summarizes the number of instances solved by each algorithm within the 3-hour cutoff. First, we observe a general downward trend in the number of solved instances as $k$ increases, confirming the exponential growth in computational complexity for larger $k$. Despite this challenge, \texttt{BBRes} consistently demonstrates superior performance, achieving the highest coverage across all settings. Significantly, the performance gap between \texttt{BBRes} and the competitors widens as $k$ increases. For example, at $k=20$, \texttt{BBRes} solves 129 instances, outperforming the runner-up (\texttt{WODC}) by 2 instances and significantly surpassing \texttt{kDC2} by 15 instances. We attribute this advantage to the effective synergy between our diameter-2 partitioning and the proposed double-coloring upper bound. Real-world graphs typically exhibit a ``globally sparse, locally dense'' topology. While diameter-2 partitioning decomposes the globally sparse structure into smaller sub-instances, these sub-instances often retain high local density. It is in these dense regions that baseline methods falter, as their bounds (e.g., single-coloring) become loose, leading to redundant branches. In contrast, our double-coloring strategy captures finer-grained structural conflicts, thereby tightening the bounds and effectively pruning the search space in these dense substructures.

\begin{figure}[t]
    \centering
    \vspace{-1em} 
    \subfigure[$k = 1$]{
        \includegraphics[width=0.22\textwidth, trim=0cm 0cm 0cm 0cm, clip]{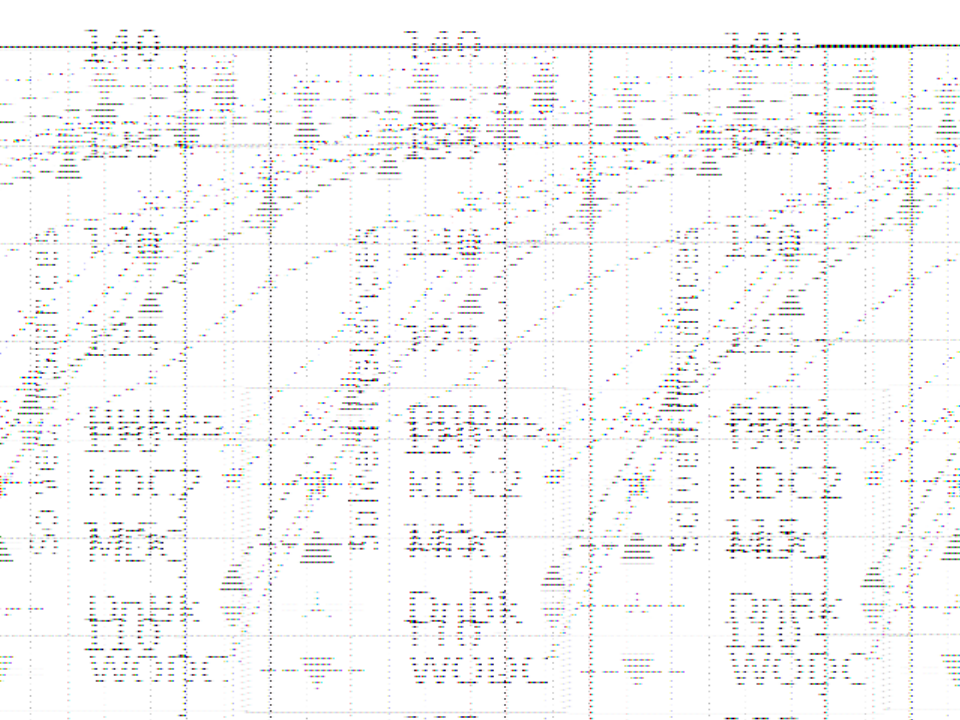}
        \label{fig:realworld_1}
    }
    \hspace{0em}
    \subfigure[$k = 3$]{
        \includegraphics[width=0.22\textwidth, trim=0cm 0cm 0cm 0cm, clip]{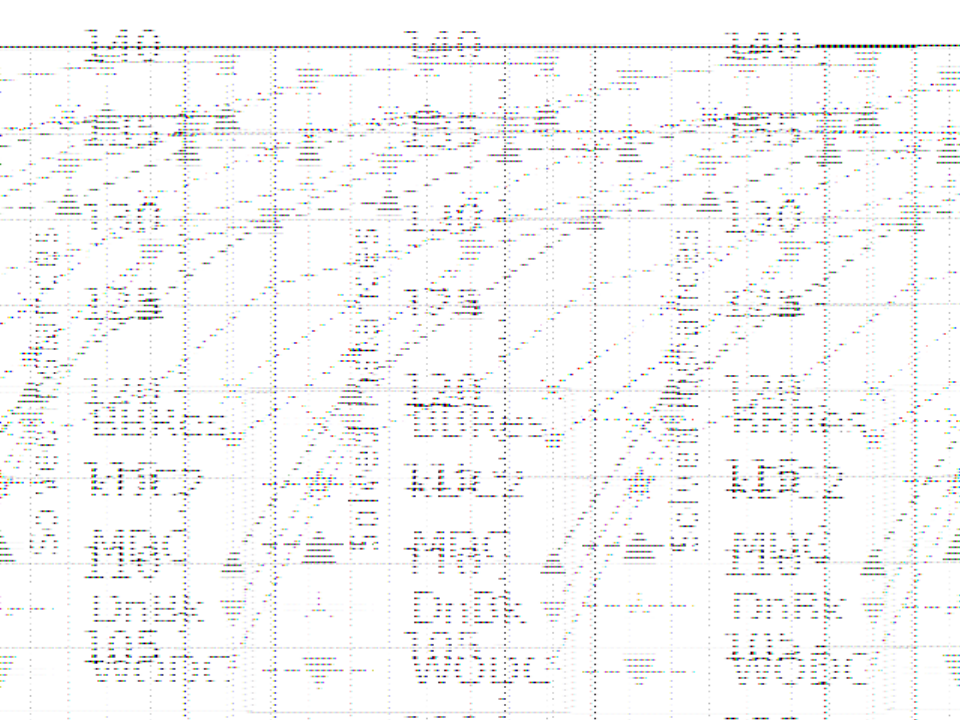}
        \label{fig:realworld_3}
    }
    
    \vspace{-1.2em} 
    
    \subfigure[$k = 5$]{
        \includegraphics[width=0.22\textwidth, trim=0cm 0cm 0cm 0cm, clip]{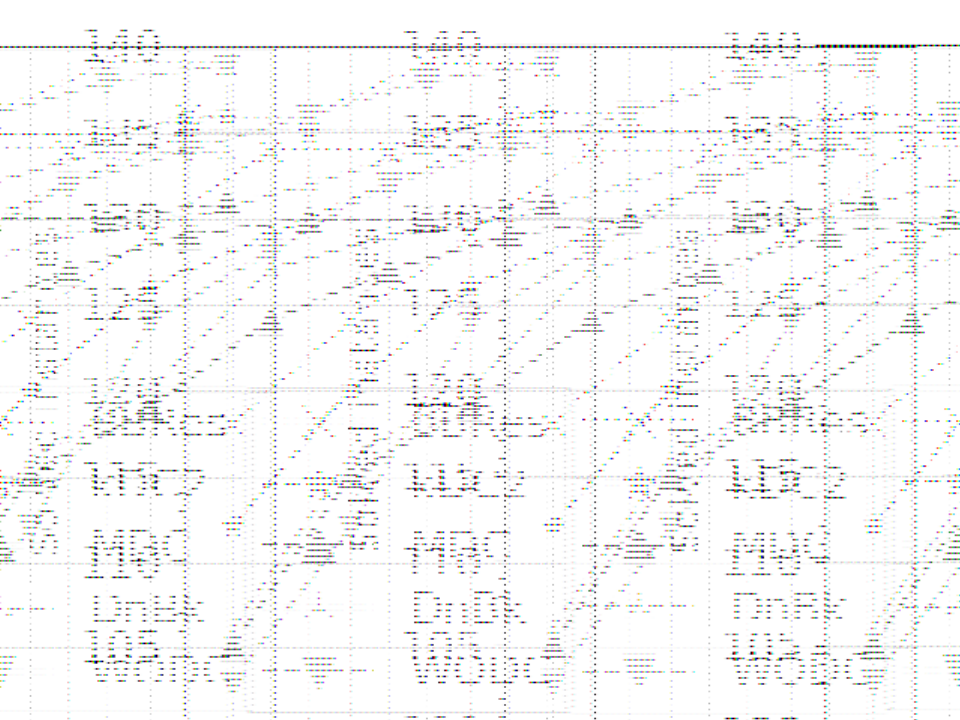}
        \label{fig:realworld_5}
    }
    \hspace{0em}
    \subfigure[$k = 10$]{
        \includegraphics[width=0.22\textwidth, trim=0cm 0cm 0cm 0cm, clip]{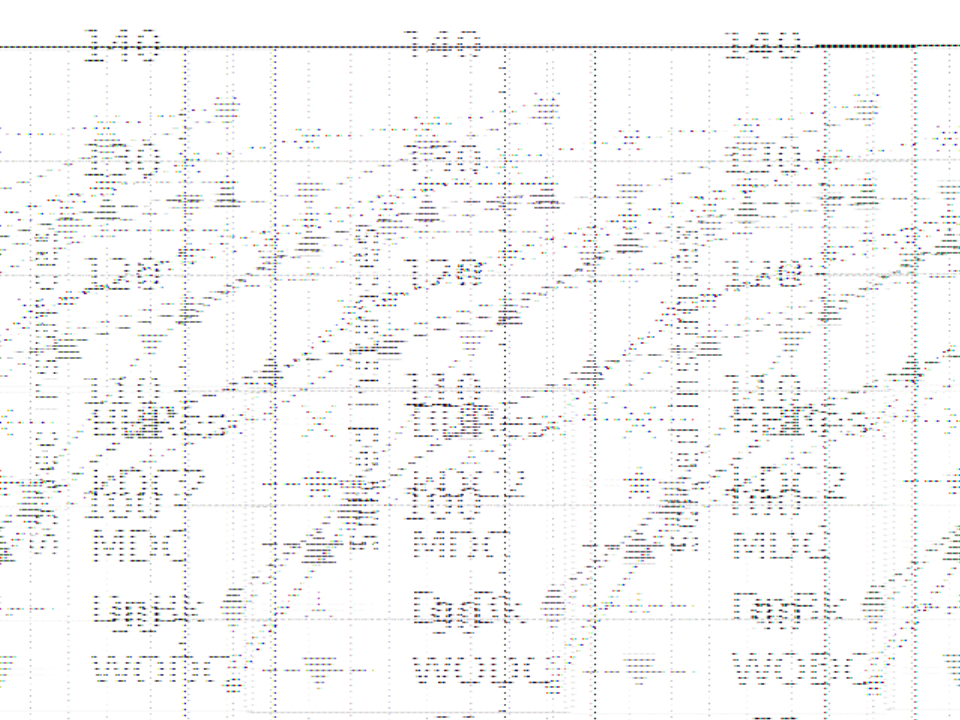}
        \label{fig:realworld_10}
    }
    
    \vspace{-1.2em} 
    
    \subfigure[$k = 15$]{
        \includegraphics[width=0.22\textwidth, trim=0cm 0cm 0cm 0cm, clip]{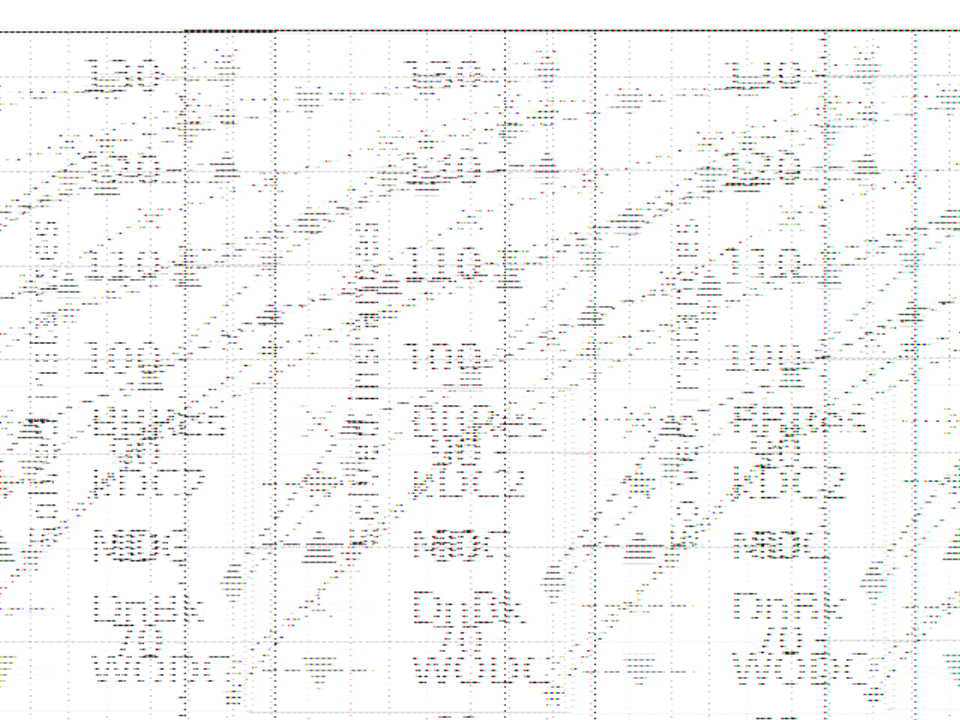}
        \label{fig:realworld_15}
    }
    \hspace{0em}
    \subfigure[$k = 20$]{
        \includegraphics[width=0.22\textwidth, trim=0cm 0cm 0cm 0cm, clip]{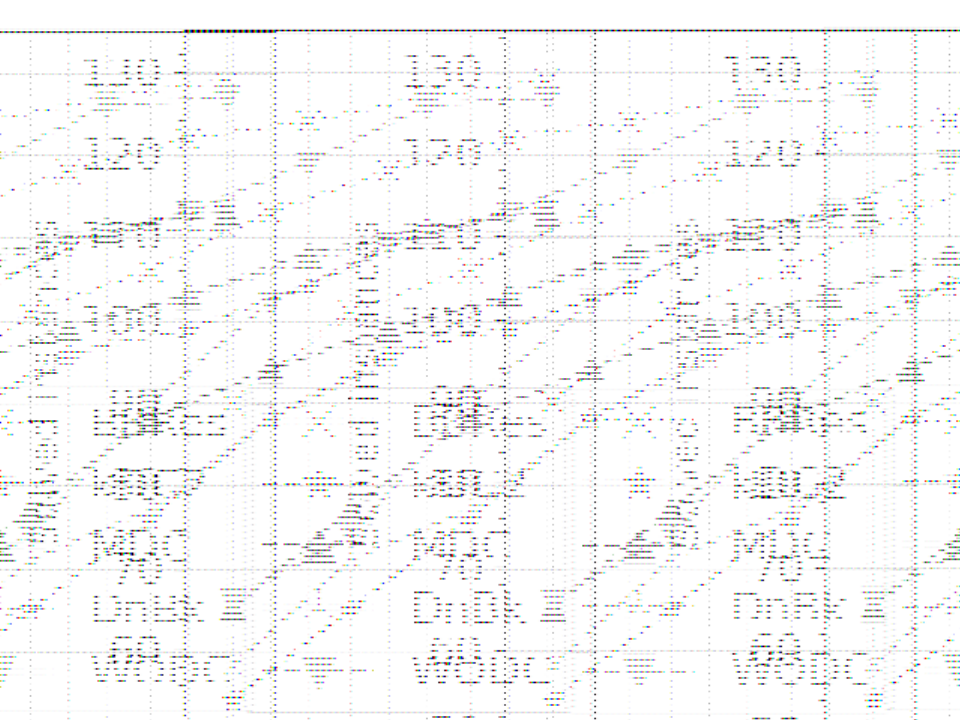}
        \label{fig:realworld_20}
    }

    \vspace{-0.2in} 
    \caption{Number of solved instances with varying time limits}\label{fig:realworld-full}
\end{figure}

\begin{table*}[t]
\caption{Running time (in seconds) on 30 representative real-world benchmark graphs. `-' indicates timeout. The best performer is highlighted in bold; specifically, if the running time is within 10\% of the fastest time, it is considered as the best.}
\vspace{-0.15in}
\label{tab:representative}
\resizebox{1.05\textwidth}{!}{
\begin{tabular}{|l|r|r|r|r|r|r|r|r|r|r|r|r|r|r|r|r|r|r|r|r|r|r|}
\hline
 &  &  & \multicolumn{5}{c|}{$k = 5$} & \multicolumn{5}{c|}{$k = 10$} & \multicolumn{5}{c|}{$k = 15$} & \multicolumn{5}{c|}{$k = 20$} \\
Graphs & $n$ & $m$ & \multicolumn{1}{c}{\texttt{BBRes}} & \multicolumn{1}{c}{\texttt{kDC2}} & \multicolumn{1}{c}{\texttt{MDC}} & \multicolumn{1}{c}{\texttt{DnBk}} & \multicolumn{1}{c|}{\texttt{WODC}} & \multicolumn{1}{c}{\texttt{BBRes}} & \multicolumn{1}{c}{\texttt{kDC2}} & \multicolumn{1}{c}{\texttt{MDC}} & \multicolumn{1}{c}{\texttt{DnBk}} & \multicolumn{1}{c|}{\texttt{WODC}} & \multicolumn{1}{c}{\texttt{BBRes}} & \multicolumn{1}{c}{\texttt{kDC2}} & \multicolumn{1}{c}{\texttt{MDC}} & \multicolumn{1}{c}{\texttt{DnBk}} & \multicolumn{1}{c|}{\texttt{WODC}} & \multicolumn{1}{c}{\texttt{BBRes}} & \multicolumn{1}{c}{\texttt{kDC2}} & \multicolumn{1}{c}{\texttt{MDC}} & \multicolumn{1}{c}{\texttt{DnBk}} & \multicolumn{1}{c|}{\texttt{WODC}}\\
\hline
socfb-A-anon & 3M & 23M & \textbf{6.513} &\textbf{6.65} & 19.343 & 21.89 & \textbf{6.645}  & \textbf{7.793} &37.453 & 32.101 & 48.902 & 8.919  & \textbf{10.689} &289.987 & 52.543 & 305.076 & 18.556  & \textbf{17.136} &988.404 & 140.895 & 1304.348 & 57.337 \\ 
soc-orkut & 2M & 106M & \textbf{104.263} &157.448 & 137.127 & 390.314 & 405.808  & \textbf{112.963} &576.496 & 188.372 & 523.604 & 455.723  & \textbf{129.552} &8249.529 & - & 1353.479 & 573.453  & \textbf{146.826} &- & - & - & 892.649 \\ 
socfb-B-anon & 2M & 20M & 10.165 & 8.984 & 18.302 & 21.135 & \textbf{6.883}  & 12.272 & 58.884 & 21.368 & 66.563 & \textbf{7.622}  & 25.553 & 1737.898 & 45.803 & 778.94 & \textbf{22.41}  & \textbf{19.994} &8811.373 & 378.424 & - & 79.273 \\ 
soc-flixster & 2M & 7M & \textbf{0.521} &4.136 & 2.74 & 2.461 & 1.189  & \textbf{0.594} &38.122 & 3.875 & 8.435 & 1.871  & \textbf{0.788} &1199.894 & 19.231 & 84.602 & 10.753  & \textbf{1.677} &- & 294.896 & 1785.488 & 89.086 \\ 
web-wikipedia2009 & 1M & 4M & \textbf{0.84} &\textbf{0.783} & 2.083 & 1.289 & 1.872  & 0.954 & \textbf{0.846} & 2.779 & 2.551 & 2.542  & \textbf{0.973} &\textbf{0.891} & 2.889 & 49.022 & 2.726  & \textbf{1.66} &\textbf{1.574} & 3.085 & 16.88 & 2.596 \\ 
tech-as-skitter & 1M & 11M & 0.707 & \textbf{0.593} & 1.877 & 0.982 & 1.288  & \textbf{0.755} &\textbf{0.702} & 1.972 & 2.173 & 1.26  & \textbf{0.886} &1.8 & 2.058 & 3.624 & 1.6  & \textbf{1.04} &21.869 & 2.804 & 10.765 & 2.498 \\ 
soc-pokec & 1M & 22M & 5.996 & \textbf{4.557} & 11.098 & 16.999 & 7.88  & 6.588 & \textbf{5.188} & 14.35 & 17.144 & 7.634  & \textbf{9.899} &\textbf{9.381} & 18.949 & 101.537 & \textbf{9.261}  & \textbf{13.476} &58.022 & 23.248 & 430.095 & 24.827 \\ 
soc-lastfm & 1M & 4M & 3.392 & 105.634 & 88.744 & 104.585 & \textbf{2.992}  & \textbf{13.906} &2826.062 & 2098.176 & - & \textbf{14.784}  & \textbf{64.639} &- & - & - & \textbf{70.443}  & \textbf{529.874} &- & - & - & \textbf{493.475} \\ 
soc-youtube-snap & 1M & 2M & \textbf{0.677} &3.844 & 2.42 & 3.009 & 1.247  & \textbf{1.705} &56.95 & 28.339 & 31.979 & 3.735  & \textbf{13.467} &1341.397 & 921.233 & 938.395 & 37.626  & \textbf{193.191} &- & - & - & 497.69 \\ 
ca-hollywood-2009 & 1M & 56M & \textbf{0.856} &\textbf{0.825} & 3.263 & \textbf{0.795} & 3.36  & 0.854 & \textbf{0.818} & 3.315 & \textbf{0.774} & 3.35  & \textbf{0.855} &\textbf{0.814} & 3.364 & \textbf{0.782} & 3.174  & \textbf{0.867} &\textbf{0.835} & 3.401 & \textbf{0.795} & 3.15 \\ 
sc-ldoor & 952K & 20M & \textbf{21.676} &64.183 & 47.185 & - & 494.426  & \textbf{143.497} &1669.188 & 894.263 & - & 2290.9  & \textbf{361.145} &3532.593 & 3093.435 & - & 4353.16  & \textbf{4114.424} &- & - & - & - \\ 
soc-digg & 770K & 5M & \textbf{25.555} &38.424 & 43.787 & 1144.44 & 47.122  & \textbf{33.257} &177.117 & 395.428 & 1885.92 & 61.119  & \textbf{41.287} &8896.453 & 852.619 & - & 266.135  & \textbf{64.569} &- & 9918.25 & - & 3046.37 \\ 
ca-coauthors-dblp & 540K & 15M & 0.24 & \textbf{0.205} & 0.815 & \textbf{0.224} & 0.839  & 0.319 & \textbf{0.245} & 0.845 & 0.315 & 3.514  & 0.331 & \textbf{0.258} & 0.874 & 0.381 & 5.966  & 0.334 & \textbf{0.269} & 0.893 & 0.453 & 5.978 \\ 
soc-delicious & 536K & 1M & 0.075 & \textbf{0.068} & 0.437 & 0.151 & 0.169  & \textbf{0.113} &0.135 & 0.417 & 0.295 & 0.301  & \textbf{0.179} &0.391 & 0.587 & 0.686 & 0.474  & \textbf{0.232} &4.095 & 1.537 & 2.157 & 1.037 \\ 
web-it-2004 & 509K & 7M & 0.163 & \textbf{0.14} & 0.42 & 0.4 & 6.159  & 0.231 & \textbf{0.202} & 0.579 & 0.704 & 16.653  & 0.237 & \textbf{0.214} & 0.603 & 0.709 & 16.907  & \textbf{0.243} &\textbf{0.225} & 0.607 & 0.736 & 16.821 \\ 
soc-youtube & 495K & 1M & \textbf{0.418} &1.619 & 1.28 & 1.953 & 1.123  & \textbf{1.082} &31.709 & 16.115 & 17.539 & 4.14  & \textbf{7.795} &758.326 & - & 509.977 & 41.827  & \textbf{102.463} &- & - & - & 674.757 \\ 
sc-msdoor & 415K & 9M & \textbf{10.726} &31.332 & 24.22 & 1343.813 & 241.336  & \textbf{81.904} &874.654 & 647.598 & 2487.434 & 1164.13  & \textbf{199.684} &1897.013 & 2263.7 & - & 2211.78  & \textbf{2360.887} &- & - & - & - \\ 
ca-MathSciNet & 332K & 820K & 0.044 & \textbf{0.032} & 0.243 & 0.036 & 0.103  & 0.045 & \textbf{0.032} & 0.283 & 0.049 & 0.117  & 0.051 & \textbf{0.038} & 0.309 & 0.045 & 0.148  & 0.098 & \textbf{0.08} & 0.675 & 0.444 & 0.772 \\ 
ca-dblp-2012 & 317K & 1M & 0.061 & \textbf{0.043} & 0.246 & \textbf{0.043} & 0.108  & 0.055 & \textbf{0.045} & 0.244 & \textbf{0.047} & 0.096  & 0.071 & \textbf{0.047} & 0.275 & 0.053 & 0.185  & 0.062 & \textbf{0.037} & 0.236 & 0.054 & 0.215 \\ 
ca-citeseer & 227K & 814K & 0.034 & \textbf{0.026} & 0.175 & 0.029 & 0.07  & 0.03 & \textbf{0.021} & 0.156 & 0.027 & 0.07  & 0.025 & \textbf{0.021} & 0.172 & 0.031 & 0.084  & \textbf{0.029} &\textbf{0.028} & 0.152 & 0.035 & 0.125 \\ 
ca-dblp-2010 & 226K & 716K & 0.023 & \textbf{0.018} & 0.138 & 0.021 & 0.067  & 0.022 & \textbf{0.019} & 0.136 & \textbf{0.02} & 0.058  & \textbf{0.023} &\textbf{0.025} & 0.147 & \textbf{0.023} & 0.103  & 0.025 & \textbf{0.019} & 0.142 & 0.028 & 0.133 \\ 
sc-pwtk & 217K & 5M & \textbf{1.679} &2.065 & 3.675 & 420.808 & 72.627  & \textbf{5.801} &18.982 & 12.244 & 589.325 & 136.318  & \textbf{32.956} &1544.762 & 180.301 & 1010.477 & 787.769  & \textbf{133.684} &1655.622 & 451.792 & 1657.136 & 1771.97 \\ 
soc-gowalla & 196K & 950K & \textbf{0.182} &0.219 & 0.314 & 0.241 & \textbf{0.176}  & \textbf{0.357} &1.409 & 0.406 & 0.489 & \textbf{0.352}  & \textbf{0.748} &67.38 & 0.831 & 1.937 & \textbf{0.793}  & 6.406 & 2115.989 & 3.487 & 17.82 & \textbf{2.709} \\ 
tech-RL-caida & 190K & 607K & \textbf{0.14} &\textbf{0.131} & 0.176 & 0.178 & 0.17  & \textbf{0.09} &0.179 & 0.188 & 0.361 & 0.339  & \textbf{0.122} &1.873 & 1.075 & 1.282 & 0.833  & \textbf{0.294} &14.532 & - & 10.435 & 2.249 \\ 
sc-shipsec5 & 179K & 2M & \textbf{0.176} &\textbf{0.161} & 0.44 & 0.693 & 1.665  & \textbf{0.258} &0.307 & 0.68 & 1.826 & 2.867  & \textbf{0.884} &6.494 & 2.288 & 15.372 & 10.85  & \textbf{2.344} &14.16 & 6.433 & 107.36 & 25.538 \\ 
web-arabic-2005 & 163K & 1M & \textbf{0.02} &\textbf{0.022} & 0.12 & 0.023 & 0.088  & \textbf{0.02} &0.023 & 0.119 & 0.024 & 0.088  & \textbf{0.026} &\textbf{0.027} & 0.123 & \textbf{0.028} & 0.145  & \textbf{0.026} &\textbf{0.024} & 0.121 & \textbf{0.024} & 0.146 \\ 
soc-douban & 154K & 327K & \textbf{0.03} &\textbf{0.029} & 0.117 & 0.109 & 0.09  & \textbf{0.029} &\textbf{0.029} & 4318.654 & 0.14 & 0.234  & \textbf{7.509} &39.523 & - & \textbf{7.028} & 15.019  & 51.643 & - & - & - & \textbf{5.31} \\ 
sc-shipsec1 & 140K & 1M & 0.14 & \textbf{0.096} & 0.216 & 0.13 & 0.337  & \textbf{0.18} &\textbf{0.169} & 0.624 & 0.927 & 1.812  & \textbf{0.691} &\textbf{0.721} & 1.677 & 3.125 & 5.079  & \textbf{1.622} &6.267 & 28.64 & 34.677 & 17.956 \\ 
web-uk-2005 & 129K & 11M & 0.252 & \textbf{0.213} & 0.472 & 0.65 & 17.792  & 0.304 & \textbf{0.268} & 0.575 & 0.839 & 25.944  & 0.316 & \textbf{0.285} & 0.595 & 1.002 & 26.144  & \textbf{0.331} &\textbf{0.302} & 0.616 & 1.088 & 26.286 \\ 
web-sk-2005 & 121K & 334K & \textbf{0.01} &\textbf{0.01} & 0.055 & \textbf{0.011} & 0.132  & \textbf{0.012} &\textbf{0.011} & 0.066 & 0.014 & 0.208  & \textbf{0.01} &\textbf{0.011} & 0.056 & 0.015 & 0.213  & 0.016 & \textbf{0.012} & 0.067 & 0.016 & 0.217 \\ 
\hline

\end{tabular}
}
\end{table*}

We extend our analysis to the time-efficiency trade-off by plotting the number of solved instances against the runtime budget (ranging from 1 second to 3 hours) in Figure~\ref{fig:realworld-full}. Across all tested values of $k$, \texttt{BBRes} maintains a dominant position, indicated by its curve consistently lying above those of the competitors (\texttt{kDC2}, \texttt{MDC}, \texttt{DnBk}, and \texttt{WODC}). Two key observations highlight the efficacy of \texttt{BBRes}: First, \texttt{BBRes} excels in rapid convergence. For hard instances where $k=20$, \texttt{BBRes} quickly solves approximately 90 instances in the first second, whereas the baselines struggle to solve 65. Even with a generous 3-hour limit, \texttt{BBRes} reaches a total of 129 solved instances, surpassing the best baseline (\texttt{WODC}, 127 instances). Second, our method demonstrates superior robustness to problem difficulty. While the performance of baseline algorithms drops sharply as $k$ increases (reflecting the exponential growth in search space), \texttt{BBRes} sustains a high completion rate. This suggests that our double-coloring bound effectively prunes the search space, mitigating the impact of increasing $k$.

\begin{table*}
\vspace{-0.1in}
\caption{Running time (in seconds) comparison of \texttt{BBRes} and its variants on 30 benchmark graphs.}
\vspace{-0.15in}
\label{tab:ablation}
\centering
\resizebox{1.05\textwidth}{!}{
\begin{tabular}{|l|r|r|cccc|cccc|cccc|cccc|}
\hline
\multirow{2}{*}{Graphs} & \multirow{2}{*}{$n$} & \multirow{2}{*}{$m$} & \multicolumn{4}{c|}{$k = 5$} & \multicolumn{4}{c|}{$k = 10$} & \multicolumn{4}{c|}{$k = 15$} & \multicolumn{4}{c|}{$k = 20$} \\ \cline{4-19}
 & & & \texttt{BBRes} & -FlowUB & -Branch & Color & \texttt{BBRes} & -FlowUB & -Branch & Color & \texttt{BBRes} & -FlowUB & -Branch & Color & \texttt{BBRes} & -FlowUB & -Branch & Color \\ \hline
socfb-A-anon & 3M & 23M & 6.510 & 6.570 & 5.990 & \textbf{3.170} & \textbf{7.790} & 7.960 & \textbf{7.170} & 32.120 & \textbf{10.689} & 11.151 & \textbf{9.888} & 261.946 & \textbf{17.136} & 18.292 & \textbf{16.070} & 921.318 \\
soc-orkut & 2M & 106M & \textbf{104.263} & \textbf{105.429} & \textbf{97.701} & 184.354 & \textbf{112.963} & 122.661 & \textbf{105.647} & 867.689 & \textbf{129.552} & 175.951 & \textbf{119.789} & - & \textbf{146.826} & 314.638 & \textbf{153.571} & - \\
socfb-B-anon & 2M & 20M & \textbf{10.165} & \textbf{10.030} & \textbf{10.589} & \textbf{10.925} & \textbf{12.272} & \textbf{12.914} & \textbf{13.226} & 61.008 & 25.553 & \textbf{22.276} & \textbf{22.370} & 1680.934 & \textbf{19.994} & 22.902 & \textbf{20.760} & 8305.512 \\
soc-flixster & 2M & 7M & \textbf{0.521} & 0.990 & \textbf{0.538} & 5.636 & \textbf{0.594} & 1.670 & \textbf{0.616} & 52.162 & \textbf{0.788} & 4.255 & 1.056 & 1393.552 & \textbf{1.677} & 24.123 & 11.537 & - \\
web-wikipedia2009 & 1M & 4M & \textbf{0.840} & \textbf{0.842} & \textbf{0.837} & \textbf{0.905} & \textbf{0.954} & \textbf{0.985} & \textbf{0.971} & 1.187 & \textbf{0.973} & \textbf{0.972} & \textbf{1.013} & 1.217 & \textbf{1.660} & \textbf{1.564} & \textbf{1.685} & 1.831 \\
tech-as-skitter & 1M & 11M & 0.707 & 0.618 & 0.719 & \textbf{0.562} & 0.755 & \textbf{0.660} & 0.802 & \textbf{0.639} & 0.886 & \textbf{0.797} & 0.901 & 1.395 & \textbf{1.040} & \textbf{1.132} & \textbf{1.069} & 19.536 \\
soc-pokec & 1M & 22M & \textbf{5.996} & 6.228 & \textbf{5.890} & \textbf{5.457} & \textbf{6.588} & \textbf{6.962} & \textbf{6.500} & \textbf{6.650} & \textbf{9.899} & \textbf{10.418} & \textbf{9.763} & 12.063 & \textbf{13.476} & 14.673 & \textbf{12.953} & 63.050 \\
soc-lastfm & 1M & 4M & \textbf{3.392} & 13.851 & \textbf{3.475} & 117.415 & \textbf{13.906} & 63.229 & 17.442 & 2923.227 & \textbf{64.639} & 237.258 & 145.164 & - & \textbf{529.874} & 1439.551 & 834.218 & - \\
soc-youtube-snap & 1M & 2M & \textbf{0.677} & 1.588 & \textbf{0.671} & 4.083 & \textbf{1.705} & 6.757 & \textbf{1.699} & 60.676 & \textbf{13.467} & 50.892 & 15.248 & 1407.650 & \textbf{193.191} & 464.716 & 222.509 & - \\
ca-hollywood-2009 & 1M & 56M & \textbf{0.856} & \textbf{0.860} & \textbf{0.849} & \textbf{0.912} & \textbf{0.854} & \textbf{0.849} & \textbf{0.846} & \textbf{0.888} & \textbf{0.855} & \textbf{0.856} & \textbf{0.851} & \textbf{0.879} & \textbf{0.867} & \textbf{0.845} & \textbf{0.866} & \textbf{0.882} \\
sc-ldoor & 952K & 20M & \textbf{21.676} & 23.890 & \textbf{21.619} & 46.003 & 143.497 & \textbf{125.057} & 173.210 & 1869.012 & \textbf{361.145} & \textbf{377.997} & 1573.838 & 4111.127 & \textbf{4114.424} & \textbf{3976.471} & - & - \\
soc-digg & 770K & 5M & \textbf{25.555} & 29.010 & \textbf{25.245} & 39.611 & \textbf{33.257} & 41.776 & \textbf{33.006} & 210.258 & \textbf{41.287} & 65.069 & \textbf{41.346} & 9546.118 & \textbf{64.569} & 287.370 & 107.880 & - \\
ca-coauthors-dblp & 540K & 15M & \textbf{0.240} & \textbf{0.228} & \textbf{0.234} & \textbf{0.236} & \textbf{0.319} & \textbf{0.303} & \textbf{0.314} & 0.398 & \textbf{0.331} & \textbf{0.322} & \textbf{0.327} & 0.460 & \textbf{0.334} & \textbf{0.325} & \textbf{0.326} & 0.477 \\
soc-delicious & 536K & 1M & \textbf{0.075} & 0.087 & 0.083 & 0.088 & 0.113 & \textbf{0.098} & \textbf{0.091} & 0.107 & 0.179 & \textbf{0.149} & \textbf{0.148} & 0.365 & \textbf{0.232} & 0.283 & \textbf{0.227} & 3.991 \\
web-it-2004 & 509K & 7M & \textbf{0.163} & \textbf{0.158} & \textbf{0.166} & 0.356 & \textbf{0.231} & \textbf{0.220} & \textbf{0.227} & 0.731 & \textbf{0.237} & \textbf{0.229} & \textbf{0.238} & 0.571 & \textbf{0.243} & \textbf{0.238} & \textbf{0.248} & 0.587 \\
soc-youtube & 495K & 1M & \textbf{0.418} & 0.795 & \textbf{0.404} & 1.678 & \textbf{1.082} & 3.806 & \textbf{1.069} & 33.877 & \textbf{7.795} & 30.482 & \textbf{8.435} & 781.582 & \textbf{102.463} & 258.912 & 114.663 & - \\
sc-msdoor & 415K & 9M & \textbf{10.726} & \textbf{11.456} & \textbf{11.095} & 22.242 & 81.904 & \textbf{66.695} & 103.154 & 988.826 & \textbf{199.684} & \textbf{198.704} & 985.887 & 2179.247 & 2360.887 & \textbf{2094.814} & - & - \\
ca-MathSciNet & 332K & 820K & \textbf{0.044} & 0.047 & \textbf{0.042} & \textbf{0.046} & \textbf{0.045} & 0.058 & 0.064 & 0.058 & \textbf{0.051} & \textbf{0.052} & 0.063 & 0.060 & \textbf{0.098} & 0.118 & \textbf{0.107} & 0.143 \\
ca-dblp-2012 & 317K & 1M & 0.061 & \textbf{0.051} & \textbf{0.051} & 0.059 & \textbf{0.055} & \textbf{0.050} & \textbf{0.054} & 0.070 & 0.071 & \textbf{0.058} & \textbf{0.056} & 0.079 & 0.062 & \textbf{0.055} & \textbf{0.055} & 0.086 \\
ca-citeseer & 227K & 814K & 0.034 & \textbf{0.027} & 0.032 & 0.034 & \textbf{0.030} & \textbf{0.032} & 0.037 & 0.043 & \textbf{0.025} & \textbf{0.027} & 0.036 & 0.033 & \textbf{0.029} & 0.036 & 0.040 & 0.036 \\
ca-dblp-2010 & 226K & 716K & \textbf{0.023} & \textbf{0.022} & 0.025 & 0.028 & \textbf{0.022} & \textbf{0.022} & 0.027 & 0.026 & \textbf{0.023} & \textbf{0.025} & 0.028 & 0.028 & \textbf{0.025} & \textbf{0.024} & 0.032 & \textbf{0.024} \\
sc-pwtk & 217K & 5M & \textbf{1.679} & \textbf{1.684} & \textbf{1.593} & 4.598 & 5.801 & \textbf{4.756} & 7.262 & 14.455 & \textbf{32.956} & 56.107 & 37.546 & 1462.455 & 133.684 & 176.029 & \textbf{119.104} & 1619.158 \\
soc-gowalla & 196K & 950K & \textbf{0.182} & 0.231 & \textbf{0.178} & 0.304 & 0.357 & 0.312 & \textbf{0.273} & 1.709 & \textbf{0.748} & 1.821 & \textbf{0.755} & 71.508 & \textbf{6.406} & 19.199 & 8.889 & 2162.518 \\
tech-RL-caida & 190K & 607K & 0.140 & 0.132 & \textbf{0.109} & 0.168 & \textbf{0.090} & 0.103 & \textbf{0.087} & 0.199 & \textbf{0.122} & 0.199 & 0.150 & 1.866 & \textbf{0.294} & 0.499 & \textbf{0.294} & 15.755 \\
sc-shipsec5 & 179K & 2M & \textbf{0.176} & \textbf{0.177} & \textbf{0.174} & \textbf{0.186} & \textbf{0.258} & \textbf{0.245} & \textbf{0.257} & 0.275 & \textbf{0.884} & 1.134 & \textbf{0.886} & 6.599 & \textbf{2.344} & 2.583 & \textbf{2.514} & 14.227 \\
web-arabic-2005 & 163K & 1M & \textbf{0.020} & \textbf{0.019} & 0.024 & 0.021 & \textbf{0.020} & 0.026 & \textbf{0.019} & \textbf{0.020} & \textbf{0.026} & \textbf{0.025} & \textbf{0.025} & \textbf{0.025} & \textbf{0.026} & 0.031 & 0.030 & 0.040 \\
soc-douban & 154K & 327K & \textbf{0.030} & \textbf{0.031} & \textbf{0.032} & 0.043 & \textbf{0.029} & \textbf{0.030} & 0.038 & 0.045 & \textbf{7.509} & \textbf{7.214} & \textbf{7.268} & 27.132 & \textbf{51.643} & \textbf{49.694} & \textbf{50.022} & - \\
sc-shipsec1 & 140K & 1M & 0.140 & 0.110 & 0.134 & \textbf{0.070} & \textbf{0.180} & \textbf{0.180} & 0.222 & 0.262 & 0.691 & \textbf{0.533} & 0.682 & 0.628 & 1.622 & \textbf{1.325} & \textbf{1.293} & 6.719 \\
web-uk-2005 & 129K & 11M & \textbf{0.252} & \textbf{0.246} & \textbf{0.255} & 0.679 & \textbf{0.304} & \textbf{0.296} & \textbf{0.306} & 0.872 & \textbf{0.316} & \textbf{0.310} & \textbf{0.320} & 0.888 & \textbf{0.331} & \textbf{0.322} & \textbf{0.332} & 0.876 \\
web-sk-2005 & 121K & 334K & \textbf{0.010} & \textbf{0.010} & 0.013 & 0.014 & \textbf{0.012} & \textbf{0.012} & \textbf{0.011} & 0.016 & \textbf{0.010} & \textbf{0.010} & 0.014 & 0.016 & 0.016 & \textbf{0.011} & \textbf{0.012} & 0.014 \\ \hline
\end{tabular}
}
\end{table*}

\noindent \textbf{Runtime efficiency on representative benchmarks}.
We evaluate the runtime performance of \texttt{BBRes} against four baselines, namely \texttt{kDC2}, \texttt{MDC}, \texttt{DnBk}, and \texttt{WODC}, on representative graphs with $k \in \{5,10,15,20\}$. 
The dataset consists of the 30 largest graphs from the real-world collection, excluding instances where all algorithms exceeded the 3-hour time limit at $k=20$. Table~\ref{tab:representative} reports the detailed running times. The results demonstrate that \texttt{BBRes} consistently delivers superior performance across all tested values of $k$, except for a few trivial instances. 
Notably, on the \texttt{soc-flixster} graph with $k=20$, \texttt{BBRes} achieves a remarkable ${53\times}$ speedup over the best-performing baseline (\texttt{WODC}, 89.1s vs. 1.7s). \revision{In addition, we can observe that \texttt{BBRes} achieves at least a $2\times$ speedup over the best competing method on 26.7\%, 20.0\%, 40.0\%, and 53.3\% of the test cases for $k=5$, $10$, $15$, and $20$, respectively.}

Furthermore, \texttt{BBRes} exhibits significant scalability and robustness with respect to increasing $k$. In contrast to the baselines, whose runtimes often degrade exponentially as $k$ grows, \texttt{BBRes} remains highly efficient. For instance, on \texttt{socfb-A-anon}, as $k$ increases from 10 to 20, the runtime of \texttt{BBRes} grows moderately from 7.8s to 17.1s. In stark contrast, the competitors suffer from severe performance degradation at $k=20$: \texttt{MDC}, \texttt{kDC2}, and \texttt{DnBk} require 140.9s, 988.4s, and 1304.3s, respectively. Even the runner-up, \texttt{WODC}, is significantly slower, taking 57.3s, which is more than $3\times$ that of \texttt{BBRes}. These results strongly validate the practical efficiency of our proposed algorithm in handling difficult constraints. %
\revision{Finally, we note that \texttt{BBRes} runs slightly slower than the baselines on some instances in Table~\ref{tab:representative}. This mainly occurs on relatively easy instances where all methods terminate within a fraction of a second. In such cases, the computational overhead of \texttt{BBRes}'s advanced techniques (e.g., \textbf{UB-Double}) may outweigh their pruning benefits.}

\subsection{Ablation Studies}
To rigorously assess the individual contributions of our proposed techniques, we conduct an ablation study by comparing \texttt{BBRes} against two distinct variants. Specifically, we define \texttt{-FlowUB} as the variant that replaces our double-coloring upper bound (Section~\ref{sec:ub}) with a standard single-coloring bound, thereby isolating the impact of our novel bounding strategy. Meanwhile, we define \texttt{-Branch} as the variant that disables our specific branching and early termination strategies (Section~\ref{sec:BBRes}). We note that \texttt{-Branch} adopts the state-of-the-art branching strategy proposed in~\cite{dai2024theoretically} for fair comparison, and it differs with \texttt{MDC} mainly in the upper bound. This comparison allows us to quantify the performance gain attributed solely to our optimized branching and early termination strategies.

Table~\ref{tab:ablation} details the runtime performance of these methods across 30 representative benchmark graphs. The results consistently demonstrate that the full \texttt{BBRes} algorithm achieves superior performance compared to both variants in the vast majority of cases.

To quantify the impact of the double-coloring upper bound, we observe that \texttt{-FlowUB} suffers from significant performance degradation, particularly on large, dense graphs with higher $k$. For instance, on \texttt{soc-lastfm} with $k=20$, the runtime surges from $527.9$ seconds (\texttt{BBRes}) to $1441.6$ seconds (\texttt{-FlowUB}), representing a nearly $2.7\times$ slowdown. Similarly, on \texttt{soc-orkut} ($k=20$), \texttt{-FlowUB} requires $314.8$ seconds compared to $137.4$ seconds for \texttt{BBRes}. These results indicate that our double-coloring strategy provides a significantly tighter bound than the single-coloring approach, effectively pruning massive branches of the search tree that standard bounds fail to detect.
Regarding the impact of the branching strategy, the comparison between \texttt{BBRes} and \texttt{-Branch} highlights the efficiency of our branching and termination rules. \revision{Although \texttt{-Branch} is often competitive with \texttt{BBRes}, this does not imply that the performance gain of \texttt{BBRes} mainly comes from the tighter upper bound alone. Indeed, \texttt{-FlowUB} is also competitive with \texttt{BBRes} on several challenging instances, indicating that the branching strategy also makes a substantial contribution.} \revision{While \texttt{-Branch} often performs better than \texttt{-FlowUB} on many easier instances, suggesting that the tighter upper bound is particularly effective at pruning unpromising branches early,} it still lags behind the full \texttt{BBRes}. For example, on \texttt{sc-msdoor} at $k=15$, \texttt{-Branch} is approximately $1.5\times$ slower than \texttt{BBRes} ($292.3$ seconds vs. $197.8$ seconds). \revision{On the other hand, on more challenging instances such as \texttt{sc-msdoor} and \texttt{sc-ldoor}, \texttt{-FlowUB} often outperforms \texttt{-Branch}, suggesting that the improved branching strategy is particularly effective in reducing the search complexity.} This confirms that our tailored branching strategy effectively reduces the search effort and facilitates faster convergence.
\revision{We also note that the efficiency of \texttt{BBRes} is closely related to the value of $\lambda_k$. As shown in Table~\ref{tab:representative}, the running time of \texttt{BBRes} generally increases with $k$, which is consistent with the complexity $O^*(\lambda_k^n)$ and the fact that $\lambda_k$ increases as $k$ grows. At the same time, the theoretical gap $2-\lambda_k$ (or $\gamma_k-\lambda_k$) becomes smaller as $k$ increases, indicating that the advantage over the naive $O^*(2^n)$ bound and the previous $O^*(\gamma_k^n)$ bound gradually narrows. Nevertheless, even when $\lambda_k$ is very close to $2$ (e.g., at $k=20$), the improvement can still be meaningful on very large graphs. For example, for a graph with $10^7$ vertices, we have $2^{10^7}/\lambda_{20}^{10^7}\approx 54$ and $\gamma_{20}^{10^7}/\lambda_{20}^{10^7}\approx 2$.}
In summary, \texttt{BBRes} consistently achieves the best overall performance by combining tight bounding with effective branching, \revision{suggesting that the two techniques are complementary in practice even though their relative contributions vary across instances}.

\subsection{Tests on Dense and Synthetic Graphs}
\label{sec:exp_dense_summary}

While the previous sections focused on sparse real-world graphs, the maximum $k$-defective clique problem is theoretically most challenging on \emph{dense} graphs. 
To stress-test \texttt{BBRes} under these extreme conditions, we extended our evaluation to the DIMACS benchmark (notoriously hard dense instances) and LFR synthetic graphs~\cite{lancichinetti2008benchmark} (controlled density analysis). Due to space constraints, the detailed DIMACS results and the complete density sensitivity analysis on LFR graphs are deferred to Appendix~\ref{app:exp} of the technical report~\cite{appendix}.

The results on DIMACS graphs show that \texttt{BBRes} remains highly competitive on dense benchmark instances and continues to outperform the baselines on most tested values of $k$. \revision{On the LFR graphs, we examine how graph density affects the relative behavior of \texttt{-Branch} and \texttt{-FlowUB}. We observe that as the graph density increases, the runtime gap between these two variants gradually narrows. This suggests that on denser graphs, the proposed branching strategy becomes increasingly relevant, and its contribution becomes more comparable to that of the flow-based upper bound. At the same time, \texttt{BBRes}, which combines both components, consistently achieves the best overall performance, with a maximum speedup of up to $174\times$ over the ablation variants. These results indicate that neither component alone is sufficient to fully capture the benefit on dense graphs; instead, their combination is important for maintaining strong performance in such challenging settings.}

%% file: 7-Related.tex
\section{Related Work}
\label{sec:related}
The concept of the $k$-defective clique was originally formalized by~\cite{yu2006predicting}. Early exact methods for the maximum $k$-defective clique problem, such as those in~\cite{gschwind2018maximum,gschwind2021branch,TBBB13}, suffered from limited scalability and were practical only on small graphs. The first algorithm capable of handling large graphs was \texttt{MADEC$^+$}~\cite{chen2021computing}. Building on this foundation, \texttt{KDBB}~\cite{gao2022exact} and \texttt{KD-Club}~\cite{jin2024kdclub} improved empirical performance by combining preprocessing with several (upper-bound-based) pruning rules. 
Recently, several algorithms, including \texttt{kDC}~\cite{chang2023kdefective}, \texttt{MDC}~\cite{dai2024theoretically}, \texttt{kDC2}~\cite{Chang24kDC-2}, \texttt{DnBK}~\cite{Luo24defective}, and \texttt{WODC}~\cite{jang2025efficient}, have introduced new branching strategies, thereby improving the theoretical time complexity of the problem. Furthermore, these methods introduce novel reduction rules and upper bound-based pruning rules to effectively eliminate unnecessary branches and vertices during the branch-and-bound search phase. In addition, they incorporate new heuristics to efficiently compute an initial solution in the preprocessing phase.
We note that the most recent methods, namely \texttt{kDC2}~\cite{Chang24kDC-2}, \texttt{MDC}~\cite{dai2024theoretically}, \texttt{DnBK}~\cite{Luo24defective}, and \texttt{WODC}~\cite{jang2025efficient} are included in our experiments for comparison. From a theoretical perspective, \texttt{MADEC$^+$} runs in $O^*(\beta_{2k}^{n})$ time, while \texttt{kDC} further reduces the complexity to $O^*(\beta_k^{n})$, where $\beta_i$ is the largest real root of the equation $x^{i+3} - 2x^{i+2} + 1 = 0$. The \texttt{kDC2} algorithm, proposed by Chang~\cite{Chang24kDC-2}, runs in $O^*(\beta_{k-1}^{n})$ time, and the \texttt{MDC} algorithm, proposed by Dai et al.~\cite{dai2024theoretically}, runs in $O^*(\gamma_{k}^{n})$ time, where $\gamma_i$ is the largest real root of the equation $x^{i+3} - 2x^{i+2} + x^2 - x + 1 = 0$. Importantly, we note that $\gamma_k < \beta_{k-1}$ for $k \geq 2$. 
Distinct from the above BB frameworks, \texttt{DnBK} and \texttt{WODC} employ alternative frameworks to reduce the exponential base. \texttt{DnBK} reduces the problem of finding the largest $k$-defective clique to $O(n({\Delta\delta})^{2k})$ instances of finding the maximum clique and runs in $O(n({\Delta\delta})^{2k} \alpha_c^n)$ time, where $\alpha_c$ is the base factor in the time complexity of the maximum clique algorithm. While theoretical bounds suggest $\alpha_c = 1.2$~\cite{xiao2017clique}, practical implementations rely on solvers like those in~\cite{chang2020clique}, where $\alpha_c \approx 1.33$ \footnote{As \cite{chang2020clique} omits the time complexity analysis for the maximum clique algorithm \texttt{MC-BRB}, we provide it here. \texttt{MC-BRB} recursively removes vertices with at most three non-neighbors in the candidate set. Thus, \texttt{MC-BRB} only branches on vertices with at least four non-neighbors in the candidate set. This leads to the characteristic equation $x^5 - x^4 - 1 = 0$, yielding a branching factor of $\alpha_c \approx 1.33$.}. Similarly, Jang et al.~\cite{jang2025efficient} introduced \texttt{WODC}, which decomposes the branching process into a maximal clique enumeration stage and a $k$-defective clique constraint verification stage. Consequently, \texttt{WODC} has the worst-case time complexity of $O(n(\Delta\delta)^{k}\alpha_e^{\delta})$, where $\alpha_e = 3^{1/3} \approx 1.44$ represents the optimal base for maximal clique enumeration. While \texttt{DnBK} and \texttt{WODC} achieve smaller exponential bases, they introduce significant overhead factors of $(\Delta\delta)^{\Theta(k)}$. These factors become prohibitive when $k$ is not a small constant, limiting their scalability in general settings. 
In this paper, our \texttt{BBRes} achieves the time complexity of $O^*(\lambda_k^n) $, where $\lambda_k$ is the largest real root of the equation $x^{k+4}-2x^{k+3}+x^{3}-x+1=0$ and $\lambda_k$ is strictly smaller than $\gamma_k$ for $k\geq 1$. Thus, our proposed \texttt{BBRes} achieves better theoretical complexity for general $k$ (noting that the complexity of \texttt{DnBK} and \texttt{WODC} depends exponentially on $k$), and empirically outperforms the SOTA methods, namely, \texttt{kDC2}, \texttt{MDC}, \texttt{DnBK}, and \texttt{WODC}, in Section~\ref{sec:exp}.

In addition, the problem of enumerating all maximal $k$-defective cliques has been investigated. The \texttt{Pivot+} algorithm~\cite{dai2023maximal} achieves $O^*(\gamma_k^{n})$ time. However, \texttt{Pivot+} is typically less effective for identifying the maximum $k$-defective clique in practice due to limited pruning techniques. Furthermore, the concept of defective cliques has been extended to bipartite graphs~\cite{wang2025identifying,cui2025efficient}.

%% file: 8-Conclusion.tex
\section{Conclusion}\label{sec:conclu}
In this paper, we proposed a new branch-and-bound method \texttt{BBRes} for finding the maximum $k$-defective clique. \texttt{BBRes} has the worst-case time complexity of $O^*(\lambda_k^n)$, which improves upon the state-of-the-art time complexity of $O^*(\gamma_k^n)$, where $\lambda_k < \gamma_k$. To boost the practical efficiency, we also proposed a tighter upper bound. Finally, extensive experiments demonstrate that \texttt{BBRes} outperforms state-of-the-art algorithms by several orders of magnitude in practice. In the future, we intend to investigate parallelization strategies for \texttt{BBRes} to handle massive graphs.

%% file: 9-Appendix.tex
\appendix
\clearpage
\nobalance
\section{Omitted Proofs}\label{app:proofs}

\begin{figure}[b]

\usetikzlibrary{positioning,fit,backgrounds,calc,shapes.geometric}

\begin{tikzpicture}[
    node distance=10mm and 8mm,
    every node/.style={draw, shape=ellipse, minimum width=10mm, minimum height=6mm, inner sep=2pt, font=\small},
    every edge quotes/.style={auto, font=\small}
    every label/.style={font=\small, draw=none, fill=none}
]

\node (I0) {$B_0$};
\node (I1) [below left=of I0] {$B_1$};
\node (I2) [below left=of I1] {$B_2$};
\node (Iq-1) [below left= of I2] {$B_{q-1}$};
\node (Iq) [below left=of Iq-1] {$B_q$};

\node[fill=gray!20] (Iq1) [below right=of I0] {$B_{q+1}$};
\node[fill=gray!20] (Iq2) [below right=of I1] {$B_{q+2}$};
\node[fill=gray!20] (Iq3) [below right=of I2] {$B_{q+3}$};
\node[fill=gray!20] (I2q) [below right=of Iq-1] {$B_{2q}$};

\draw[-] (I0) -- node[draw=none,above left] {$+p_0$} (I1);
\draw[-] (I1) -- node[draw=none,above left] {$+p_1$} (I2);
\draw[dashed] (I2) -- node[draw=none,above left] {$\cdots$} (Iq-1);
\draw[-] (Iq-1) -- node[draw=none,above left] {$+p_{q-1}$} (Iq);

\draw[-] (I0) -- node[draw=none,,above right] {$-p_0$} (Iq1);
\draw[-] (I1) -- node[draw=none,,above right] {$-p_1$} (Iq2);
\draw[-] (I2) -- node[draw=none,,above right] {$-p_2$} (Iq3);
\draw[-] (Iq-1) -- node[draw=none,,above right] {$-p_{q-1}$} (I2q);

\end{tikzpicture}

\caption{Branching Tree in the Proof of Lemma~\ref{lemma:worst-case-TC}.}
\label{fig:branching-tree-proof}

\end{figure}
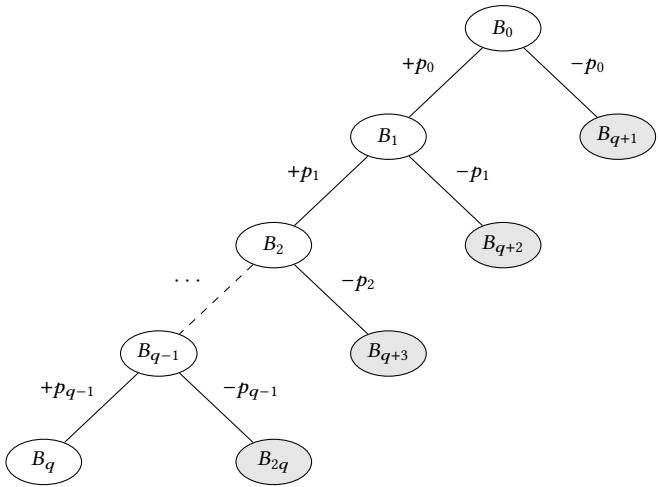

\subsection{Proof of Lemma~\ref{lemma:worst-case-TC}}
\begin{proof}[Proof of Lemma~\ref{lemma:worst-case-TC}]
We note that each recursion of \texttt{BBRes} runs in polynomial time $O(|C|^2)$, which is dominated by computing the upper bound and conducting the reductions (details are discussed in Section~\ref{sec:ub}). When the size of the maximum $k$-defective clique $|V(g^*)|\geq k+2$, we can utilize the diameter-two property, where the size of $G_{v_i}$ can be bounded by $\Delta\delta$~\cite{Chang24kDC-2}, implying $|C|\leq \Delta\delta$; otherwise, we have $|C|\leq n$. Below, we analyze the number of recursions (or branches). 

We analyze the branching process based on the recursion tree $\mathcal{T}$ illustrated in Figure~\ref{fig:branching-tree-proof}. Specifically, $\mathcal{T}$ is a binary tree where each node $B_i$ represents a state $(g_i, S_i, C_i)$. Recall that our branching strategy, \textbf{BS-three}, prioritizes selecting a pivot $p_i$ that has at least one non-neighbor in the current partial set $S_i$. Based on this, the \emph{left child} of $B_i$ is generated by adding $p_i$ to $S_i$ (i.e., $S_{i+1} \leftarrow S_i \cup \{p_i\}$), while the \emph{right child} corresponds to the branch where $p_i$ is excluded from $C_i$.

We focus our analysis on the \emph{leftmost path} $(B_0, B_1, \dots, B_q)$, formed by recursively traversing the left children. We initially restrict our analysis to the scenario where $S_0$ is non-empty, deferring the discussion of the case where $S_0 = \emptyset$ to the end of the proof.
We define the endpoint $B_q$ as either a leaf node or the first node along this path satisfying $|C_q| \leq |C_{q-1}| - 2$. This definition implies that, in the transition to $B_q$, either 1) at least one additional vertex (besides the pivot) is pruned from $C_{q-1}$ due to constraint violations (i.e., inability to form a $k$-defective clique with $S_q$); or 2) the remaining instance becomes tractable (i.e., it has fewer than three non-neighbors in $C_{q-1}$ after removing $p_{q-1}$), allowing \texttt{IRSolver} to directly obtain the final solution (i.e., terminate the branch). We first analyze the depth of this leftmost path and prove that:
\begin{equation}
    q \leq k.
\end{equation}

To prove this, let $B_x$ ($0\leq x\leq q$) denote the last node on the path such that the pivot $p_x$ is fully connected to $S_x$ (i.e., $p_x$ introduces no missing edges). First, consider the case where no such node exists. This implies that every pivot $p_i$ ($0 \leq i < q$) along the path contributes at least one non-edge to $S_{i+1}$. Therefore, the number of missing edges grows monotonically: $|\overline{E}(S_{i+1})| \geq |\overline{E}(S_{i})| + 1$. Summing over the path, we obtain $|\overline{E}(S_{q})| \geq |\overline{E}(S_{0})| + q$.  Given the feasibility constraint that $S_q$ is a $k$-defective clique (otherwise $B_q$ would not be a valid branch), i.e., $|\overline{E}(S_{q})| \leq k$, it follows that $q \leq k - |\overline{E}(S_{0})| \leq k$.
Next, we assume that such a node $B_x$ exists. We claim that the accumulated number of missing edges in $S_x$ satisfies $|\overline{E}(S_x)| > x$.
Recall that $B_q$ is defined as the first node where pruning occurs (i.e., $|C_q| \leq |C_{q-1}| - 2$). This implies that for any node $B_i$ prior to $B_q$ (specifically $0 \leq i < x$), no pruning occurred. Consequently, for each pivot $p_i \in \{p_0, \dots, p_{x-1}\}$, its non-neighbors from $C_i$ (at least three) must have been added to $S_x$. Since $B_x$ is a node where $p_x$ is fully connected to $S_x$, the candidate set $C_x$ is cleared of any non-neighbors of $S_x$. Therefore, all non-neighbors of the set $\{p_0, \dots, p_{x-1}\}$ must reside within $S_x$. In the complement graph of $S_x$, every vertex in $\{p_0, \dots, p_{x-1}\}$ has a degree of at least 3. By the Handshaking Lemma, the total number of missing edges is at least the sum of degrees divided by 2, which implies that $|\overline{E}(S_x)| \geq \frac{3x}{2} = 1.5x$. For $x > 0$, we have $1.5x \geq x + 0.5 > x$. We now address the boundary case $x=0$. If $|\overline{E}(S_0)| > 0$, the claim $|\overline{E}(S_x)| > x$ holds trivially since $|\overline{E}(S_0)| \ge 1 > 0$. If $|\overline{E}(S_0)| = 0$ (i.e., $S_0$ is a clique), we show that $x$ cannot be 0. Note that $S_0$ contains an inherited pivot $p'$ with at least three non-neighbors in $C_0$. If $x=0$, $B_0$ would be the fully connected node $B_x$, implying that $C_0$ contains no non-neighbors of $S_0$ (including $p'$), which is a contradiction. Thus, if $S_0$ is a clique, we must have $x \ge 1$, ensuring the bound holds.
Combining this with the monotonicity of missing edges after $B_x$, we derive $|\overline{E}(S_q)| \geq |\overline{E}(S_x)| + (q - x - 1) \geq (x + 1) + q - x - 1 = q$. Given the feasibility constraint $|\overline{E}(S_q)| \leq k$, it follows that $q \leq k$.

Let $\ell(B_0)$ denote the number of leaf nodes rooted at node $B_0$ with non-empty $S_0$. We then prove that $\ell(B_0)\leq \lambda_k^{|C_0|}$ by induction. For the base case where $B_0$ is a leaf node, w.l.o.g., we let $C_0$ be empty since the corresponding branch can be solved directly. We thus have $\ell(B_0)=1\leq \lambda_k^{|C_0|}$. For a non-leaf node $B_0$, we consider the path $(B_0,B_1,...,B_q)$ and trivially have $\ell(B_0)=\ell(B_q)+\ell(B_{q+1})+....+\ell(B_{2q})$, where $B_{q+i}$ ($1\leq i\leq q$) is the right child of node $B_{i-1}$ (See Figure~\ref{fig:branching-tree-proof}). We have: 
\begin{enumerate}[leftmargin=*]
    \item $q\leq k$ based on the above discussion;
    \item $|C_q|\leq |C_{q-1}|-2\leq |C_0|-q-1$ based on the definition of $B_q$;
    \item $|C_{q+i}|\leq |C_{i-1}|-1\leq |C_0|-i$ for $1\leq i\leq q$, based on the branching relationship between the parent node $B_{i-1}$ and its children $B_i$ and $B_{q+i}$.
\end{enumerate}
We note that when $q=k$, $\ell(B_0)$ reaches the largest number of leaf nodes, which corresponds to the worst case (we will discuss other cases $q<k$ later). In addition, when $q=k$, we note that $|\overline{E}(S_k)|=k$ based on the above discussion of $q\leq k$. Therefore, we have:
    \begin{equation}
        |C_k| \leq |C_{k-1}|-4\leq |C_0|-k-3.
    \end{equation}
This is because $p_{k-1}$ has three non-neighbors in $C_{k-1}$ and those non-neighbors can be pruned since each of them cannot form a $k$-defective clique with $S_{k}$. We thus have:
\begin{eqnarray}
    \ell(B_0)&=& \ell(B_q)+\ell(B_{q+1})+....+\ell(B_{2q})\\
    &\leq& \ell(B_k)+\ell(B_{k+1})+....+\ell(B_{2k})\\
    &\leq& \lambda_k^{|C_0|-k-3}+\lambda_k^{|C_0|-1}+...+\lambda_k^{|C_0|-k},
\end{eqnarray}
where the inequality of $\lambda_k^{|C_0|-k-3}+\lambda_k^{|C_0|-1}+...+\lambda_k^{|C_0|-k}\leq \lambda_k^{|C_0|}$ holds if $\lambda_k$ is the largest real root of the equation $x^{k+4}-2x^{k+3}+x^{3}-x+1=0$.

When $q<k$, the recurrence becomes:
    \begin{eqnarray}
        \ell(B_0)&=& \ell(B_q)+\ell(B_{q+1})+....+\ell(B_{2q})\\
        &\leq& \lambda_{k}'^{|C_0|-q-1} + \lambda_k'^{|C_0|-1} + ... + \lambda_k'^{|C_0|-q},
    \end{eqnarray}
where $\lambda_q'$ is the largest real root of the equation $x^{q+2} - 2x^{q+1} + 1 = 0$. Since $\lambda_q'$ demonstrably smaller than $\lambda_k$ for $q<k$, we adopt $\lambda_k$ as our worst-case time complexity factor.

In addition, we consider a general branch $B_0$ without the constraint $S_0\neq \emptyset$. Similar to the proof of Lemma 4.3 in~\cite{Chang24kDC-2}, we can prove that the number of leaf nodes satisfies $\ell(B_0)\leq \lambda_k^{|C_0|}$, where $|C_0| = |V(g_0)|$. If the size of the maximum $k$-defective clique is at least $k+2$, we can apply the diameter-two property to bound the initial candidate set by $|C_0| \le \Delta\delta$~\cite{Chang24kDC-2}, yielding a time complexity of $O^*(\lambda_k^{\Delta\delta})$. Otherwise, we trivially have $|C_0| \le n$, yielding $O^*(\lambda_k^n)$.

\revision{Finally, we show that $\lambda_k < \gamma_k$ for all $k\geq 1$. Recall that $\lambda_k$ is the largest real root of $x^{k+4}-2x^{k+3}+x^3-x+1=0$, and $\gamma_k$ is the largest real root of $x^{k+3}-2x^{k+2}+x^2-x+1=0$. Since both $\lambda_k$ and $\gamma_k$ are greater than 1, we can rewrite these equations as $1 = \sum_{i=1}^{k}\lambda_k^{-i} + \lambda_k^{-(k+3)}$ and $1 = \sum_{i=1}^{k}\gamma_k^{-i} + \gamma_k^{-(k+2)}$.
Define $f(x)=\sum_{i=1}^{k}x^{-i}+x^{-(k+3)}$.
For $x>1$, $f(x)$ is strictly decreasing. Moreover, by the above reformulation, we have $f(\lambda_k)=\sum_{i=1}^{k}\lambda_k^{-i}+\lambda_k^{-(k+3)}=1$. For $\gamma_k$, we have
\begin{align*}
    f(\gamma_k)
    &=\sum_{i=1}^{k}\gamma_k^{-i}+\gamma_k^{-(k+3)} \\
    &=\left(\sum_{i=1}^{k}\gamma_k^{-i}+\gamma_k^{-(k+2)}\right)-\gamma_k^{-(k+2)}+\gamma_k^{-(k+3)} \\
    &=1-\bigl(\gamma_k^{-(k+2)}-\gamma_k^{-(k+3)}\bigr) \\
    &<1.
\end{align*}
Hence, $f(\gamma_k)<f(\lambda_k)$. Since $f$ is strictly decreasing on $x>1$, it follows that $\gamma_k>\lambda_k$, i.e., $\lambda_k<\gamma_k$.}

\revision{
\smallskip
\noindent \textbf{Remark}. We now show that both $\lambda_k$ and $\gamma_k$ increase monotonically with $k$ and converge to $2$. We first show this for $\lambda_k$. Define $f_k(x) = x^{k+4} - 2x^{k+3} + x^3 - x + 1$. 
By definition of $\lambda_k$, $f_k(\lambda_k)=0$ and $f_{k+1}(\lambda_{k+1})=0$, where $\lambda_k,\lambda_{k+1}\in(1,2)$. Moreover, $f_{k+1}(x)-f_k(x)=x^{k+3}(x-2)(x-1)$.
Hence, for any $x\in(1,2)$, we have $f_{k+1}(x)-f_k(x)<0$, and in particular, $f_{k+1}(\lambda_k)<f_k(\lambda_k)=0$. Since $f_{k+1}(2)=7>0$ and $f_{k+1}(x)$ is continuous, the intermediate value theorem implies that $f_{k+1}$ has a root in $(\lambda_k,2)$. Therefore, $\lambda_k<\lambda_{k+1}<2$. Thus, $(\lambda_k)$ is monotonically increasing and bounded above by $2$.
To show that $\lambda_k\to 2$, we rewrite $f_k(\lambda_k)=0$ as $2-\lambda_k=\frac{\lambda_k^3-\lambda_k+1}{\lambda_k^{k+3}}$.
Since $1<\lambda_k<2$, we have $1<\lambda_k^3-\lambda_k+1<7$. Also, since $(\lambda_k)$ is increasing, $\lambda_k\ge \lambda_1>1$ for all $k$. Thus, $0\le 2-\lambda_k \le \frac{7}{\lambda_1^{k+3}}\to 0$. Hence, $\lambda_k\to 2$. By the same argument, one can show that $\gamma_k$ is also monotonically increasing and converges to $2$.

Finally, we have $\gamma_k-\lambda_k\to 0$ as $k\to\infty$, since both sequences converge to $2$: $\lim_{k\to\infty}(\gamma_k-\lambda_k)=\lim_{k\to\infty}\gamma_k-\lim_{k\to\infty}\lambda_k=2-2=0$.
We do not claim that the sequence $(\gamma_k-\lambda_k)$ is monotone; rather, our formal result is that this gap vanishes asymptotically.}
\end{proof}

\subsection{Proof of Lemma~\ref{lem:double-coloring-flow}}
\begin{proof}[Proof of Lemma~\ref{lem:double-coloring-flow}]
We first show that there is a one-to-one correspondence between the selected flow edges in the constrained maximum flow of $g^c$ and the set $D$.

Recall that each vertex $v \in C$ in $g$ corresponds to an edge $(col_1(v), col_2(v)) \in E(g^c)$ via the mapping $p: v \mapsto (col_1(v), col_2(v))$. The capacity of each edge between partitions $V_1'$ and $V_2'$ is set to $1$, so in the maximum flow, each such edge either carries $1$ unit of flow (selected) or $0$ (not selected). Define $E^c_{\mathrm{ans}}$ as the set of edges with flow value $1$; this set corresponds to a subset $D$ of $C$, with $|E^c_{\mathrm{ans}}| = flow(g^c)$. Let $f_{\text{out}}(v)$ and $f_{\text{in}}(v)$ be the outgoing and incoming flow at vertex $v$ in $g^c$. Since flow is conserved at each vertex except for $s$ and $t$, we have $f_{\text{in}}(v) = f_{\text{out}}(v)$. Thus, $|E^c_{\mathrm{ans}}| = \sum_{u \in V_1'} f_{\text{out}}(u) = \sum_{u \in V_1'} f_{\text{in}}(u) = f_{\text{out}}(s) = flow(g^c)$.
This establishes the mapping between the selected flow edges in the constrained maximum flow and the set $D$.

\smallskip
Now, we analyze the cost $cost(g^c)$ associated with the constrained maximum flow in $g^c$, which corresponds to the sum of the costs of all flow-carrying edges. We partition the total cost as $cost(g^c) = cost(E_1^c) + cost(E_2^c) + cost(E_3^c)$, where $E_1^c$ denotes the edges between the source $s$ and the first color class $V'_1$, $E_2^c$ denotes the edges between the second color class $V'_2$ and the sink $t$, and $E_3^c$ denotes the edges between $V'_1$ and $V'_2$ in $g^c$.

First, we have 
\begin{equation*}
cost(E_3^c) = \sum_{e \in E^c_{\mathrm{ans}}} cost(e) = \sum_{e \in E^c_{\mathrm{ans}}} \overline{d}_S(p^{-1}(e)) = \sum_{u \in D} \overline{d}_S(u),
\end{equation*}
where the second equality follows the construction of $g_c$ and the mapping function $p$, and the last equality follows from the correspondence between $E^c_{\mathrm{ans}}$ and $D$.

Then, we consider $cost(E_1^c) + cost(E_2^c)$. For each vertex $v$ in $V'_1$, the flow is $f_{\text{in}}(v)$ which implies that there are $f_{\text{in}}(v)$ selected edges in the maximum flow between $s$ and $v$. By our construction of $g^c$ and the constrained maximum flow, the $f_{\text{in}}(v)$ edges with minimum cost, i.e., those corresponding to $\{0,1,\ldots,f_{\text{in}}(v)-1\}$, will be selected. The same observation applies to each vertex $v$ in $V'_2$. Thus, we have $cost(E_1^c) = \sum_{v \in V'_1} \sum_{i=0}^{f_{\text{in}}(v)-1} i$ and $cost(E_2^c) = \sum_{v \in V'_2} \sum_{i=0}^{f_{\text{out}}(v)-1} i$. Subsequently, we can derive the following:
\begin{align*}
&\sum_{v \in V'_1} \sum_{i=0}^{f_{\text{in}}(v)-1} i + \sum_{v \in V'_2} \sum_{i=0}^{f_{\text{out}}(v)-1} i \\
&= \frac{1}{2}\left(\sum_{v\in V_1'} f_{\text{in}}^2(v)-f_{\text{in}}(v)\right) + \frac{1}{2}\left(\sum_{v\in V_2'} f_{\text{out}}^2(v)-f_{\text{out}}(v)\right) \\
&= \frac{1}{2} \left(\sum_{v\in V_1'}\sum_{\substack{p\in D\\\mathrm{col}_1(p)=v}}\sum_{\substack{q\in D\\\mathrm{col}_1(q)=v}} 1 - \sum_{v\in V_1'}\sum_{\substack{p\in D\\\mathrm{col}_1(p)=v}} 1\right) \\
&\quad + \frac{1}{2} \left(\sum_{v\in V_2'}\sum_{\substack{p\in D\\\mathrm{col}_2(p)=v}}\sum_{\substack{q\in D\\\mathrm{col}_2(q)=v}} 1 - \sum_{v\in V_2'}\sum_{\substack{p\in D\\\mathrm{col}_2(p)=v}} 1\right) \\
&= \frac{1}{2} \left(\sum_{p\in D}\sum_{q\in D} [\mathrm{col}_1(p) = \mathrm{col}_1(q)] - \sum_{p\in D} 1\right) \\
&\quad + \frac{1}{2} \left(\sum_{p\in D}\sum_{q\in D} [\mathrm{col}_2(p) = \mathrm{col}_2(q)] - \sum_{p\in D} 1\right) \\
&= \frac{1}{2} \left(\sum_{p\in D}\sum_{q\in D} [\mathrm{col}_1(p) = \mathrm{col}_1(q)] + [\mathrm{col}_2(p) = \mathrm{col}_2(q)] - 2\sum_{p\in D} 1\right) \\
&= \frac{1}{2} \sum_{u,v\in D} I_G(u, v),
\end{align*}
where the second equality follows from the fact that $f_{\text{in}}(v) = \sum_{p\in D \text{ and } \mathrm{col}_1(p)=v} 1$, and the last equality follows directly from the definition of the double-coloring indicator.

Combining the above, we have:
\begin{align*}
cost(g^c) &= cost(E_1^c) + cost(E_2^c) + cost(E_3^c) \\
    &= \sum_{c \in V_1'} \sum_{i=0}^{f_{\text{in}}(c)-1} i + \sum_{c \in V_2'} \sum_{i=0}^{f_{\text{out}}(c)-1} i + \sum_{e \in E^c_{\mathrm{ans}}} cost(e) \\
    &= \frac{1}{2} \sum_{u \in D} \sum_{v \in D} I_G(u, v) + \sum_{u \in D} \overline{d}_S(u). 
\end{align*}
This completes the proof of Lemma~\ref{lem:double-coloring-flow}.
\end{proof}

\revision{
\subsection{Proof of Lemma~\ref{lemma:tight-bound}}

\begin{proof}[Proof of Lemma~\ref{lemma:tight-bound}]
Let $h$ be an arbitrary candidate subgraph considered in the upper-bound computation. For two distinct vertices $u,v \in V(h)$, define
\[
I^{\text{single}}(u,v)=
\begin{cases}
1, & \text{if } col_1(u)=col_1(v),\\
0, & \text{otherwise},
\end{cases}
\]
and
\[
I^{\text{double}}(u,v)=
\begin{cases}
1, & \text{if } col_1(u)=col_1(v)\ \text{or}\ col_2(u)=col_2(v),\\
0, & \text{otherwise}.
\end{cases}
\]
By definition, for every distinct pair $u,v \in V(h)$, $I^{\text{double}}(u,v)\ge I^{\text{single}}(u,v)$.
We next define the violation value.
The violation value used by \textbf{UB-Single} for $h$ is
$Vio_{\text{single}}(h)
=
\frac{1}{2}\sum_{u,v\in V(h)} I^{\text{single}}(u,v)
+
\sum_{u\in V(h)} \overline{d}_S(u)$,
while the violation value used by \textbf{UB-Double} is
$Vio_{\text{double}}(h)
=
\frac{1}{2}\sum_{u,v\in V(h)} I^{\text{double}}(u,v)
+
\sum_{u\in V(h)} \overline{d}_S(u)$.
Thus,
\begin{align*}
Vio_{\text{double}}(h)
&=
\frac{1}{2}\sum_{u,v\in V(h)} I^{\text{double}}(u,v)
+
\sum_{u\in V(h)} \overline{d}_S(u) \\
&\ge
\frac{1}{2}\sum_{u,v\in V(h)} I^{\text{single}}(u,v)
+
\sum_{u\in V(h)} \overline{d}_S(u) \\
&=
Vio_{\text{single}}(h).
\end{align*}
Now let $\mathcal{H}_{\text{single}}=\{h \mid Vio_{\text{single}}(h)\le k\}$ and $\mathcal{H}_{\text{double}}=\{h \mid Vio_{\text{double}}(h)\le k\}$.
Since $Vio_{\text{double}}(h)\ge Vio_{\text{single}}(h)$ for every $h$, we have $\mathcal{H}_{\text{double}} \subseteq \mathcal{H}_{\text{single}}$.
Hence,
\[
\textbf{UB-Double}
=
\max_{h\in \mathcal{H}_{\text{double}}} |V(h)|
\le
\max_{h\in \mathcal{H}_{\text{single}}} |V(h)|
=
\textbf{UB-Single}.
\]
Therefore, \textbf{UB-Double} is always at least as tight as \textbf{UB-Single}.
\end{proof}

}

\begin{figure}[t]
\centering
\subfigure[Select $v_3$ (Strategy 2).]{
\centering
\begin{tikzpicture}[scale=0.55, node distance=2cm, thick] 
  \node[circle, minimum size=2mm, draw] (0) at (0,2.5) {$v_0$};
  \node[circle, minimum size=2mm, draw] (1) at (1.75,2.5) {$v_1$};
  \node[circle, minimum size=2mm, draw] (2) at (3.0,0) {$v_2$};
  \node[circle, minimum size=2mm, draw] (3) at (1.75,0) {$v_3$};
  \node[circle, minimum size=2mm, draw] (4) at (3.0,2.5) {$v_4$};
  \node[circle, minimum size=2mm, draw] (5) at (1.75,1.25) {$v_5$};
  
  \begin{scope}[on background layer]
    \node[fit=(1)(2)(3)(4)(5), fill=gray!10, rounded corners, inner sep=5pt, label=below:$C_{temp}$] {};
    \node[fit=(0), fill=gray!20, rounded corners, inner sep=5pt, label=below:$S_{opt}$] {};
  \end{scope}
  
  \foreach \a/\b in {0/1,0/3,0/5,1/4,2/1,2/5,2/4,3/5,2/3}
    \draw (\a) -- (\b);
  
\end{tikzpicture}
}
\hspace{0.08\columnwidth}
\subfigure[Add $v_5$.]{
\centering

\begin{tikzpicture}[scale=0.55, node distance=2cm, thick] 
  \node[circle, minimum size=2mm, draw] (0) at (0,2.5) {$v_0$};
  \node[circle, minimum size=2mm, draw] (1) at (1.75,2.5) {$v_1$};
  \node[circle, minimum size=2mm, draw] (2) at (3.0,0) {$v_2$};
  \node[circle, minimum size=2mm, draw] (3) at (0,0) {$v_3$};
  \node[circle, minimum size=2mm, draw] (4) at (3.0,2.5) {$v_4$};
  \node[circle, minimum size=2mm, draw] (5) at (1.75,1.25) {$v_5$};
  
  \begin{scope}[on background layer]
    \node[fit=(1)(2)(4)(5), fill=gray!10, rounded corners, inner sep=5pt, label=below:$C_{temp}$] {};
    \node[fit=(0)(3), fill=gray!20, rounded corners, inner sep=5pt, label=below:$S_{opt}$] {};
  \end{scope}
  
  \foreach \a/\b in {0/1,0/3,0/5,1/4,2/1,2/5,2/4,3/5,2/3}
    \draw (\a) -- (\b);
  
\end{tikzpicture}
}

\vspace{-8pt} 

\subfigure[Add $v_2$.]{
\centering

\begin{tikzpicture}[scale=0.55, node distance=2cm, thick] 
  \node[circle, minimum size=2mm, draw] (0) at (0,2.5) {$v_0$};
  \node[circle, minimum size=2mm, draw] (1) at (2.25,2.5) {$v_1$};
  \node[circle, minimum size=2mm, draw] (2) at (2.25,0) {$v_2$};
  \node[circle, minimum size=2mm, draw] (3) at (0,0) {$v_3$};
  \node[circle, minimum size=2mm, draw] (4) at (2.95,1.25) {$v_4$};
  \node[circle, minimum size=2mm, draw] (5) at (0.7,1.25) {$v_5$};
  
  \begin{scope}[on background layer]
    \node[fit=(1)(2)(4), fill=gray!10, rounded corners, inner sep=2pt, label=below:$C_{temp}$] {};
    \node[fit=(0)(3)(5), fill=gray!20, rounded corners, inner sep=2pt, label=below:$S_{opt}$] {};
  \end{scope}
  
  \foreach \a/\b in {0/1,0/3,0/5,1/4,2/1,2/5,2/4,3/5,2/3}
    \draw (\a) -- (\b);
  
\end{tikzpicture}
}
\hspace{0.08\columnwidth}
\subfigure[Termination.]{
\centering
\begin{tikzpicture}[scale=0.55, node distance=2cm, thick] 
  \node[circle, minimum size=2mm, draw] (0) at (0,2) {$v_0$};
  \node[circle, minimum size=2mm, draw] (1) at (3.0,2) {$v_1$};
  \node[circle, minimum size=2mm, draw] (2) at (1.5,0) {$v_2$};
  \node[circle, minimum size=2mm, draw] (3) at (0,0) {$v_3$};
  \node[circle, minimum size=2mm, draw] (4) at (3.0,0.75) {$v_4$};
  \node[circle, minimum size=2mm, draw] (5) at (1.3,1.3) {$v_5$};
  
  \begin{scope}[on background layer]
    \node[fit=(1)(4), fill=gray!10, rounded corners, inner sep=2pt, label=below:$C_{temp}$] {};
    \node[fit=(0)(3)(5)(2), fill=gray!20, rounded corners, inner sep=2pt, label=below:$S_{opt}$] {};
  \end{scope}
  
  \foreach \a/\b in {0/1,0/3,0/5,1/4,2/1,2/5,2/4,3/5,2/3}
    \draw (\a) -- (\b);
  
\end{tikzpicture}
}
\vspace{-0.2in}
\caption{Execution trace of \texttt{IRSolver} with $k=1$. Dark and light gray regions denote $S_{opt}$ and $C_{temp}$, respectively. The algorithm resolves a tie using Greedy Strategy 2 in (a), sequentially adds unique candidates in (b)-(c), and terminates in (d) as adding $v_1$ violates the $k$-defective constraint.}
\label{fig:IRSolver-example}
\end{figure}
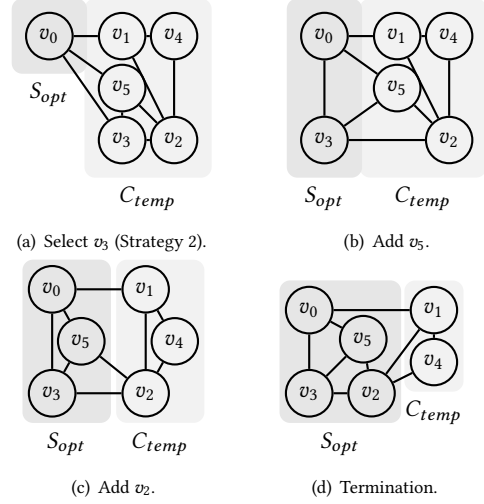

\begin{table}[t]
\caption{Number of DIMACS instances solved by the algorithms with a 3-hour limit.}\label{tab:num-instances-dimacs}
 \vspace{-0.15in}
\resizebox{0.6\columnwidth}{!}{
\begin{tabular}{|l|rrrrr|}
\hline
& \multicolumn{5}{c|}{DIMACS graphs}  \\
& \texttt{BBRes} & \texttt{kDC2} & \texttt{MDC} & \texttt{DnBk} & \texttt{WODC} \\
\hline
$k=1$ & \textbf{45} & 37 & 42 & 44 & 44 \\
$k=3$ & \textbf{43} & 23 & 31 & 19 & 41 \\
$k=5$ & \textbf{40} & 19 & 24 & 17 & 37 \\
$k=10$ & \textbf{33} & 14 & 18 & 16 & \textbf{33} \\
$k=15$ & \textbf{24} & 12 & 14 & 13 & 23 \\
$k=20$ & \textbf{22} & 11 & 14 & 12 & 16 \\

\hline
\end{tabular}
}
\end{table}

\begin{table*}
\caption{
Running time (in seconds) on 10 representative DIMACS benchmark graphs. The best performer is highlighted in bold; specifically, if the running time is within 10\% of the fastest time, it is considered as the best.
}\label{tab:dimacs-representative}
\vspace{-0.05in}
\resizebox{1.05\textwidth}{!}{
\begin{tabular}{|l|r|r|r|r|r|r|r|r|r|r|r|r|r|r|r|r|r|r|r|r|r|r|}
\hline
 &  &  & \multicolumn{5}{c|}{$k = 5$} & \multicolumn{5}{c|}{$k = 10$} & \multicolumn{5}{c|}{$k = 15$} & \multicolumn{5}{c|}{$k = 20$} \\
Graphs & $n$ & $m$ & \multicolumn{1}{c}{\texttt{BBRes}} & \multicolumn{1}{c}{\texttt{kDC2}} & \multicolumn{1}{c}{\texttt{MDC}} & \multicolumn{1}{c}{\texttt{DnBk}} & \multicolumn{1}{c|}{\texttt{WODC}} & \multicolumn{1}{c}{\texttt{BBRes}} & \multicolumn{1}{c}{\texttt{kDC2}} & \multicolumn{1}{c}{\texttt{MDC}} & \multicolumn{1}{c}{\texttt{DnBk}} & \multicolumn{1}{c|}{\texttt{WODC}} & \multicolumn{1}{c}{\texttt{BBRes}} & \multicolumn{1}{c}{\texttt{kDC2}} & \multicolumn{1}{c}{\texttt{MDC}} & \multicolumn{1}{c}{\texttt{DnBk}} & \multicolumn{1}{c|}{\texttt{WODC}} & \multicolumn{1}{c}{\texttt{BBRes}} & \multicolumn{1}{c}{\texttt{kDC2}} & \multicolumn{1}{c}{\texttt{MDC}} & \multicolumn{1}{c}{\texttt{DnBk}} & \multicolumn{1}{c|}{\texttt{WODC}}\\
\hline

c-fat500-2 & 500 & 9K & \textbf{0.001} &0.002 & 0.005 & 0.003 & 0.083  & \textbf{0.001} &0.004 & 0.008 & 0.006 & 0.084  & \textbf{0.01} &0.046 & 0.045 & 0.145 & 0.388  & \textbf{0.045} &0.096 & 0.072 & 0.339 & 0.514 \\ 
c-fat500-5 & 500 & 23K & \textbf{0.004} &0.011 & 0.014 & 0.008 & 0.397  & \textbf{0.009} &0.023 & 0.023 & 0.014 & 0.383  & \textbf{0.011} &0.034 & 0.032 & 0.025 & 0.577  & \textbf{0.015} &0.043 & 0.044 & 0.028 & 0.661 \\ 
c-fat500-1 & 500 & 4K & \textbf{0.0} &\textbf{0.0} & 0.003 & 0.002 & 0.029  & 0.006 & \textbf{0.005} & 0.007 & 0.023 & 0.056  & 0.009 & \textbf{0.007} & 0.017 & 0.024 & 0.071  & \textbf{1.196} &6.744 & 11.956 & 220.672 & \textbf{1.282} \\ 
c-fat500-10 & 500 & 46K & \textbf{0.022} &0.055 & 0.035 & 0.028 & 1.206  & \textbf{0.029} &0.117 & 0.056 & 0.074 & 1.722  & \textbf{0.029} &0.18 & 0.116 & 0.093 & 1.858  & \textbf{0.047} &0.242 & 0.134 & 0.098 & 2.201 \\ 
san400-0-5-1 & 400 & 39K & 56.024 & 2070.312 & 2632.115 & - & \textbf{9.666}  & \textbf{93.256} &- & 7391.9 & - & 6707.4  & \textbf{95.123} &- & - & - & 5638.29  & \textbf{460.396} &- & - & - & - \\ 
MANN-a27 & 378 & 70K & \textbf{2262.145} &- & 3057.384 & - & -  & \textbf{2319.961} &- & - & - & -  & \textbf{2265.974} &- & - & - & -  & \textbf{3280.755} &- & - & - & - \\ 
hamming8-2 & 256 & 31K & \textbf{0.601} &1827.654 & 787.899 & - & 111.738  & \textbf{17.051} &- & - & - & 1328.32  & \textbf{339.576} &- & - & - & -  & \textbf{2109.191} &- & - & - & - \\ 
san200-0-7-1 & 200 & 13K & 14.743 & 2509.822 & 72.747 & 205.319 & \textbf{0.917}  & 74.559 & - & 1323.845 & 1202.407 & \textbf{13.874}  & \textbf{83.843} &- & - & - & 237.672  & \textbf{101.182} &- & - & - & - \\ 
c-fat200-2 & 200 & 3K & 0.005 & \textbf{0.001} & 0.002 & \textbf{0.001} & 0.028  & 0.006 & \textbf{0.002} & \textbf{0.002} & \textbf{0.002} & 0.037  & \textbf{0.005} &0.015 & 0.025 & 0.036 & 0.089  & \textbf{0.015} &0.035 & 0.021 & 0.086 & 0.22 \\ 
c-fat200-5 & 200 & 8K & 0.003 & 0.005 & 0.005 & \textbf{0.002} & 0.146  & \textbf{0.004} &0.015 & 0.008 & 0.005 & 0.204  & \textbf{0.006} &0.016 & 0.013 & \textbf{0.006} & 0.259  & \textbf{0.008} &0.027 & 0.016 & 0.01 & 0.319 \\ 
\hline
\end{tabular}
}
\end{table*}

\section{Omitted Example}
\label{app:example}

\textbf{Example of \texttt{IRSolver}}. Figure~\ref{fig:IRSolver-example} demonstrates \texttt{IRSolver} on a \texttt{MissingTwoDeg} instance with $k=1$, initialized with $S_{opt}=\{v_0\}$. In Figure~\ref{fig:IRSolver-example}(a), the candidate set is $\Gamma_{min}=\{v_1, v_3, v_5\}$. Since $\overline{d}_{C_{temp}}(v)=2$ for all candidates, Greedy Strategy 2 (Equation~(\ref{eq:greedy-case2})) is triggered, selecting $v_3$ for its minimum $\overline{d}_{\Gamma_{min}}$. Subsequently, unique candidates $v_5$ and $v_2$ are directly added to $S_{opt}$ in Figure~\ref{fig:IRSolver-example}(b)-(c). The process terminates in Figure~\ref{fig:IRSolver-example}(d) as adding $v_1$ violates the $k$-defective constraint, yielding $G[S_{opt}]$ as the solution.

\section{Additional Experiments on Dense and Synthetic Graphs}\label{app:exp}

While the experimental results in the main text demonstrate the superior efficiency of \texttt{BBRes} on large-scale real-world graphs (which are typically sparse), the maximum $k$-defective clique problem is theoretically most challenging on \textbf{dense graphs} due to the exponential explosion of the search space. To rigorously stress-test our algorithm under these extreme conditions and to verify the effectiveness of our pruning strategy in high-density environments, we extend our evaluation to two complementary datasets:

\begin{itemize}[leftmargin=*]
    \item \textbf{DIMACS Benchmarks}\footnote{https://networkrepository.com/dimacs.php}: A standard collection comprising \textbf{78} dense graphs. These instances serve as a stress test for raw performance on notoriously hard combinatorial problems.
    \item \textbf{LFR Synthetic Graphs}: Generated via the LFR benchmark~\cite{lancichinetti2008benchmark} with $n=5,000$. By systematically varying the graph density from $0.004$ to $0.04$, these graphs enable a fine-grained sensitivity analysis of the algorithms to varying graph densities.
\end{itemize}

\smallskip
\noindent\textbf{Detailed performance on DIMACS benchmarks}.
Table~\ref{tab:num-instances-dimacs} presents the aggregate number of solved instances. In particular, the results confirm that \texttt{BBRes} maintains its dominance even in high-density environments. As summarized in Table~\ref{tab:num-instances-dimacs}, \texttt{BBRes} consistently outperforms state-of-the-art baselines across all tested values of $k$. The performance gap becomes particularly pronounced for larger $k$; for instance, at $k=20$, \texttt{BBRes} solves 22 instances, whereas \texttt{kDC2} and \texttt{DnBk} only solve 11 and 12, respectively, marking a nearly $2\times$ improvement in solvability. Table~\ref{tab:dimacs-representative} provides a granular breakdown of the running times on 10 representative DIMACS graphs. These detailed statistics offer a deeper insight into the performance gaps between \texttt{BBRes} and the baselines, particularly on hard instances. The results in Table~\ref{tab:dimacs-representative} reveal that \texttt{BBRes} achieves orders-of-magnitude speedups. Taking the \texttt{hamming8-2} graph as an example, for $k \in \{5, 10, 15\}$, \texttt{BBRes} finishes in 0.6s, 17.1s, and 339.6s, respectively. In stark contrast, \texttt{kDC2}, \texttt{MDC}, and \texttt{DnBk} require over 700s even for the easier case of $k=5$ and fail to terminate within the 3-hour cutoff (10,800s) for larger $k$. \texttt{WODC} also struggles, taking 111.7s for $k=5$ ($186\times$ slower than \texttt{BBRes}) and timing out for $k \ge 10$.
We attribute this efficiency to the synergy between our diameter-2 partitioning and the double-coloring upper bound, which effectively prunes the search space in the dense clusters characteristic of DIMACS graphs.

\begin{figure}[t]
    \centering
    \subfigure[$k = 5$]{
        \includegraphics[width=0.3\textwidth, trim=0cm 0cm 0cm 0cm, clip]{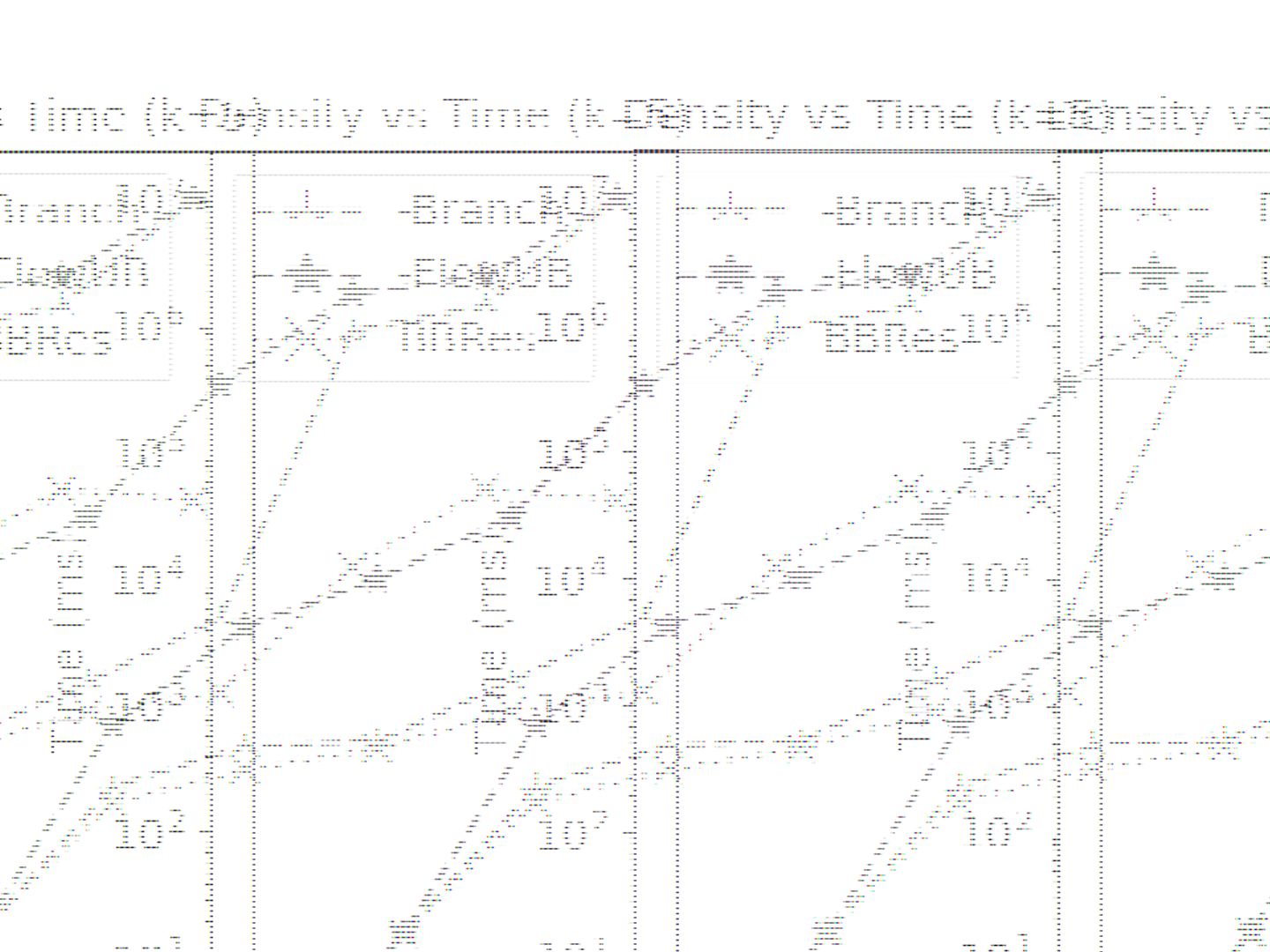}
        \label{fig:LFR_5}
    }
    \subfigure[$k = 15$]{
        \includegraphics[width=0.3\textwidth, trim=0cm 0cm 0cm 0cm, clip]{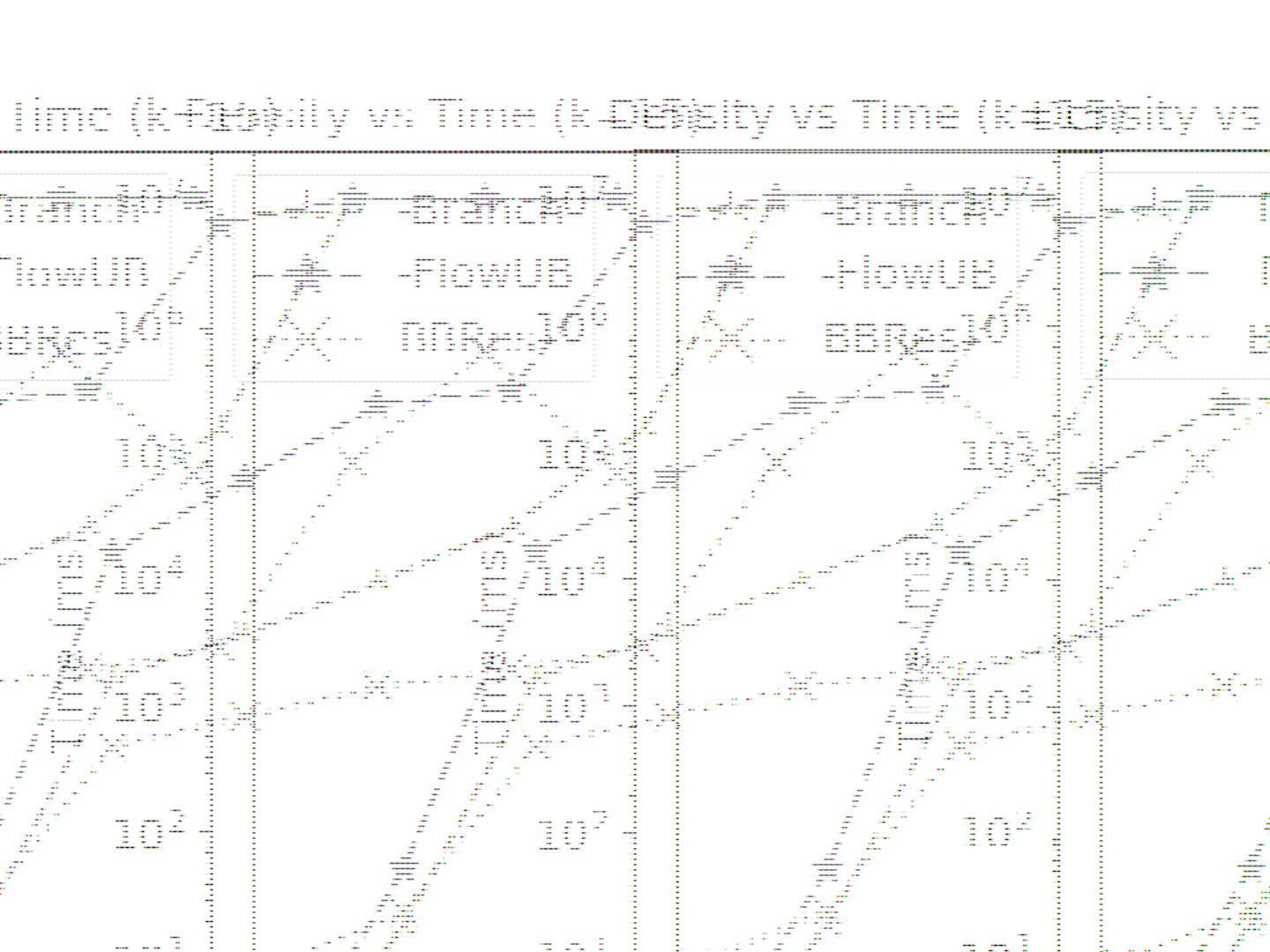}
        \label{fig:LFR_15}
    }
    \vspace{-0.15in}
    \caption{Runtime analysis on LFR synthetic graphs with varying densities. The results illustrate the scalability of \texttt{BBRes} compared to its variants as graph density increases ($k=5$ and $k=15$).}\label{fig:LFR_density}
\end{figure}

\smallskip
\noindent\textbf{Impact of graph density}.
To explicitly isolate the correlation between graph density and our pruning efficiency, we evaluate \texttt{BBRes} and its ablation variants (\texttt{-Branch} and \texttt{-FlowUB}) on the LFR synthetic datasets. Figure~\ref{fig:LFR_density} plots the runtime against varying graph densities for $k=5$ and $k=15$. The results clearly demonstrate that as density increases, the performance advantage of \texttt{BBRes} becomes increasingly pronounced. Specifically, while \texttt{BBRes} outperforms its variants across the board, the relative speedup ratio grows substantially with density. For instance, at $k=5$ (Figure~\ref{fig:LFR_density}(a)), on the graph with a density of $0.012$, \texttt{-Branch} and \texttt{-FlowUB} are $1.3\times$ and $4.1\times$ slower than \texttt{BBRes}, respectively. However, at a higher density of $0.032$, these slowdown factors escalate to approximately ${73\times}$ and ${174\times}$. This empirical evidence confirms that our double-coloring upper bound is particularly effective in handling high-density regions, where simpler bounds (such as those used in \texttt{-FlowUB}) become loose and ineffective.

\begin{table*}[htbp]
\centering
\caption{Detailed runtime comparison for coloring sequences}
\label{tab:second-coloring-strategies}
\resizebox{\linewidth}{!}{
\begin{tabular}{l|rrrrrr|rrrrrr}
\hline
\multirow{2}{*}{Graph} & \multicolumn{6}{c|}{$k=5$} & \multicolumn{6}{c}{$k=15$} \\
 & BBRes & random & S\_ord & S\_rev & peel\_ord & peel\_rev & BBRes & random & S\_ord & S\_rev & peel\_ord & peel\_rev \\
\hline
ca-MathSciNet & 0.054 & 0.060 & \textbf{0.049} & 0.059 & \textbf{0.045} & 0.063 & 0.054 & \textbf{0.046} & 0.066 & 0.064 & \textbf{0.047} & 0.056 \\
ca-citeseer & \textbf{0.029} & 0.035 & 0.034 & 0.031 & \textbf{0.028} & 0.034 & \textbf{0.029} & \textbf{0.031} & 0.035 & 0.040 & 0.039 & \textbf{0.030} \\
ca-coauthors-dblp & \textbf{0.241} & \textbf{0.258} & \textbf{0.246} & \textbf{0.242} & \textbf{0.238} & \textbf{0.241} & \textbf{0.301} & \textbf{0.311} & 0.332 & \textbf{0.327} & \textbf{0.309} & \textbf{0.310} \\
ca-dblp-2010 & \textbf{0.026} & 0.031 & 0.031 & 0.031 & \textbf{0.027} & 0.029 & 0.029 & 0.029 & 0.027 & \textbf{0.024} & 0.030 & 0.030 \\
ca-dblp-2012 & 0.063 & 0.057 & 0.053 & 0.065 & \textbf{0.048} & 0.058 & \textbf{0.046} & 0.056 & 0.073 & 0.071 & \textbf{0.048} & 0.057 \\
ca-hollywood-2009 & \textbf{0.900} & \textbf{0.913} & \textbf{0.871} & \textbf{0.856} & \textbf{0.853} & \textbf{0.901} & \textbf{0.810} & \textbf{0.815} & 0.896 & 0.929 & \textbf{0.809} & \textbf{0.847} \\
sc-ldoor & \textbf{21.403} & 24.028 & \textbf{21.417} & 24.961 & \textbf{21.812} & 24.997 & \textbf{365.331} & 435.735 & 433.159 & 430.226 & \textbf{374.732} & 442.794 \\
sc-msdoor & \textbf{10.693} & 11.747 & \textbf{10.637} & 12.662 & \textbf{10.995} & 12.581 & \textbf{202.630} & 241.546 & 235.900 & 241.893 & \textbf{207.412} & 246.285 \\
sc-pwtk & \textbf{1.603} & \textbf{1.686} & \textbf{1.654} & \textbf{1.668} & \textbf{1.667} & \textbf{1.735} & \textbf{32.673} & 42.043 & 55.849 & 47.503 & 46.640 & 36.266 \\
sc-shipsec1 & \textbf{0.107} & \textbf{0.113} & \textbf{0.113} & 0.139 & \textbf{0.111} & 0.135 & \textbf{0.626} & 0.664 & \textbf{0.644} & \textbf{0.587} & \textbf{0.618} & 0.700 \\
sc-shipsec5 & \textbf{0.168} & \textbf{0.183} & \textbf{0.175} & \textbf{0.183} & \textbf{0.184} & 0.199 & \textbf{0.848} & \textbf{0.913} & 1.058 & 0.944 & \textbf{0.895} & \textbf{0.898} \\
soc-delicious & \textbf{0.081} & \textbf{0.087} & 0.103 & \textbf{0.086} & \textbf{0.084} & \textbf{0.084} & \textbf{0.168} & \textbf{0.172} & \textbf{0.183} & \textbf{0.177} & \textbf{0.183} & \textbf{0.170} \\
soc-digg & \textbf{26.293} & 29.616 & \textbf{28.104} & \textbf{27.454} & 29.653 & \textbf{27.522} & \textbf{41.128} & 49.652 & 50.199 & \textbf{44.693} & 51.987 & \textbf{44.959} \\
soc-douban & 0.039 & 0.040 & \textbf{0.031} & \textbf{0.031} & \textbf{0.032} & 0.040 & \textbf{7.567} & 8.247 & 8.507 & \textbf{7.460} & \textbf{8.143} & 8.241 \\
soc-flixster & \textbf{0.557} & 0.651 & 0.640 & \textbf{0.599} & 0.706 & 0.652 & \textbf{0.773} & 1.029 & 1.120 & 0.960 & 1.404 & 0.944 \\
soc-gowalla & 0.236 & \textbf{0.189} & 0.233 & 0.237 & \textbf{0.183} & \textbf{0.185} & \textbf{0.953} & 1.074 & 1.105 & 1.041 & 1.124 & \textbf{0.892} \\
soc-lastfm & \textbf{3.398} & 5.137 & 5.496 & \textbf{3.678} & 5.461 & 3.938 & \textbf{64.611} & 139.305 & 174.340 & 74.833 & 164.495 & 73.866 \\
soc-orkut & \textbf{88.773} & \textbf{97.033} & 107.130 & 104.046 & 105.480 & 102.082 & \textbf{109.229} & \textbf{119.144} & 126.696 & \textbf{115.648} & 127.000 & \textbf{118.267} \\
soc-pokec & \textbf{5.949} & 7.890 & 6.655 & \textbf{6.191} & \textbf{6.131} & \textbf{6.191} & \textbf{8.904} & \textbf{8.417} & 10.800 & 10.909 & \textbf{8.750} & \textbf{9.115} \\
soc-youtube & \textbf{0.454} & 0.575 & \textbf{0.482} & 0.555 & 0.522 & 0.549 & \textbf{7.767} & 11.239 & 12.185 & 9.966 & 17.252 & \textbf{7.883} \\
soc-youtube-snap & \textbf{0.693} & 0.791 & 0.784 & \textbf{0.724} & 0.776 & 0.764 & \textbf{13.491} & 20.128 & 20.939 & 18.154 & 30.392 & \textbf{13.322} \\
socfb-A-anon & \textbf{6.059} & 6.477 & \textbf{5.806} & \textbf{5.873} & 6.825 & 6.714 & \textbf{9.177} & \textbf{9.386} & 10.375 & \textbf{9.819} & \textbf{9.197} & \textbf{9.220} \\
socfb-B-anon & \textbf{10.832} & 11.558 & 11.309 & \textbf{10.390} & \textbf{10.069} & \textbf{10.564} & \textbf{19.409} & \textbf{19.861} & 23.992 & 22.308 & \textbf{20.546} & \textbf{20.227} \\
tech-RL-caida & \textbf{0.112} & 0.138 & 0.138 & 0.140 & \textbf{0.110} & 0.138 & 0.144 & 0.156 & 0.134 & 0.150 & 0.170 & \textbf{0.120} \\
tech-as-skitter & \textbf{0.726} & \textbf{0.763} & \textbf{0.738} & \textbf{0.720} & \textbf{0.762} & \textbf{0.764} & \textbf{0.792} & \textbf{0.819} & 0.926 & 0.893 & \textbf{0.836} & \textbf{0.836} \\
web-arabic-2005 & \textbf{0.021} & 0.025 & 0.024 & \textbf{0.021} & 0.025 & \textbf{0.021} & 0.028 & 0.028 & 0.030 & 0.032 & 0.030 & \textbf{0.024} \\
web-it-2004 & \textbf{0.164} & \textbf{0.170} & \textbf{0.166} & \textbf{0.169} & \textbf{0.167} & \textbf{0.177} & \textbf{0.229} & \textbf{0.227} & \textbf{0.231} & \textbf{0.233} & \textbf{0.235} & \textbf{0.231} \\
web-sk-2005 & \textbf{0.010} & \textbf{0.010} & 0.012 & 0.013 & \textbf{0.010} & \textbf{0.010} & 0.013 & 0.013 & 0.013 & 0.013 & \textbf{0.010} & 0.013 \\
web-uk-2005 & \textbf{0.254} & \textbf{0.259} & \textbf{0.257} & \textbf{0.253} & \textbf{0.263} & \textbf{0.258} & \textbf{0.314} & \textbf{0.316} & \textbf{0.326} & \textbf{0.326} & \textbf{0.317} & \textbf{0.316} \\
web-wikipedia2009 & \textbf{0.876} & 0.968 & \textbf{0.943} & \textbf{0.879} & \textbf{0.884} & \textbf{0.915} & \textbf{0.833} & \textbf{0.810} & 1.078 & 1.106 & 0.925 & 0.954 \\
\hline
\end{tabular}
}
\end{table*}

\section{Effect of Second-Coloring Strategies on \textbf{UB-Double}}
\label{sec:second-coloring}
\revision{To further evaluate the design of the second coloring phase in \textbf{UB-Double}, we conduct supplementary experiments using several alternative orderings for the second coloring, including \texttt{random}, \texttt{S\_ord}, \texttt{S\_rev}, \texttt{peel\_ord}, and \texttt{peel\_rev}. Table~\ref{tab:second-coloring-strategies} reports the detailed running times. Here, \texttt{BBRes} denotes the default strategy used in our algorithm, which processes the vertices in $C$ according to their current memory order. \texttt{random} randomly permutes the vertices in $C$ before the second coloring. \texttt{S\_ord} orders the vertices increasingly by the number of non-neighbors they have in $S$, while \texttt{S\_rev} uses the reverse order. Finally, \texttt{peel\_ord} follows the degeneracy order, and \texttt{peel\_rev} uses the reverse degeneracy order.

Overall, the default strategy used in \texttt{BBRes} achieves the best performance on the majority of the tested instances, while \texttt{peel\_rev} is often the strongest alternative. A possible explanation is that \texttt{peel\_rev}, which reverses the initial degeneracy-based order, tends to produce a second coloring that is structurally different from the first one, and can therefore help expose additional conflicting pairs. However, on the most challenging instances, the strategy in \texttt{BBRes} still shows a clear advantage. For example, on \texttt{sc-msdoor} with $k=15$, \texttt{BBRes} finishes in 202.63 seconds, compared with 246.285 seconds for \texttt{peel\_rev}; on \texttt{sc-ldoor} with $k=15$, \texttt{BBRes} runs in 365.331 seconds, compared with 442.794 seconds for \texttt{peel\_rev}.

Interestingly, the second-coloring strategy used in \texttt{BBRes} does not rely on an explicit structural heuristic. Instead, it follows the current storage order of candidate vertices, which is dynamically changed by the recursive branch-and-bound process. Empirically, this dynamic ordering appears to produce a favorable level of diversity between the two colorings, leading to better pruning performance than both strictly deterministic variants and the purely random ordering on many hard instances. However, we do not have a formal theoretical explanation for why this strategy performs best. Understanding the structural properties of effective second-coloring orders remains an interesting direction for future work.
}

%% file: reference.bib
@article{ahujia1993network,
  title={Network flows: Theory, algorithms and applications},
  author={Ahujia, Ravindra K and Magnanti, Thomas L and Orlin, James B},
  journal={New Jersey: Rentice-Hall},
  volume={3},
  year={1993}
}

@article{lancichinetti2008benchmark,
  title={Benchmark graphs for testing community detection algorithms},
  author={Lancichinetti, Andrea and Fortunato, Santo and Radicchi, Filippo},
  journal={Physical Review E—Statistical, Nonlinear, and Soft Matter Physics},
  volume={78},
  number={4},
  pages={046110},
  year={2008},
  publisher={APS}
}

@misc{appendix,
  howpublished = {\url{https://github.com/Thaumaturge2020/BBRes/}},
  year         = {Technical Report and Source Code}
}

@article{wang2010recent,
  title={Recent advances in clustering methods for protein interaction networks},
  author={Wang, Jianxin and Li, Min and Deng, Youping and Pan, Yi},
  journal={BMC Genomics},
  volume={11},
  number={Suppl 3},
  pages={S10},
  year={2010},
  publisher={Springer}
}

@inproceedings{xia2025iterqc,
  author    = {Hongbo Xia and Kaiqiang Yu and Shengxin Liu and Cheng Long and Xun Zhou},
  title     = {Maximum degree-based quasi-clique search via an iterative framework},
  year      = {2025},
  pages = {3285--3296},
  booktitle = {Proc. ACM SIGKDD Int. Conf. Knowl. Discov. Data Mining (SIGKDD)},
}

@incollection{yu2021graph,
  title={Graph mining meets fake news detection},
  author={Yu, Kaiqiang and Long, Cheng},
  booktitle={Data Science for Fake News: Surveys and Perspectives},
  pages={169--189},
  year={2021},
  publisher={Springer}
}

@inproceedings{goldberg1987solving,
  title={Solving minimum-cost flow problems by successive approximation},
  author={Goldberg, Andrew and Tarjan, Robert},
  booktitle={Proceedings of STOC},
  pages={7--18},
  year={1987}
}

@inproceedings{berry2004emergent,
  title={Emergent clique formation in terrorist recruitment},
  author={Berry, Nina and Ko, Teresa and Moy, Tim and Smrcka, Julienne and Turnley, Jessica and Wu, Ben},
  booktitle={Proceedings of the AAAI Workshop on Agent Organizations: Theory and Practice},
  pages={1198--1208},
  year={2004}
}

@article{ahmed2016survey,
  title={A survey of anomaly detection techniques in financial domain},
  author={Ahmed, Mohiuddin and Mahmood, Abdun Naser and Islam, Md Rafiqul},
  journal={Future Generation Computer Systems},
  volume={55},
  pages={278--288},
  year={2016},
  publisher={Elsevier}
}

@inproceedings{leung2005unsupervised,
  title={Unsupervised anomaly detection in network intrusion detection using clusters},
  author={Leung, Kingsly and Leckie, Christopher},
  booktitle={Proceedings of the Australasian Conference on Computer Science},
  pages={333--342},
  year={2005}
}

@article{pattillo2013clique,
  title={On clique relaxation models in network analysis},
  author={Pattillo, Jeffrey and Youssef, Nataly and Butenko, Sergiy},
  journal={European Journal of Operational Research},
  volume={226},
  number={1},
  pages={9--18},
  year={2013},
  publisher={Elsevier}
}

@article{eppstein2013listing,
  title={Listing all maximal cliques in large sparse real-world graphs},
  author={Eppstein, David and L{\"o}ffler, Maarten and Strash, Darren},
  journal={Journal of Experimental Algorithmics (JEA)},
  volume={18},
  pages={3--1},
  year={2013},
  publisher={ACM New York, NY, USA}
}

@article{cheng2011finding,
  title={Finding maximal cliques in massive networks},
  author={Cheng, James and Ke, Yiping and Fu, Ada Wai-Chee and Yu, Jeffrey Xu and Zhu, Linhong},
  journal={ACM Transactions on Database Systems (TODS)},
  volume={36},
  number={4},
  pages={1--34},
  year={2011},
  publisher={ACM New York, NY, USA}
}

@inproceedings{tomita2010simple,
  title={A simple and faster branch-and-bound algorithm for finding a maximum clique},
  author={Tomita, Etsuji and Sutani, Yoichi and Higashi, Takanori and Takahashi, Shinya and Wakatsuki, Mitsuo},
  booktitle={International Workshop on Algorithms and Computation},
  pages={191--203},
  year={2010},
  organization={Springer}
}

@inproceedings{tomita2017efficient,
  title={Efficient algorithms for finding maximum and maximal cliques and their applications},
  author={Tomita, Etsuji},
  booktitle={International workshop on algorithms and computation},
  pages={3--15},
  year={2017},
  organization={Springer}
}

@article{san2016new,
  title={A new exact maximum clique algorithm for large and massive sparse graphs},
  author={San Segundo, Pablo and Lopez, Alvaro and Pardalos, Panos M},
  journal={Computers \& Operations Research},
  volume={66},
  pages={81--94},
  year={2016},
  publisher={Elsevier}
}

@article{pattabiraman2015fast,
  title={Fast algorithms for the maximum clique problem on massive graphs with applications to overlapping community detection},
  author={Pattabiraman, Bharath and Patwary, Md Mostofa Ali and Gebremedhin, Assefaw H and Liao, Wei-keng and Choudhary, Alok},
  journal={Internet Mathematics},
  volume={11},
  number={4-5},
  pages={421--448},
  year={2015},
  publisher={Taylor \& Francis}
}

@article{pardalos1994maximum,
  title={The maximum clique problem},
  author={Pardalos, Panos M and Xue, Jue},
  journal={Journal of global Optimization},
  volume={4},
  number={3},
  pages={301--328},
  year={1994},
  publisher={Springer}
}

@article{carraghan1990exact,
  title={An exact algorithm for the maximum clique problem},
  author={Carraghan, Randy and Pardalos, Panos M},
  journal={Operations Research Letters},
  volume={9},
  number={6},
  pages={375--382},
  year={1990},
  publisher={Elsevier}
}

@article{bedi2016community,
  title={Community detection in social networks},
  author={Bedi, Punam and Sharma, Chhavi},
  journal={Wiley Interdisciplinary Reviews: Data Mining and Knowledge Discovery},
  volume={6},
  number={3},
  pages={115--135},
  year={2016},
}

@article{BVZ03m,
  author        = {Batagelj, Vladimir and Zaver{\v{s}}nik, Matja{\v{z}}},
  title         = {An $O(m)$ algorithm for cores decomposition of networks},
  journal       = {CoRR},
  volume        = {cs.DS/0310049},
  archiveprefix = {arXiv},
  year          = {2003}
}

@article{chang2020clique,
  author       = {Lijun Chang},
  title        = {Efficient maximum clique computation an enumeration over large sparse graphs},
  journal = {The VLDB Journal},
  volume       = {29},
  number       = {5},
  pages        = {999--1022},
  year         = {2020},
}

@article{chang2023kdefective,
  title     = {Efficient maximum $k$-defective clique computation with improved time complexity},
  author    = {Chang, Lijun},
  journal   = {Proceedings of the ACM on Management of Data (SIGMOD)},
  volume    = {1},
  number    = {3},
  pages     = {1--26},
  year      = {2023},
  publisher = {ACM New York, NY, USA}
}

@article{chang2022efficient,
  title={Efficient maximum $k$-plex computation over large sparse graphs},
  author={Chang, Lijun and Xu, Mouyi and Strash, Darren},
  journal={Proceedings of the VLDB Endowment},
  volume={16},
  number={2},
  pages={127--139},
  year={2022},
  publisher={VLDB Endowment}
}

@article{Chang24kDC-2,
  author       = {Lijun Chang},
  title        = {Maximum defective clique computation: {I}mproved time complexities and
                  practical performance},
  journal      = {Proceedings of the VLDB Endowment},
  volume       = {18},
  number       = {2},
  pages        = {200--212},
  year         = {2024},
}

@article{chen2021computing,
  title     = {Computing maximum $k$-defective cliques in massive graphs},
  author    = {Chen, Xiaoyu and Zhou, Yi and Hao, Jin-Kao and Xiao, Mingyu},
  journal   = {Computers \& Operations Research},
  volume    = {127},
  pages     = {105131},
  year      = {2021},
  publisher = {Elsevier}
}

@article{cui2025efficient,
  title={On the efficient discovery of maximum $k$-defective biclique},
  author={Cui, Donghang and Li, Ronghua and Dai, Qiangqiang and Qin, Hongchao and Wang, Guoren},
  journal={arXiv preprint arXiv:2506.16121},
  year={2025}
}

@article{dai2023maximal,
  title   = {Maximal defective clique enumeration},
  author  = {Dai, Qiangqiang and Li, Rong-Hua and Liao, Meihao and Wang, Guoren},
  journal = {Proc. ACM SIGMOD Int. Conf. Manage. Data (SIGMOD)},
  volume  = {1},
  number  = {1},
  pages   = {1--26},
  year    = {2023}
}

@article{dai2024theoretically,
  title={Theoretically and practically efficient maximum defective clique search},
  author={Dai, Qiangqiang and Li, Ronghua and Cui, Donghang and Wang, Guoren},
  journal={Proceedings of the ACM on Management of Data (SIGMOD)},
  volume={2},
  number={4},
  pages={1--27},
  year={2024},
  publisher={ACM New York, NY, USA}
}

@article{ahuja1995capacity,
  title={A capacity scaling algorithm for the constrained maximum flow problem},
  author={Ahuja, Ravindra K and Orlin, James B},
  journal={Networks},
  volume={25},
  number={2},
  pages={89--98},
  year={1995},
}

@article{dinic1970maxflowAlg,
  title={Algorithm for solution of a problem of maximum flow in networks with power estimation},
  author={Dinic, Efim A},
  journal={Soviet Math. Doklady},
  volume={11},
  pages={1277--1280},
  year={1970}
}

@article{FHQ+20survey,
  title   = {A survey of community search over big graphs},
  author  = {Fang, Yixiang and Huang, Xin and Qin, Lu and Zhang, Ying and Zhang, Wenjie and Cheng, Reynold and Lin, Xuemin},
  journal = {The VLDB Journal},
  volume  = {29},
  pages   = {353--392},
  year    = {2020}
}

@inproceedings{gao2022exact,
  title     = {An exact algorithm with new upper bounds for the maximum $k$-defective clique problem in massive sparse graphs},
  author    = {Gao, Jian and Xu, Zhenghang and Li, Ruizhi and Yin, Minghao},
  booktitle = {Proceedings of AAAI},
  pages     = {10174--10183},
  year      = {2022}
}

@article{gschwind2018maximum,
  title={Maximum weight relaxed cliques and {R}ussian {D}oll {S}earch revisited},
  author={Gschwind, Timo and Irnich, Stefan and Podlinski, Isabel},
  journal={Discrete Applied Mathematics},
  volume={234},
  pages={131--138},
  year={2018},
  publisher={Elsevier}
}

@article{gschwind2021branch,
  title={A branch-and-price framework for decomposing graphs into relaxed cliques},
  author={Gschwind, Timo and Irnich, Stefan and Furini, Fabio and Calvo, Roberto Wolfler},
  journal={INFORMS Journal on Computing},
  volume={33},
  number={3},
  pages={1070--1090},
  year={2021},
  publisher={INFORMS}
}

@inproceedings{jin2024kdclub,
  title={KD-Club: An efficient exact algorithm with new coloring-based upper bound for the maximum $k$-defective clique problem},
  author={Jin, Mingming and Zheng, Jiongzhi and He, Kun},
  booktitle={Proceedings of AAAI},
  pages={20735--20742},
  year={2024}
}

@article{khalil2022parallel,
  title   = {Parallel mining of large maximal quasi-cliques},
  author  = {Khalil, Jalal and Yan, Da and Guo, Guimu and Yuan, Lyuheng},
  journal = {The VLDB Journal},
  volume  = {31},
  number  = {4},
  pages   = {649--674},
  year    = {2022}
}

@article{LEWIS1980219,
title = {The node-deletion problem for hereditary properties is {NP}-complete},
journal = {Journal of Computer and System Sciences},
volume = {20},
number = {2},
pages = {219-230},
year = {1980},
author = {John M. Lewis and Mihalis Yannakakis},
}

@inproceedings{Luo24defective,
  author       = {Chunyu Luo and
                  Yi Zhou and
                  Zhengren Wang and
                  Mingyu Xiao},
  title        = {A faster branching algorithm for the maximum $k$-defective clique problem},
  booktitle    = {Proceedings of the European Conference on Artificial Intelligence (ECAI)},
  pages        = {4132--4139},
  year         = {2024},
}

@inproceedings{pei2005mining,
  title     = {On mining cross-graph quasi-cliques},
  author    = {Pei, Jian and Jiang, Daxin and Zhang, Aidong},
  booktitle = {Proceedings of KDD},
  pages     = {228--238},
  year      = {2005},
}

@article{bourjolly2002exact,
  title={An exact algorithm for the maximum $k$-club problem in an undirected graph},
  author={Bourjolly, Jean-Marie and Laporte, Gilbert and Pesant, Gilles},
  journal={European Journal of Operational Research},
  volume={138},
  number={1},
  pages={21--28},
  year={2002},
}

@article{seidman1978kplex,
  title={A graph-theoretic generalization of the clique concept},
  author={Seidman, Stephen B. and Foster, Brian L},
  journal={Journal of Mathematical Sociology},
  volume={6},
  number={1},
  pages={139--154},
  year={1978},
  publisher={Taylor \& Francis}
}

@article{TBBB13,
  author       = {Svyatoslav Trukhanov and Chitra Balasubramaniam and Balabhaskar Balasundaram and Sergiy Butenko},
  title        = {Algorithms for detecting optimal hereditary structures in graphs, with application to clique relaxations},
  journal      = {Computational Optimization and Applications},
  volume       = {56},
  number       = {1},
  pages        = {113--130},
  year         = {2013},
}

@inproceedings{wang2025identifying,
  title={Identifying maximum defective bicliques in large bipartite graphs},
  author={Wang, Zhiyi and Chang, Lijun and Yu, Jeffrey Xu},
  booktitle={Proceedings of the IEEE International Conference on Data Engineering (ICDE)},
  pages={3710--3723},
  year={2025},
}

@inproceedings{wang2023fast,
  title     = {A fast maximum $k$-plex algorithm parameterized by the degeneracy gap},
  author    = {Wang, Zhengren and Zhou, Yi and Luo, Chunyu and Xiao, Mingyu},
  booktitle = {Proceedings of the International Joint Conference on Artificial Intelligence (IJCAI)},
  pages     = {5648-5656},
  year      = {2023}
}

@inproceedings{xiao2017fast,
  title={A fast algorithm to compute maximum $k$-plexes in social network analysis},
  author={Xiao, Mingyu and Lin, Weibo and Dai, Yuanshun and Zeng, Yifeng},
  booktitle={Proceedings of the AAAI Conference on Artificial Intelligence (AAAI)},
  pages = {919-925},
  year={2017}
}

@article{xiao2017clique,
author       = {Mingyu Xiao and Hiroshi Nagamochi},
title = {Exact algorithms for maximum independent set},
journal = {Information and Computation},
volume = {255},
pages = {126--146},
year = {2017},
}

@article{yu2006predicting,
  author  = {Haiyuan Yu and Alberto Paccanaro and Valery Trifonov and Mark Gerstein},
  title        = {Predicting interactions in protein networks by completing defective cliques},
  journal      = {Bioinformatics},
  volume       = {22},
  number       = {7},
  pages        = {823--829},
  year         = {2006},
}

@article{yu2023fast,
  title     = {Fast maximal quasi-clique enumeration: {A} pruning and branching co-design approach},
  author    = {Yu, Kaiqiang and Long, Cheng},
  journal   = {Proc. ACM SIGMOD Int. Conf. Manage. Data (SIGMOD)},
  volume    = {1},
  number    = {3},
  pages     = {1--26},
  year      = {2023},
  publisher = {ACM New York, NY, USA}
}

@inproceedings{zeng2006coherent,
  title     = {Coherent closed quasi-clique discovery from large dense graph databases},
  author    = {Zeng, Zhiping and Wang, Jianyong and Zhou, Lizhu and Karypis, George},
  booktitle = {Proc. of KDD (SIGKDD)},
  pages     = {797--802},
  year      = {2006}
}

@article{brelaz1979new,
  title={New methods to color the vertices of a graph},
  author={Br{\'e}laz, Daniel},
  journal={Communications of the ACM},
  volume={22},
  number={4},
  pages={251--256},
  year={1979},
  publisher={ACM New York, NY, USA}
}

@book{fomin2010exact,
  title={Exact Exponential Algorithms},
  author={Fomin, Fedor V and Kratsch, Dieter},
  year={2010},
  publisher={Springer Science \& Business Media}
}

@article{jang2025efficient,
  title={Efficient defective clique enumeration and search with worst-case optimal search space},
  author={Jang, Jihoon and Nam, Yehyun and Park, Kunsoo and Kim, Hyunjoon},
  journal={Proc. ACM SIGMOD Int. Conf. Manage. Data (SIGMOD)},
  volume={3},
  number={6},
  pages={1--28},
  year={2025},
}

@inproceedings{zhou2021improving,
  title={Improving maximum $k$-plex solver via second-order reduction and graph color bounding},
  author={Zhou, Yi and Hu, Shan and Xiao, Mingyu and Fu, Zhang-Hua},
  booktitle={Proceedings of the AAAI Conference on Artificial Intelligence (AAAI)},
  pages={12453--12460},
  year={2021}
}

@book{ahuja1993network,
  title={Network Flows: Theory, Algorithms, and Applications},
  author={Ahuja, Ravindra K. and Magnanti, Thomas L. and Orlin, James B.},
  year={1993},
  publisher={Prentice Hall},
  isbn={9780136175490}
}

@article{gao2024maximum,
  title={Maximum $k$-plex search: An alternated reduction-and-bound method},
  author={Gao, Shuohao and Yu, Kaiqiang and Liu, Shengxin and Long, Cheng},
  journal={Proceedings of the VLDB Endowment},
  volume={18},
  number={2},
  pages={363--376},
  year={2024},
  publisher={VLDB Endowment}
}
